\documentclass[format=acmsmall, review=false, 11pt]{acmart} 
\usepackage{acm-ec-26}
\usepackage{booktabs} 

\setcitestyle{acmnumeric}
\usepackage{cleveref}
\usepackage[most]{tcolorbox} 
\usepackage{threeparttable}
\usepackage{tabularx, booktabs} 
  
 \usepackage{pdfpages}
 \usepackage{mathtools}
\usetikzlibrary{positioning}

\usepackage{threeparttable}
\usepackage{tabularx, booktabs} 
\usepackage{multirow}
\usepackage{color, colortbl}
\definecolor{Gray}{gray}{0.9}

\makeatletter
\newcommand{\thickhline}{%
    \noalign {\ifnum 0=`}\fi \hrule height 1pt
    \futurelet \reserved@a \@xhline
}
\newcolumntype{"}{@{\hskip\tabcolsep\vrule width 1pt\hskip\tabcolsep}}
\makeatother

\newcommand{\tbx}{\ensuremath{\tilde{\mathbf{x}}}\xspace}
\newcommand{\ty}{\ensuremath{\tilde{y}}\xspace}
\newcommand{\tx}{\ensuremath{\tilde{x}}\xspace}

\newcommand{\bx}{\ensuremath{\mathbf{x}}\xspace}

\newcommand{\om}{\ensuremath{\mathsf{OM\mbox{-}SR}}\xspace}
\newcommand{\Exp}{\ensuremath{\mathrm{Exp}}\xspace} 
\newcommand{\sto}{\ensuremath{\mathsf{ST\mbox{-}BA}}\xspace}

\newcommand{\bfq}{\ensuremath{\mathbf{q}}\xspace}
\newcommand{\cQ}{\ensuremath{\mathbf{Q}}\xspace}

\newcommand{\tpsi}{\ensuremath{\tilde{\psi}}\xspace}

\newcommand{\Gam}{\ensuremath{\Gamma}\xspace}

\tcbset{
    colback=blue!3,
    colframe=blue!60!black,
    coltitle=white,
    colbacktitle=blue!70!black,
    fonttitle=\bfseries
}

\usepackage{macros_pan}
\title{Gain-Sharing Optimization in Randomized Primal--Dual Analysis:
Structure and Certification}

\author{Pan Xu}
\authornote{Email: \texttt{tcs.enthusiast@gmail.com}}
\affiliation{%
  \institution{\large Independent Researcher}
  \city{Your City}
  \country{USA}
}
\email{tcs.enthusiast@gmail.com}

\begin{abstract}
\normalsize
Randomized primal--dual guarantees for online matching and allocation
often depend on a gain-sharing function that distributes algorithmic
gains among dual variables and, in some settings, also influences
allocation decisions. An analysis based on a particular function does
not determine the best certificate obtainable within the same
construction. We study this question for a specified benchmark and an
\emph{algorithm--analysis template}, which fixes how the gain-sharing
function enters the allocation rule and the dual updates. For an
admissible function class $\mathcal F$, we formulate the objective as
\[
\Gamma(\mathcal F):=\sup_{f\in\mathcal F}\mathcal L[f],
\]
where $\mathcal L[f]$ is the worst-case expected dual-feasibility factor.
Our goal is to certify this functional optimum from both sides, rather
than establish only a lower guarantee for one choice of function.

We develop two complementary certification methods. The continuous
method uses a monotone Volterra comparison principle to exploit
integral constraints with nonnegative dependence on earlier function
values. Comparison identifies a least function satisfying these
constraints; establishing its admissibility then enables a global
characterization of optimal solutions. Stability estimates convert
independently verified numerical residuals into rigorous error bounds,
separating candidate generation from certification. The discrete
method uses auxiliary linear programs together with refinement and
coarsening arguments that explicitly bound the gap between finite-grid
optima and the corresponding functional values.

We apply these methods to the vertex-weighted Stochastic Balance
template for online matching with stochastic rewards, with unit offline
capacities, adversarial arrivals, and uniform vanishing success
probabilities. Huang and Zhang (SICOMP 2024) obtained a guarantee of
approximately $0.576$ using a particular gain-sharing function. For
the configuration-LP benchmark and the specified dual updates, we
instead optimize over
$\mathcal F_0=C_{\uparrow}([0,\infty),[0,1])$, the class of all
continuous non-decreasing functions from $[0,\infty)$ to $[0,1]$.
A canonical adversarial reduction makes the functional objective
explicit. Our comparison-based characterization and validated numerical
analysis establish attainment and give
\[
\Gamma(\mathcal F_0)\in[0.5802,\;0.5805].
\]
We also show that the more tractable subclass used for auxiliary-LP
certification attains the full-class optimum. The upper bound limits
the certificates obtainable within the specified template, not the
actual competitiveness of the associated algorithms or the power of
all randomized primal--dual analyses.
\end{abstract}

\begin{document}
\pagenumbering{roman}

\maketitle

\setcounter{tocdepth}{1}
\tableofcontents

\clearpage

\pagenumbering{arabic}



\section{Introduction}
\label{sec:intro}

Primal--dual methods are central to the design and analysis of online
matching and allocation algorithms. Randomized primal--dual analysis
recovers the optimal ratio $1-1/e$ for
Ranking~\cite{devanur2013randomized}, while primal--dual methods also
explain the optimal small-bid guarantee for
Adwords~\cite{buchbinder2007online,mehta2007adwords}.
These methods and their extensions provide a common language for
reasoning about online decisions, randomness, and competitive guarantees;
see~\cite{buchbinder2009design,mehta2012online,huang2024online}.

An important choice in many such analyses is a \emph{gain-sharing
function}, which determines how algorithmic gains are credited to dual
variables. A particular function, or a restricted family chosen for
tractability, can yield a competitive guarantee. But a successful
construction does not by itself determine how much the guarantee could
improve through a different function. Answering that question requires
an upper bound on the best certificate obtainable within the specified
analysis, not only a lower bound supplied by one construction.

\begin{tcolorbox}[
colback=gray!5,
colframe=black!40,
boxrule=0.4pt,
sharp corners,
left=6pt,
right=6pt,
top=5pt,
bottom=5pt
]
\centering
\emph{For a specified primal--dual template, what is the best guarantee
certifiable over an admissible class of gain-sharing functions?}
\end{tcolorbox}

\xhdr{The two roles of gain-sharing functions.}
A gain-sharing function $f$ can play two distinct roles in randomized
primal--dual analysis. In some settings, it is \emph{purely analytical}:
it specifies how algorithmic gains are distributed among dual variables
without affecting the underlying algorithm. A classical example is the
randomized primal--dual analysis of unweighted
Ranking~\cite{devanur2013randomized}. The algorithm's decisions follow
a random ordering of the offline vertices and do not depend on $f$;
varying $f$ changes only the dual construction and the resulting
certificate.

In other settings, $f$ has both an \emph{algorithmic} and an
\emph{analytical} role: it enters the allocation rule as well as the
dual updates. This is the case for the vertex-weighted Stochastic
Balance template studied here, where $f$ determines load-dependent
allocation priorities. Changing $f$ can therefore change the
algorithm's decisions, not merely its analysis.

To encompass both settings, we fix a benchmark LP and a prescription
for using $f$ in the dual updates and, where applicable, the allocation
decisions. We call this specification an
\emph{algorithm--analysis template}. Optimizing over $f$ then means
optimizing the certificate of a fixed algorithm in the first setting,
or optimizing certificates within a function-parameterized family of
algorithms in the second.

For a fixed input instance $\mathcal I$, let $\balp[f]$ and
$\bbeta[f]$ be the nonnegative random dual vectors generated by this
prescription; their dependence on $\mathcal I$ is implicit. Write the
nonvacuous dual constraints as
$C_{\mathcal I,\lam}(\balp,\bbeta)\ge c_{\mathcal I,\lam}$,
where $\lam\in\Lam(\mathcal I)$, each $C_{\mathcal I,\lam}$ is
linear, and $c_{\mathcal I,\lam}>0$.
When the expected dual objective equals the algorithm's expected
reward, a common expected dual-feasibility factor $\gam>0$ certifies
a $\gam$-competitive guarantee: dividing the expected dual variables
by $\gam$ makes them feasible, and weak duality applies.
At fixed model parameters, the largest such uniform factor has the form
\begin{align}\label{def:cL-f}
\cL[f]
:=\inf_{\mathcal I}\inf_{\lam\in\Lam(\mathcal I)}
\frac{\E_{\mathcal I}\!\left[
C_{\mathcal I,\lam}(\balp[f],\bbeta[f])\right]}
{c_{\mathcal I,\lam}}.
\end{align}
The infimum includes the input instance, not only the constraints on
one instance. In asymptotic regimes, we use the corresponding limiting
worst-case factor. 


For an admissible function class $\cF$, our objective is
\begin{align}\label{def:out-max}
\Gam(\cF):=\sup_{f\in\cF}\cL[f].
\end{align}
A lower bound for a particular $\cL[f]$ supplies a lower certificate;
a uniform upper bound on $\cL[f]$ for every $f\in\cF$ supplies an
upper certificate. Such an upper bound is a barrier for the specified
benchmark, dual construction, and use of the gain-sharing function.
It is not an upper bound on actual algorithmic performance or on the
power of all randomized primal--dual analyses.

The distinction is especially important when restrictions on $f$ are
introduced to simplify the proof. A restricted class may be easier to
optimize but have a smaller optimum. Alternatively, the restrictions
may be lossless, even if this is not evident from their definition.
Certifying the functional optimum can distinguish these possibilities.
We develop two complementary approaches: global comparison of continuous
profiles and finite-dimensional auxiliary LPs with explicit error bounds.

\xhdr{Global comparison and validated certificates.}
Our continuous approach exploits an order structure in integral
constraints. After a suitable change of variables, the constraints
bound the unknown function at each point $t$ from below using
nonnegative weighted integrals of its values on $[0,t]$. These bounds
define an \emph{order-preserving Volterra operator}: its output at
$t$ depends only on input values up to $t$, and increasing the input
pointwise cannot decrease the output. A fixed-point construction and
comparison principle identify a function lying pointwise below every
function satisfying these inequalities.

The remaining constraint favors smaller transformed functions:
it bounds the certificate value by an integral expression that cannot
decrease when the function decreases. Once the constructed function
is shown to satisfy the required shape and boundary conditions,
feasibility therefore reduces to checking this final scalar
constraint. We use this comparison to establish attainment and
characterize optimal solutions of the functional program studied
here. The differential system used for computation follows from
this characterization, rather than from an assumed form of the
optimizer. Section~\ref{sec:overview} explains how we combine this
characterization with independent numerical verification to obtain
rigorous certificates.

\xhdr{Auxiliary LPs with finite-grid error bounds.}
Our discrete approach approximates the functional constraints by
finite-dimensional LPs. The essential step is quantitative error control,
not the discretization alone. Refinement and coarsening transfer feasible
solutions between grid resolutions with an explicit objective loss.
Combined with a consistency argument identifying the limiting
optimization problem, these transfers give two-sided bounds on the gap
between the finite LP optimum and the functional optimum. A rigorous
bound on a finite LP value therefore yields a certificate for the
continuous problem. When the discretized constraints represent only
an upper envelope of the original objective, the same reasoning gives
an upper certificate rather than an exact characterization.

\xhdr{Stochastic Balance: the optimum and a lossless restriction.}
We give a complete application to the vertex-weighted Stochastic Balance
template for online matching with stochastic rewards, under adversarial
arrivals and identical vanishing success probabilities. The benchmark is the configuration LP, and the dual construction is
the gain-sharing rule specified in Appendix~\ref{app:thm-structure}. Huang and
Zhang~\cite{huang2020online} obtained a guarantee of approximately
$0.576$ using a particular gain-sharing function. We instead optimize
the uniform dual-feasibility factor over the full class
$\cF_0=C_{\uparrow}([0,\infty),[0,1])$, consisting of all
continuous non-decreasing functions mapping $[0,\infty)$ to $[0,1]$.

The application-specific starting point is an exact adversarial
reduction: Theorem~\ref{thm:structure} expresses $\cL[f]$ as a
three-parameter load-then-divert minimization for every admissible $f$.
The certification methods then operate on the resulting functional
problem. Our main theorem proves that the supremum is attained and
that
\[
\Gam(\cF_0)\in[0.5802,\;0.5805].
\]
The upper endpoint limits the certificates obtainable
by optimizing $f$ within this template; the lower
endpoint supplies a certified guarantee.\footnote{
Some numerically larger guarantees for Stochastic Balance
are established against a different stochastic offline
benchmark. These results are not directly comparable with,
and do not contradict, our upper bound on the certifying
capacity of the specified configuration-LP template;
see Subsection~\ref{sec:related}.} The comparison theorem also shows that the restrictions
defining the subclass used by the auxiliary-LP method are lossless:
a full-class optimizer satisfies the additional threshold, envelope,
and constant-tail conditions. Thus the analysis establishes both the
value of the full problem and why a more tractable subclass can attain
it. Appendix~\ref{app:thm-structure} gives the model and certificate
definition, and Section~3 states the results for the function-space
hierarchy.

\xhdr{Scope and connections.}
Optimization over proof parameters and finite certificates for limiting
LPs have precedents in strongly factor-revealing LPs and upper-bound
factor-revealing programs~\cite{mahdian2011online,fernandes2015systematic}.
Functional optimization also arises in the primal--dual analysis of
Ranking with random arrivals~\cite{peng2025revisiting}. Our emphasis
is on connecting global extremal structure with quantitative
certification: identifying an optimal critical profile through comparison
and validating approximate profiles independently through their
residuals, alongside explicit refinement/coarsening bounds for auxiliary
LPs.

The methods have concrete structural requirements. The continuous
approach requires an order-preserving causal formulation, integral
stability, and a proof that the constructed profile is admissible.
The discrete approach requires feasibility transfers between resolutions
and identification of the limiting problem. We establish these
requirements for the stochastic-matching program; applying either
method elsewhere requires verifying the corresponding properties of
the target problem. These criteria distinguish the reusable
certification mechanisms from the application-specific adversarial
reduction. Section~\ref{sec:overview} explains the technical challenges,
these applicability requirements, and the relation to existing approaches.

\section{Functional Setup and Comparison Preliminaries}
\label{sec:prelim}

We fix the functional notation and review the comparison tools used in
our continuous analysis.

\xhdr{Function spaces and notation.}
For an interval $I\subseteq\mathbb R$, let $C(I)$ denote the space of
continuous real-valued functions on $I$. For intervals $I,J\subseteq
\mathbb R$, write $C_{\uparrow}(I,J)$ for the class of continuous
non-decreasing functions mapping $I$ to $J$. In particular,
\[
\cF_0:=C_{\uparrow}([0,\infty),[0,1]).
\]
No differentiability is imposed on $\cF_0$. We write $C^1([a,b])$
for continuously differentiable functions, with one-sided derivatives
at the endpoints, and $\mathrm{AC}([a,b])$ for absolutely continuous
functions. The notations $\dot f$ and $f'$ denote the first derivative wherever
it exists; almost-everywhere statements refer to Lebesgue measure.

Inequalities between functions on the same domain are pointwise:
$u\le v$ means $u(t)\le v(t)$ for every $t$ in that domain. On a
compact interval $[0,\eta]$, we use the uniform norm
\[
\|u\|_{\infty,[0,\eta]}:=\max_{0\le t\le\eta}|u(t)|,
\]
and abbreviate it by $\|u\|_\infty$ when the interval is clear.
The space $C([0,\eta])$ is complete under this norm, hence is a
Banach space; see~\cite[Section~2.1]{teschl2012ordinary}.

\xhdr{Functional certificates.}
A \emph{functional} assigns a scalar value to a function. Here
$\cL[f]$ denotes the certificate associated with $f$, as introduced
in~\eqref{def:cL-f}; Theorem~\ref{thm:structure} supplies its exact
representation for our application. For a nonempty admissible class
$\cF\subseteq\cF_0$, we study
\begin{align}\label{eqn:Lf}
\Gam(\cF):=\sup_{f\in\cF}\cL[f].
\end{align}
A lower certificate follows from an admissible function whose value
is bounded below. An upper certificate requires a bound on $\cL[f]$
valid for every $f\in\cF$. We use a supremum without presupposing
attainment. Class inclusion gives
$\Gam(\cF)\le\Gam(\widetilde{\cF})$ whenever
$\cF\subseteq\widetilde{\cF}\subseteq\cF_0$.

\xhdr{Causal comparison and residual stability.}
Here an \emph{operator} is a mapping from one function space to
another, so its output is itself a function. Fix $0<\eta<1$ and
consider a possibly nonlinear operator
\[
\mathcal T:C([0,\eta])\longrightarrow C([0,\eta]),
\qquad u\longmapsto\mathcal T u.
\]
Thus $\mathcal T u$ is a continuous function, and $(\mathcal T u)(t)$
is its value at $t$. A \emph{fixed point} is a function $y$ satisfying
$\mathcal T y=y$.

The operator is \emph{order-preserving} if $u\le v$ implies
$\mathcal T u\le\mathcal T v$. It is \emph{causal} if agreement of
$u$ and $v$ on $[0,t]$ implies $(\mathcal T u)(t)=(\mathcal T v)(t)$:
its output at $t$ does not depend on input values after $t$.
Volterra integral operators have this earlier-value dependence;
see~\cite{gripenberg1990volterra} for background on these equations
and their comparison and stability theory.

For the consequences below, suppose that $\mathcal T$ is
order-preserving and that, for some constant $L\ge0$,
\begin{align}\label{eq:prelim-volterra-lipschitz}
|(\mathcal T u)(t)-(\mathcal T v)(t)|
\le L\int_0^t|u(s)-v(s)|\,\sd s,
\qquad 0\le t\le\eta,
\end{align}
for all $u,v\in C([0,\eta])$. This estimate in particular implies
causality. We record the resulting fixed-point, comparison, and
error bounds in the form needed later.

\emph{Existence and uniqueness.}
For any $k>L$, define the weighted norm
\[
\|u\|_k:=\max_{0\le t\le\eta}e^{-kt}|u(t)|.
\]
The inequalities $e^{-k\eta}\|u\|_\infty\le\|u\|_k\le\|u\|_\infty$
show that this norm is equivalent to the uniform norm, so
$C([0,\eta])$ remains complete in this weighted norm.
Equation~\eqref{eq:prelim-volterra-lipschitz} gives
\[
\|\mathcal T u-\mathcal T v\|_k
\le\frac{L}{k}\|u-v\|_k.
\]
Since $L/k<1$, Banach's fixed-point theorem
\cite[Theorem~2.1]{teschl2012ordinary} gives a unique fixed point
$y\in C([0,\eta])$. Starting from any $u_0\in C([0,\eta])$, the
iterates $u_{n+1}:=\mathcal T u_n$ converge uniformly to $y$.
This weighted-norm argument is an instance of Bielecki's renorming
technique~\cite{bessenyei2021applications}; it does not require
$\mathcal T$ to be a contraction in the original uniform norm.

\emph{Comparison.}
A function $u$ satisfying $u\ge\mathcal T u$ is called a
\emph{supersolution}. Starting the iteration at $u_0=u$, order
preservation gives $u_0\ge u_1\ge u_2\ge\cdots$. Passing to its
uniform limit yields
\[
u\ge\mathcal T u\quad\Longrightarrow\quad u\ge y.
\]
Thus the fixed point lies below every supersolution. This comparison
is the mechanism used later to bound an exact profile by feasible
profiles.

\emph{Residual stability.}
For an explicit candidate $u$, the function $\mathcal T u-u$ is its
\emph{residual}, measuring its failure to satisfy the fixed-point
equation. If $\|\mathcal T u-u\|_\infty\le\varepsilon$, where
$\varepsilon\ge0$, then~\eqref{eq:prelim-volterra-lipschitz} and
$\mathcal T y=y$ give
\[
|y(t)-u(t)|\le\varepsilon
 +L\int_0^t|y(s)-u(s)|\,\sd s.
\]
The integral form of Gronwall's inequality
\cite[Lemma~2.7]{teschl2012ordinary} implies
$|y(t)-u(t)|\le\varepsilon e^{Lt}$, and consequently
\begin{align}\label{eq:prelim-residual-bound}
\|y-u\|_{\infty,[0,\eta]}\le\varepsilon e^{L\eta}.
\end{align}
This estimate converts a verified residual into a rigorous error bound
for the entire candidate profile.

For the parameterized operator $\mathcal T_\alp$ in
Appendix~\ref{app:gam-f}, the parameter $\alp$ and interval endpoint
$\eta$ are fixed when applying these statements. Its self-map and
integral estimates are proved there, and the residual bound is used
in Appendix~\ref{app:error}. The ambient space is the full
$C([0,\eta])$: operator order does not imply that an individual output
profile is monotone in $t$, or that the fixed point satisfies the
additional shape and boundary constraints of an optimization problem.
Those properties are established separately in the application.

\section{Main Contributions}

We certify functional optima arising from parameterized randomized
primal--dual analyses. Our results combine an exact representation of
the certificate functional, a structural characterization of optimal
profiles, and finite computations with explicit error control.

\xhdr{Application and exact functional representation.}
We apply our methods to the Stochastic Balance (\sto) template for
vertex-weighted online matching with stochastic rewards (\om), using
the configuration-LP benchmark~\cite{huang2020online}. The gain-sharing function parameterizes
both the allocation rule and the dual updates, and $\cL[f]$ is their
worst-case expected dual-feasibility factor. Specifically, with a common success probability, an arrival is assigned to an
available neighbor $u$ maximizing $w_u(1-f(\ell_u))$, where
$\ell_u$ is its accumulated load before the assignment.
The fractions $f(\ell_u)$ and $1-f(\ell_u)$ of the assignment's
expected reward are credited to the offline and online dual
variables, respectively.

The following theorem
reduces the implicit optimization over input instances and dual
constraints to an infimum over three scalar parameters.

\begin{theorem}[Canonical representation for Stochastic Balance]
\label{thm:structure}
Consider the \sto\ template for vertex-weighted \om\ with unit
offline capacities, adversarial arrivals, and uniform vanishing success
probabilities. For every $f\in\cF_0$, the certificate
functional in~\eqref{def:cL-f}, specialized to this limiting regime,
satisfies
\begin{align}\label{eqn:Lf-d-1}
\cL[f]
=\inf_{\substack{\ell\ge0,\;0<\psi\le1\\0\le\tpsi\le\psi}}
\frac{
\displaystyle\int_0^\ell e^{-z}f(z)\,\sd z
+\displaystyle\int_\ell^{\ell+\tpsi}e^{-z}\,\sd z
+(\psi-\tpsi)\bigl(1-f(\ell+\tpsi)\bigr)
}{\psi}.
\end{align}
\end{theorem}

Here $\ell$ is the initial pre-load on a target offline agent,
$\psi$ is the total load of the selected arrivals, and $\tpsi$ is
the portion used to add load to the target before the remaining
arrivals are diverted to competing agents. The three numerator terms
are the contributions of the initial pre-load, the load-adding block,
and the diversion block, respectively. The infimum allows limiting
adversarial configurations; attainment is not asserted for arbitrary $f$.

The model, dual updates, precise limiting convention, and proof are
given in Appendix~\ref{app:thm-structure}. We take
\eqref{eqn:Lf-d-1} as the starting point for certifying
$\Gam(\cF)=\sup_{f\in\cF}\cL[f]$ over admissible function classes.

\xhdr{Two complementary certification methods.}
Our continuous method uses the order-preserving Volterra structure
introduced in Section~\ref{sec:prelim} to compare profiles satisfying
the partial-load constraints. We construct a least profile, establish
its admissibility and boundary behavior, and prove attainment and
uniqueness of the optimal critical trajectory. Flatness of its
minimized adversarial tests yields the system used for numerical
computation. An independent residual-based verification then converts
explicit numerical candidates into rigorous bounds on the optimum,
without assuming the accuracy of the candidate-generation algorithm.

Our discrete method uses auxiliary LPs with explicit finite-grid error
bounds. Refinement and coarsening transfer feasible solutions between
resolutions with controlled objective loss. Together with a consistency
argument identifying the limiting functional problem, these transfers
convert finite LP bounds into certificates for the original optimum or
an upper envelope. Section~\ref{sec:overview} explains the technical
challenges, the structural requirements for the two methods, and their
connections to factor-revealing programs.

\xhdr{Function classes.}
We compare the nested classes
$\cF_3\subseteq\cF_2\subseteq\cF_1\subseteq\cF_0$.
Recall from Section~\ref{sec:prelim} that
\begin{align}
\cF_0
&=C_{\uparrow}([0,\infty),[0,1]).
\label{eq:def-F0}
\end{align}
Define
\begin{align}
\cF_1
&:=\{f\in\cF_0:f(1)\le1-1/e\},
\label{eq:def-F1}\\
\cF_2
&:=\left\{f\in\cF_1:
\begin{array}{l}
1-e^{-z}\le f(z)\le1-1/e\quad(0\le z\le1),\\
f(z)=1-1/e\quad(z\ge1)
\end{array}\right\},
\label{eq:def-F2}\\
\cF_3
&:=\left\{f\in\cF_2:
\begin{array}{l}
f|_{[0,1]}\in\mathrm{AC}([0,1]),\\
0\le\dot f(z)\le e^{-z}-(1-f(z))
\quad\text{for a.e. }z\in[0,1]
\end{array}\right\}.
\label{eq:def-F3}
\end{align}
Here $\mathrm{AC}([0,1])$ denotes the absolutely continuous functions,
and $\dot f$ is their almost-everywhere derivative. This additional
regularity condition is imposed only on $\cF_3$.\footnote{Absolute
continuity ensures that $f(b)-f(a)=\int_a^b\dot f(z)\,\sd z$ for
all $0\le a\le b\le1$, allowing us to integrate the derivative
constraint. Monotonicity alone guarantees almost-everywhere
differentiability, but not this identity.}
The class $\cF_1$ caps the value at load one, $\cF_2$ adds a lower
envelope and a constant tail, and $\cF_3$ imposes a differential
inequality. These restrictions simplify the functional problem;
the results below determine whether they reduce its optimal value.

\begin{theorem}[Auxiliary-LP and closed-form certificates]
\label{thm:aux-lp-cert}
For the classes in~\eqref{eq:def-F0}--\eqref{eq:def-F3}, the
auxiliary-LP analysis gives
\[
\Gam(\cF_2)\in[0.5796,0.5809],
\qquad
\Gam(\cF_1)\le0.5809,
\qquad
\Gam(\cF_0)\le0.5842.
\]
These bounds incorporate the explicit finite-grid error estimates.
The most restricted class has the closed-form value
\[
\Gam(\cF_3)=\frac{1+e^{-2}}2.
\]
\end{theorem}

Section~\ref{sec:lb} proves the $\cF_2$ certificate; the remaining
bounds are proved in the appendices. The closed-form value matches
the classical stochastic-matching guarantee of Mehta and
Panigrahi~\cite{mehta2013online}. The continuous analysis sharpens
the enclosure and proves that the restrictions defining $\cF_1$ and
$\cF_2$ are lossless at the optimum.

\begin{theorem}[Full-class certification and lossless restrictions]
\label{thm:main}
The supremum defining $\Gam(\cF_0)$ is attained by a function in
$\cF_2$. The three classes have the same optimal value, with the
certified enclosure
\[
\Gam(\cF_0)=\Gam(\cF_1)=\Gam(\cF_2)
\in[0.5802,\;0.5805].
\]
\end{theorem}

\xhdr{Structure of an optimal certificate.}
The comparison proof in Section~\ref{sec:thm2} also characterizes
the optimizers of the exact functional program.
Theorem~\ref{thm:main-structural} shows that they share the same
critical interval, ending at their first crossing of $1-1/e$, and
coincide there. This common trajectory is strictly increasing and
$C^1$. After minimizing over the initial load, its partial-load test
has the same value at every total load in the critical interval.
This flatness property yields the functional--differential
characterization in Corollary~\ref{cor:main-struc}; it does not assert
that every individual adversarial test is tight.

The equality of the three optima is structural, not inferred from
overlapping numerical intervals. Appendix~\ref{app:thm-exact-opt}
shows that continuing the optimal critical trajectory constantly at
$1-1/e$ gives a full-class maximizer in $\cF_2$. The class inclusions
then imply equality of the values for $\cF_0$, $\cF_1$, and $\cF_2$.
In contrast, the additional restrictions defining $\cF_3$ strictly
reduce the optimal value.

\xhdr{Interpretation of the upper certificate.}
The bound $0.5805$ limits the certificate obtainable from the specified
configuration-LP benchmark and dual updates, optimized over
$f\in\cF_0$. It does not bound the actual competitiveness of the
associated algorithms or the power of all randomized primal--dual
analyses. A stronger guarantee could follow from a different analysis
or benchmark; the result excludes obtaining it solely by optimizing
$f$ within this template. Table~\ref{table:con} separates the bounds
supplied by each method from the final enclosure of the common optimum.

\begin{table}[ht!]
\centering
\small
\renewcommand{\arraystretch}{1.2}
\setlength{\tabcolsep}{5pt}
\caption{Certificates for the function classes in
\eqref{eq:def-F0}--\eqref{eq:def-F3}. The auxiliary-LP rows report
bounds obtained by that method alone. The final row gives the sharper
enclosure of the common optimum for $\cF_0$, $\cF_1$, and $\cF_2$;
different auxiliary bounds do not indicate different optimal values.}
\label{table:con}
\begin{tabularx}{\linewidth}{@{}
>{\centering\arraybackslash}p{0.19\linewidth}
>{\raggedright\arraybackslash}X
>{\raggedright\arraybackslash}p{0.25\linewidth}
@{}}
\toprule
\textbf{Class} & \textbf{Certificate or method} & \textbf{Certified bound}\\
\midrule
$\{\tilde f\}$
& Prior fixed-function guarantee;
  see~\eqref{eqn:1-25-a}
& $\ge0.576$~\cite{huang2020online}\\
\addlinespace
$\cF_3$
& Differentially restricted class; closed-form optimum
& $\displaystyle\frac{1+e^{-2}}2\approx0.56767$\\
\addlinespace
$\cF_2$
& Auxiliary-LP two-sided certificate
& $[0.5796,\;0.5809]$\\
\addlinespace
$\cF_1$
& Auxiliary-LP upper certificate
& $\le0.5809$\\
\addlinespace
$\cF_0$
& Auxiliary-LP upper certificate
& $\le0.5842$\\
\midrule
$\cF_0,\ \cF_1,\ \cF_2$
& Common optimum; comparison and validated residuals
& $[0.5802,\;0.5805]$\\
\bottomrule
\end{tabularx}
\end{table}

\section{Overview of Challenges and Techniques}
\label{sec:overview}

Our objective is to certify the value of the functional optimization
problem
\[
\Gam(\cF)=\sup_{f\in\cF}\cL[f].
\]
There are two distinct tasks: make the adversarial minimization defining
$\cL[f]$ explicit, and then optimize over the admissible functions.
Theorem~\ref{thm:structure} resolves the first task for the Stochastic
Balance template by reducing the adversary to a three-parameter
load-then-divert form. We take this exact representation as the starting
point here and focus on the two certification methods for the outer
optimization: refinement and coarsening of auxiliary LPs, and global
comparison of continuous profiles. 

ChatGPT (OpenAI) assisted in developing the comparison proof,
the shooting and endpoint-certification procedures, and parts of
the auxiliary-LP analysis, as detailed in the AI methodology statement.

\subsection{Auxiliary-LP Certification by Refinement and Coarsening}

A finite grid does not automatically certify an optimization problem
over functions. Constraints must hold between grid points, and a
sequence of numerical optima does not provide a bound on its limit
without quantitative error control.

For $\cF_2$, Proposition~\ref{pro:2-17-a} reduces $\cL[f]$ to the
minimum of two functionals, $W_1$ and $W_2$. The auxiliary program
\textbf{AUG-LP} records values of $f$ on a uniform grid and replaces
the integrals by finite sums. The main analytical step is to transfer
feasible solutions between resolutions in both directions. Refinement
from $n$ to $2n$ and coarsening from $2n$ to $n$ each lose at most
$(1-1/e)/(2n)$ in objective value. Thus, writing $\eta(n)$ for the
optimal value of \textbf{AUG-LP},
\[
|\eta(2n)-\eta(n)|\le\frac{1-1/e}{2n}.
\]
A separate consistency argument identifies the limiting LP value with
$\Gam(\cF_2)$. Telescoping the transfer bounds along the grids
$n,2n,4n,\ldots$ then gives
\[
|\eta(n)-\Gam(\cF_2)|\le\frac{1-1/e}{n}.
\]
Consequently, a rigorous enclosure of the finite LP optimum yields an
enclosure of the functional optimum. The feasibility transfers and the
identification of the limit are both needed; numerical convergence
alone does not establish either conclusion.

For larger classes, the same strategy can certify an upper envelope
rather than the original functional itself. The tests $W_1$ and $W_2$
upper bound $\cL[f]$ on $\cF_0$ and $\cF_1$, although their minimum
need not equal $\cL[f]$ for each function in these classes. Discretizing
these envelopes and controlling the finite-grid error gives the
auxiliary bounds $\Gam(\cF_0)\le0.5842$ and
$\Gam(\cF_1)\le0.5809$ in Table~\ref{table:con}. The continuous
method below determines the full-class optimum more tightly and
establishes the structure of its optimizers.

\subsection{Global Optimality by Monotone Volterra Comparison}

To complement the auxiliary-LP bounds, we turn to a continuous
analysis that seeks to characterize optimal gain-sharing functions
over the full admissible class. We begin by retaining two families
of canonical adversarial tests: a scalar test associated with the
threshold $1-1/e$, and partial-load tests over the corresponding
critical interval. These tests define
Program~\eqref{prog:var-gamma-F0}, an upper-bounding relaxation of
the original functional optimization problem. Establishing exactness
requires more than finding a function satisfying the selected tests:
we must characterize an optimizer of the relaxation and show that
its value is also attainable in the original problem.

\xhdr{Why direct stationarity is insufficient.}
The admissible functions are continuous and non-decreasing, but their
smoothness and the existence of an optimizer are not given. Moreover,
the interval of partial-load constraints ends at the first time the
unknown function reaches $1-1/e$. A direct stationarity argument must
therefore account for the restrictions on perturbations imposed by
monotonicity and for this function-dependent endpoint.

Even after a fixed-domain reformulation makes the program convex, a
KKT-type argument would need more than complementary slackness to
establish the desired conclusion. Complementary slackness identifies
binding constraints where a multiplier is supported; it does not by
itself show that the minimized partial-load value is constant at every
total load in the critical interval. Nor does it supply the required
regularity or threshold-crossing behavior. We instead exploit an order
structure in the constraints that permits a global comparison with
every feasible profile.

\xhdr{From constraints to a causal integral operator.}
The normalization in Appendix~\ref{app:gam-f} allows us, for certificate
levels above $1/e$, to replace $f$ by its cap at $1-1/e$ and work with
\[
g(z):=1-\min\{f(z),1-1/e\},\qquad 0\le z\le1.
\]
The transformed profile is continuous and non-increasing, with
$1/e\le g\le1$. For a trial certificate level $\alp\in(1/e,1)$,
the scalar constraint becomes
\[
\mathcal A(g):=1-\frac1{e^2}-\int_0^1e^{-z}g(z)\,\sd z
\ge\alp.
\]
The partial-load constraints, after rearrangement, require
\[
g(t)\ge
\frac{\alp-1+e^{-t}+\int_0^s e^{-z}g(z)\,\sd z}{1-t+s},
\qquad 0\le s\le t<1.
\]
A level is feasible when an admissible profile satisfies all the
constraints at that level; the construction below does not assume
feasibility of every trial value.

The last inequality exposes the useful structure: the lower bound at
load $t$ depends on the profile through integrals over earlier loads,
and this dependence has a positive kernel. Taking the maximum over
$s\in[0,t]$, together with the lower bound $1/e$, defines an
order-preserving causal operator $\mathcal T_\alp$. Thus the partial
constraints for $t<1$ are precisely $g\ge\mathcal T_\alp g$.
On each $[0,\eta]$ with $0<\eta<1$, the integral estimate from
Section~\ref{sec:prelim} gives a unique fixed point $y_\alp$ and the
comparison implication
\[
g\ge\mathcal T_\alp g\quad\Longrightarrow\quad g\ge y_\alp.
\]
The fixed point is constructed for every $\alp\in(1/e,1)$ on
$[0,1)$. Tests at $t=1$, where the rearranged denominator can vanish,
are treated separately. This is a comparison between entire profiles,
not a first-order condition at an assumed optimizer.

\xhdr{Admissibility and the critical interval.}
Order preservation does not by itself imply that $y_\alp$ decreases
with the load or has the required boundary behavior. Establishing these
properties is a central part of the proof. The fixed-point equation
first yields local Lipschitz regularity. While
$1/e<y_\alp(t)<e^{-t}$, the profile is strictly decreasing and its
minimum partial-load test equals $\alp$. The minimizing initial load
is then unique and continuous, which upgrades the critical branch to
$C^1$ regularity.

For a feasible level $\alp>1/e$, comparison with any feasible profile
and the scalar constraint exclude both a premature contact with
$e^{-t}$ and a branch that persists to load $1$. The branch must instead
reach $1/e$ at some $\tau_\alp<1$. Continuing it constantly by $1/e$
preserves the partial constraints, including those at load $1$.
Denote this continuation by $g_\alp$. It is admissible and lies below
every feasible profile at level $\alp$.

\xhdr{Saturation, uniqueness, and flatness.}
The scalar functional $\mathcal A$ decreases when its profile increases.
Consequently, for every feasible pair $(\alp,g)$,
\[
g_\alp\le g,
\qquad
\mathcal A(g_\alp)\ge\mathcal A(g)\ge\alp.
\]
Thus the least profile is the most favorable one for the scalar
constraint. Parameter continuity and control of the constant tails
also establish attainment of the supremal value $\gam^*$.

At that value, $\mathcal A(g_{\gam^*})$ must equal $\gam^*$.
Otherwise, continuity would allow a slightly larger level whose
constructed profile still satisfies the scalar constraint, contradicting
optimality. Every normalized optimizer $g^*$ dominates $g_{\gam^*}$;
tightness of the scalar constraint then gives
\[
\int_0^1e^{-z}\bigl(g^*(z)-g_{\gam^*}(z)\bigr)\,\sd z=0.
\]
The integrand is continuous and nonnegative, so the two profiles
coincide. Returning to $f$, this proves uniqueness on the critical
interval and the flatness conclusion of
Theorem~\ref{thm:main-structural}: after minimizing over the initial
load, every total load in that interval yields the same value
$\gam^*$. It does not assert that every individual partial-load test
is tight.

This structure also closes the relaxation gap. The constant
continuation of the optimal critical function belongs to $\cF_2$,
where the full canonical objective can be evaluated exactly. It attains
$\gam^*$, proving Theorem~\ref{thm:main-exact}. In particular, together
with the inclusions of the function classes, this gives
\[
\Gam(\cF_0)=\Gam(\cF_1)=\Gam(\cF_2).
\]
Thus the restrictions defining $\cF_2$ are lossless at the optimum,
even though they need not preserve the certificate of each individual
function in $\cF_0$.

\subsection{Numerical Candidates and Independent Verification}

Flatness and regularity yield the functional--differential system in
Corollary~\ref{cor:main-struc}. Shooting from an initial value
$f(0)=1-\alp$ produces a candidate trajectory, whose returned scalar
value must equal $\alp$ at optimality. Euler integration is useful for
finding approximate profiles, but does not itself certify the optimum.

We separate candidate generation from verification. Each candidate
$u$ is an explicitly specified piecewise-linear approximation to the
transformed profile. The verifier checks its shape and partial-load
tests using exact integer arithmetic and outward-rounded intervals.
These finite checks yield a uniform residual bound for the exact
Volterra equation. On $[0,\eta]$ with $\eta \in (0,1)$, the corresponding stability estimate
is
\[
\|\mathcal T_\alp u-u\|_{\infty,[0,\eta]}\le\varepsilon
\quad\Longrightarrow\quad
\|y_\alp-u\|_{\infty,[0,\eta]}
\le\varepsilon e^{\eta/(1-\eta)}.
\]
Additional verified shape and terminal-margin conditions establish that
the exact trajectory reaches $1/e$ before $\eta$. Integrating the
common constant tail then bounds the returned value on a fixed interval,
without a separate stopping-time error estimate or sampled bounds on
derivatives of the unknown exact solution.

At the endpoints $a=0.5802$ and $b=0.5805$, the resulting enclosures
prove $q(a)>a$ and $q(b)<b$, where $q$ is the return value defined in
Section~\ref{sec:numerical-gamma-star}. The first inequality supplies
a feasible certificate. The second, together with profile comparison,
excludes feasibility at or above $b$. This yields
\[
\Gam(\cF_0)\in[0.5802,\;0.5805].
\]
The verification does not assume that either endpoint is feasible in
advance, and it does not use the auxiliary-LP numerical enclosure.
Appendix~\ref{app:error} gives the finite checks, error analysis, and
reproducible certificate files.

\subsection{Scope and Connections to Factor-Revealing Programs}

The two methods address a difficulty that also occurs in factor-revealing
programs: a finite optimization may suggest the limiting guarantee
without bounding it in the direction needed for the analysis.
Strongly factor-revealing LPs were developed by Mahdian and
Yan~\cite{mahdian2011online}, and Fernandes et
al.~\cite{fernandes2015systematic} developed upper-bound
factor-revealing programs that provide finite certificates for
approximation guarantees. Our auxiliary-LP argument is related to
this line of work. Its reusable step is an explicit transfer of
feasibility between resolutions, combined with an identification of
the limiting functional problem. Applying it to another LP family
requires suitable refinement and coarsening maps with controlled
objective loss.

The continuous method exploits a different structure. Suppose a
functional program has constraints $g\ge\mathcal T_\alp g$ and
$\mathcal A(g)\ge\alp$, where comparison gives a least fixed point
$y_\alp$ and $\mathcal A$ is order-reversing. Whenever $y_\alp$
belongs to the admissible class, feasibility reduces to the scalar test
\[
\alp\text{ is feasible}
\quad\Longleftrightarrow\quad
\mathcal A(y_\alp)\ge\alp.
\]
Positive causal dependence and integral stability make this reduction
and residual-based verification possible. Admissibility, including
monotonicity and boundary conditions, must also be established for the
target program. Our analysis verifies these requirements for the
Stochastic Balance functional program; a factor-revealing formulation
with noncausal or sign-changing couplings need not admit the same
comparison. These structural requirements specify where the methods
can transfer beyond the present application.

\subsection{Other Related Work}
\label{sec:related}

Discretized primal--dual analyses have been used in several online
matching settings~\cite{derakhshan2025,huang2020fully,huang2019online,
jin2021improved,huang2023online}. Closely related to the functional
optimization viewpoint, Peng and Tang~\cite{peng2025revisiting}
construct LP lower and upper bounds for an optimization problem
arising in the primal--dual analysis of Ranking with random arrivals. In a complementary setting, Xu~\cite{xu2026certifying}
studies the fixed rank-dependent gain-sharing protocol for
Ranking, where the gain-sharing function affects only the
analysis, and establishes a strict gap between the protocol's
optimal certifying capacity and Ranking's actual competitiveness. 
The comparison approach here additionally characterizes the optimal
critical profile and validates numerical candidates through their
integral-equation residuals. Bhangale, Chakraborty, and
Harsha~\cite{bhangale2025optimal} establish matching algorithmic and
impossibility bounds for online bipartite matching when online degrees
are at most two, another setting where structural restrictions permit
an exact performance characterization.

For stochastic rewards, benchmark and capacity assumptions are
important when comparing guarantees. Goyal and
Udwani~\cite{goyal2023online} analyze a stochastic offline benchmark
that knows the instance and arrival order but not the assignment
outcomes in advance. Against this benchmark, their $0.596$ guarantee
for infinitesimal probabilities is improved for Stochastic Balance
to $0.611$ by Huang et al.~\cite{huang2023online}, and to $0.613$
when the probabilities are also equal. In a different capacity regime,
Albers and Schubert~\cite{albers2025online} show that Stochastic
Balance approaches the optimal ratio $1-1/e$ as offline capacities
increase, allowing heterogeneous non-vanishing probabilities. These
results concern different benchmarks or regimes from the unit-capacity
configuration-LP template whose certificates we optimize here.

\xhdr{Roadmap.}
Section~\ref{sec:lb} develops the auxiliary-LP certification over
$\cF_2$; the upper envelopes for $\cF_0$ and $\cF_1$ and the
closed-form benchmark over $\cF_3$ are treated in the appendices.
Section~\ref{sec:thm2} establishes exactness and the comparison-based
structure of the full-class optimum, then gives the validated endpoint
certificate. We conclude in Section~\ref{sec:con}.

\section{Auxiliary-LP Certification of
\texorpdfstring{$\Gam(\cF_2)$}{Γ(𝓕₂)}}
\label{sec:lb}

This section proves the $\cF_2$ part of Theorem~\ref{thm:aux-lp-cert}.  Recall that $\cF_2$ is the restricted class defined in~\eqref{eq:def-F2}.  Over this class, the canonical representation from Theorem~\ref{thm:structure} collapses to two explicit adversarial expressions, making $\cF_2$ the cleanest setting for the auxiliary-LP certification method.  {Combining independently verified lower and upper bounds on a finite LP optimum with an explicit discretization-error bound, we obtain}
\[
\Gam(\cF_2)\in[0.5796,\;0.5809].
\]

The remaining parts of Theorem~\ref{thm:aux-lp-cert} are proved in the appendix: the upper certificates for $\cF_0$ and $\cF_1$ are given in Appendix~\ref{sec:ub}, and the closed-form certification over $\cF_3$ is given in Appendix~\ref{app:analy}.  The $\cF_0$ and $\cF_1$ bounds reuse the auxiliary-LP error-control argument developed here, while the $\cF_3$ bound serves as a closed-form benchmark at the most restricted end of the hierarchy.

The first step is to simplify the expression for $\cL[f]$ over $\cF_2$.

\begin{proposition}[Appendix~\ref{app:17-b}]\label{pro:2-17-a}
For any $f \in \cF_2$, the value of $\cL[f]$ in~\eqref{eqn:Lf-d-1} can be written as
\[
\cL[f] = \min\{W_1,\,W_2\},
\]
where
\begingroup
\allowdisplaybreaks
\begin{align}
W_1
&=
\int_0^1 e^{-z} f(z)\,\sd z
+ e^{-1}(1-e^{-1}),
\label{int:23-a}\\
W_2
&=
\min_{\ell,\tpsi \in [0,1]:\, \ell+\tpsi \le 1}
\Bigg(
\int_0^\ell e^{-z} f(z)\,\sd z
+ \int_\ell^{\ell+\tpsi} e^{-z}\,\sd z
+ (1-\tpsi)\bigl(1-f(\ell+\tpsi)\bigr)
\Bigg).
\label{int:23-b}
\end{align}
\endgroup
\end{proposition}

Substituting Proposition~\ref{pro:2-17-a} into Program~\eqref{eqn:Lf}, the optimization over $\cF_2$ reduces to the problem
\begin{align}\label{eqn:max_f}
\Gam(\cF_2)
=
\max_{f \in \cF_2}
\cL[f]
=
\max_{f \in \cF_2}
\min\{W_1,\,W_2\}.
\end{align}

\xhdr{Why $\cF_2$?} The restriction to $\cF_2$ is not meant to define the final admissible class.  Rather, $\cF_2$ is the largest class in our hierarchy for which the general three-parameter expression in Theorem~\ref{thm:structure} reduces to the two explicit terms $W_1$ and $W_2$.  This makes $\cF_2$ the cleanest setting for introducing the auxiliary-LP certification method.

The same certification idea extends to larger classes, but the computational burden increases.  In particular, applying the discretization directly to the full expression~\eqref{eqn:Lf-d-1} requires discretizing the additional total-load parameter $\psi$, increasing the number of LP constraints by roughly a factor of the grid size.  
Appendix~\ref{sec:ub} instead uses the simplified expressions $W_1,W_2$ as upper certificates for larger classes, yielding the auxiliary-LP upper bounds for $\cF_1$ and $\cF_0$.

\subsection{An Auxiliary-LP Certification Framework for Program~\eqref{eqn:max_f}}
\label{sec:disc}

We now construct the finite-dimensional certificate used to prove the bound on $\Gam(\cF_2)$.  The certificate is a sequence of auxiliary LPs whose optimal values converge to Program~\eqref{eqn:max_f} with an explicit error bound.  {The role of the LP is therefore not merely computational: a certified enclosure of a finite LP optimum, together with the convergence guarantee below, yields a rigorous interval for the original functional optimization problem.  A feasible primal solution alone gives only a lower bound on the finite optimum; an upper bound requires a separate certificate.}

Fix an integer $n\ge 1$.  We discretize the restriction of $f$ to $[0,1]$ by introducing variables
\[
x_t := f(t/n), \qquad t\in (n):=\{0,1,2,\ldots,n\}.
\]
The monotonicity and envelope constraints defining $\cF_2$ become finite linear constraints on the vector $(x_t)_{t=0}^n$.  The two adversarial expressions $W_1$ and $W_2$ in Proposition~\ref{pro:2-17-a} are then represented by right-endpoint Riemann sums.  This gives the following auxiliary LP.

\smallskip
\begin{tcolorbox}[title=(\textbf{AUG-LP}) Auxiliary LP for Program~\eqref{eqn:max_f}]
 \begingroup
\allowdisplaybreaks
 \begin{align}
&\max y \label{lp:obj} \\
& 1-e^{-t/n} \le x_t \le x_{t+1}, &&\forall t \in (n-1),  \label{lp:cons-1}\\
&x_{n}= 1-e^{-1},  \label{lp:cons-2}\\
& y \le  \frac{1}{n}\sum_{t=1}^{n}x_t\,\sfe^{-(t/n)} +\sfe^{-1}(1-1/\sfe),  \label{lp:cons-3}
\\
& y \le   \frac{1}{n}\sum_{t=1}^{i}x_t\,\sfe^{-(t/n)} + \frac{1}{n}\sum_{t=i+1}^{i+j} e^{-(t/n)}+(1-j/n)\,(1-x_{i+j}), && \forall i \in (n), j\in (n-i).  \label{lp:cons-4}
\end{align}
\endgroup 
\end{tcolorbox}

\xhdr{Interpretation of \textbf{AUG-LP}.}
The LP has $n+2$ variables and $\Theta(n^2)$ constraints.  Constraint~\eqref{lp:cons-1} encodes the lower envelope and monotonicity conditions in $\cF_2$, while Constraint~\eqref{lp:cons-2} encodes the terminal value at $1$.  Constraint~\eqref{lp:cons-3} is the finite certificate corresponding to $W_1$, and Constraint~\eqref{lp:cons-4} is the family of finite certificates corresponding to $W_2$ evaluated at grid points
\[
\ell=i/n,\qquad \tpsi=j/n.
\]
We use right-endpoint sums throughout.  This convention keeps the LP linear and gives the explicit error control in Proposition~\ref{pro:2-18-a}.  The variable $x_0$ is retained only for notational uniformity.

\begin{proposition}[Auxiliary-LP certificate for $\Gam(\cF_2)$]\label{pro:2-18-a}
Let $\eta(n)$ denote the optimal value of \textbf{AUG-LP}~\eqref{lp:obj} with discretization parameter $n$.  Then
\[
\eta(n)\to \Gam(\cF_2)
\qquad\text{as } n\to\infty,
\]
and, more precisely,
\begin{align}\label{ineq:2-25-1}
\bigl|\eta(n)-\Gam(\cF_2)\bigr|
\le
\frac{1-1/e}{n},
\qquad \forall n\in\mathbb{N}_{\ge 1}.
\end{align}
{More generally, any independently verified enclosure
$\underline{\eta}_n\leq\eta(n)\leq\overline{\eta}_n$ gives}
\[
\underline{\eta}_n-\frac{1-1/e}{n}
\leq \Gam(\cF_2)
\leq \overline{\eta}_n+\frac{1-1/e}{n}.
\]
{In particular, the finite LP certificates in
Table~\ref{table:nume} establish}
\[
0.5802563\leq\eta(1000)\leq0.5802564,
\]
{and consequently}
\[
\Gam(\cF_2)\in[0.5796,\;0.5809].
\]
\end{proposition}

The proof has two ingredients.  First, \textbf{AUG-LP} is consistent with Program~\eqref{eqn:max_f}: the right-endpoint discretizations of $W_1$ and $W_2$ converge to the corresponding continuous expressions, so the limiting LP value is $\Gam(\cF_2)$.  Second, the following refinement lemmas give an explicit finite-grid error bound.  The first transfers feasible solutions from grid size $n$ to grid size $2n$, while the second transfers feasible solutions from grid size $2n$ back to grid size $n$.

\begin{lemma}[Refinement: from $n$ to $2n$; Appendix~\ref{app:20-lb}]
\label{lem:20-lb}
Suppose \textbf{AUG-LP}~\eqref{lp:obj} with parameter $n\in\mathbb{N}_{\ge1}$ admits a feasible solution $(y,\bx)$, where $\bx=(x_t)_{t\in(n)}$.  Then \textbf{AUG-LP}~\eqref{lp:obj} with parameter $2n$ admits a feasible solution $(\tilde y,\tilde\bx)$, where $\tilde\bx=(\tilde x_t)_{t\in(2n)}$, defined by
\begin{align}\label{def:1}
\tilde y
=
y-\frac{1-1/e}{2n},
\qquad
\tilde x_0=x_0,
\qquad
\tilde x_{2t-1}=\tilde x_{2t}=x_t,
\quad \forall\,1\le t\le n.
\end{align}
\end{lemma}

\begin{lemma}[Coarsening: from $2n$ to $n$; Appendix~\ref{app:20-ub}]
\label{lem:20-ub}
Suppose \textbf{AUG-LP}~\eqref{lp:obj} with parameter $2n$, where $n\in\mathbb{N}_{\ge1}$, admits a feasible solution $(\tilde y,\tilde\bx)$, where $\tilde\bx=(\tilde x_t)_{t\in(2n)}$.  Then \textbf{AUG-LP}~\eqref{lp:obj} with parameter $n$ admits a feasible solution $(y,\bx)$, where $\bx=(x_t)_{t\in(n)}$, defined by
\begin{align}\label{def:2}
y
&=
\tilde y-\frac{1-1/e}{2n},
\qquad
x_0=\tilde x_0,
\qquad
x_n=\tilde x_{2n}, \nonumber\\
x_t
&=
\frac{\tilde x_{2t}+\tilde x_{2t+1}}{2},
\qquad \forall\,1\le t\le n-1.
\end{align}
\end{lemma}

\begin{proof}[\textbf{Proof of Proposition~\ref{pro:2-18-a}}.]
{\xhdr{Continuum consistency.} Fix $n\geq1$, put $h:=1/n$,
$t_r:=rh$, and $c:=1-1/e$. Write $\sigma:=\ell+\tpsi$ and denote
the test in~\eqref{int:23-b} by $K_f(\ell,\sigma)$, so that
$W_2(f)=\min_{0\leq\ell\leq\sigma\leq1}K_f(\ell,\sigma)$.
Set $a(z):=e^{-z}$ and $b_f(z):=e^{-z}(1-f(z))$. Then}
\[
{K_f(\ell,\sigma)
=\int_0^\sigma a(z)\,\sd z-\int_0^\ell b_f(z)\,\sd z
 +(1-\sigma+\ell)(1-f(\sigma)).}
\]
{Likewise, the integral in $W_1(f)$ is
$\int_0^1 a(z)\,\sd z-\int_0^1 b_f(z)\,\sd z$. Both $a$ and
$b_f$ are nonincreasing and take values in $[0,1]$. For any such
function $u$, comparison of the left- and right-endpoint sums gives}
\[
{0\leq \int_0^{t_m}u(z)\,\sd z-h\sum_{r=1}^m u(t_r)
\leq h\bigl(u(0)-u(t_m)\bigr)\leq h,\qquad 0\leq m\leq n.}
\]
{Consequently, on setting $x_r=f(t_r)$, the right-hand sides
of~\eqref{lp:cons-3} and~\eqref{lp:cons-4} differ from $W_1(f)$
and $K_f(t_i,t_{i+j})$, respectively, by at most $h$: in each case
the discrepancy is the difference of two errors in $[0,h]$, and
the nonintegral term is unchanged. Thus every $f\in\cF_2$ yields
a feasible grid solution at level $\cL[f]-h$, proving
$\eta(n)\geq\Gam(\cF_2)-h$. The grid LP is nonempty
(take $f(z)=1-e^{-\min\{z,1\}}$), and~\eqref{lp:cons-3}
bounds its objective above.}

{Conversely, let $(y,\bx)$ be feasible for \textbf{AUG-LP}, and
let $p_x$ be the piecewise-linear interpolant of $(t_r,x_r)_{r=0}^n$.
Set $f_x(z):=\max\{p_x(z),1-e^{-z}\}$ on $[0,1]$ and $f_x(z):=c$
for $z\geq1$. Taking the maximum enforces the lower envelope
between nodes without changing any nodal value. Both functions
in the maximum are continuous, nondecreasing, and bounded above
by $c$, and $f_x(1)=x_n=c$. Thus $f_x\in\cF_2$.
The grid-point estimate gives $W_1(f_x)\geq y-h$ and
$K_{f_x}(t_i,t_k)\geq y-h$ for all $0\leq i\leq k\leq n$.}

{For an arbitrary $0\leq\ell\leq\sigma\leq1$, round the
initial load down and the total load up: take
$i:=\lfloor n\ell\rfloor$ and $k:=\lceil n\sigma\rceil$.
Since $0\leq1-t_k+t_i\leq1-\sigma+\ell$ and
$0\leq1-f_x(t_k)\leq1-f_x(\sigma)$, this rounding cannot
increase the nonintegral term. Hence}
\[
{K_{f_x}(t_i,t_k)-K_{f_x}(\ell,\sigma)
 \leq \int_\sigma^{t_k}a(z)\,\sd z
      +\int_{t_i}^{\ell}b_{f_x}(z)\,\sd z
 \leq 2h.}
\]
{Thus every continuous partial-load test is at least $y-3h$.
By Proposition~\ref{pro:2-17-a}, $\cL[f_x]\geq y-3h$.
Taking the supremum over feasible grid solutions and combining
the two directions gives}
\[
{\eta(n)-\frac{3}{n}\leq\Gam(\cF_2)
       \leq\eta(n)+\frac{1}{n}.}
\]
{Thus the limit $\eta(\infty):=\lim_{n\to\infty}\eta(n)$
exists along the full sequence of grid sizes and equals
$\Gam(\cF_2)$.}

{\xhdr{Sharp discretization bound.}
Lemmas~\ref{lem:20-lb} and~\ref{lem:20-ub} give}
\begin{equation}\label{ineq:eta-dyadic}
{\bigl|\eta(2n)-\eta(n)\bigr|\leq\frac{c}{2n}.}
\end{equation}
{For fixed $n$, telescoping this estimate over
$n,2n,\ldots,2^k n$ yields}
\[
{\bigl|\eta(2^k n)-\eta(n)\bigr|
 \leq\sum_{r=1}^k
       \bigl|\eta(2^r n)-\eta(2^{r-1}n)\bigr|
 \leq\frac{c}{n}(1-2^{-k}).}
\]
{Letting $k\to\infty$ and using the continuum limit just
proved gives~\eqref{ineq:2-25-1}.}

For the numerical conclusion, use the separately verified
primal and upper-bound certificates described in
Appendix~\ref{sec:ub}. The $n=1000$ row of
Table~\ref{table:nume} supplies the exact rational bounds
$\underline{\eta}_{1000}=0.5802563$ and
$\overline{\eta}_{1000}=0.5802564$ on the finite optimum. Thus
\[
0.5802563-\frac{1-1/e}{1000}
\leq \Gam(\cF_2)
\leq 0.5802564+\frac{1-1/e}{1000}.
\]
{Rounding the two endpoints outward gives the stated enclosure
$\Gam(\cF_2)\in[0.5796,0.5809]$.}
The finite LP bounds already include arithmetic and feasibility
errors; the term $(1-1/e)/1000$ accounts separately for
discretization. No rounded solver objective is used as an exact
LP optimum.
\end{proof}

\xhdr{Remark on Proposition~\ref{pro:2-18-a}.}
{The telescoping argument applies to any sequence of finite LP
values $v(n)$ for which continuum consistency $v(n)\to V$ and
a dyadic transfer bound $|v(2n)-v(n)|\leq\delta/(2n)$ have been
established. Together they imply}
\begin{equation}\label{ineq:2-25-2}
{\bigl|v(n)-V\bigr|\leq\frac{\delta}{n},
       \qquad n\in\mathbb N_{\geq1}.}
\end{equation}
{Here $v(n)=\eta(n)$, $V=\Gam(\cF_2)$, and $\delta=1-1/e$.
For the upper-envelope LPs in Appendix~\ref{sec:ub}, $V$ is the
corresponding envelope optimum. The same consistency proof applies
there with the ordinary interpolant $p_x$, extended constantly
by $x_n$ beyond $1$, since neither LP imposes the lower
envelope $1-e^{-z}$.}

\section{Comparison-Based Certification of
\texorpdfstring{$\Gam(\cF_0)$}{Γ(𝓕₀)}}
\label{sec:thm2}

This section proves Theorem~\ref{thm:main}. Starting from the canonical
representation of $\cL[f]$ in Theorem~\ref{thm:structure}, we formulate
an optimization problem over $f\in\cF_0$ whose value is exactly
\[
\Gam(\cF_0)=\sup_{f\in\cF_0}\cL[f].
\]
We characterize its optimal solutions by a comparison argument and then
certify the optimal value using numerical candidates and a separate,
rigorous verification.

\xhdr{Two families of adversarial tests.}
For $f\in\cF_0$, define
\begin{align}\label{eqn:11-29-a}
\tau(f):=\inf\{z\in[0,1]:f(z)\ge1-1/e\},
\end{align}
with the convention $\tau(f)=1$ if $f(1)<1-1/e$. We call
$[0,\tau(f)]$ the \emph{critical interval}: it ends at the first
threshold crossing, or at $1$ if the threshold is not reached.
We write $\tau=\tau(f)$ when the function is clear from context.
Define
\begin{align}
\Phi_1(f)
&:=\int_0^{\tau(f)}e^{-z}f(z)\,\sd z
 +(1-1/e)e^{-\tau(f)},
\label{int:12-2-a}\\
\Phi_2(f)
&:=\min_{\substack{\ell,\tpsi\in[0,1]\\\ell+\tpsi\le\tau(f)}}
\Psi_f(\ell,\tpsi),
\label{int:12-2-b}\\
\Psi_f(\ell,\tpsi)
&:=\int_0^\ell e^{-z}f(z)\,\sd z
 +\int_\ell^{\ell+\tpsi}e^{-z}\,\sd z
 +(1-\tpsi)\sbp{1-f(\ell+\tpsi)}.
\nonumber
\end{align}
The functional $\Phi_1$ evaluates the canonical test with
$(\ell,\psi,\tpsi)=(\tau(f),1,1)$, while $\Phi_2$ minimizes the
tests with $\psi=1$ and $\ell+\tpsi\le\tau(f)$.
Both functionals depend only on $f$ on its critical interval. When
$\tau(f)=1$, they reduce to $W_1$ and $W_2$
from~\eqref{int:23-a}--\eqref{int:23-b}.

\xhdr{Minimizing over the initial load.}
For a fixed total load $\sig=\ell+\tpsi\in[0,\tau(f)]$, set
\[
\Psi_{f,\sig}(\ell):=\Psi_f(\ell,\sig-\ell),
\qquad 0\le\ell\le\sig.
\]
The following lemma identifies its minimizing initial load.

\begin{lemma}\label{lem:11-29-c}
For every $f\in\cF_0$ and $\sig\in(0,\tau(f)]$, the function
$\Psi_{f,\sig}$ is strictly convex on $[0,\sig]$. For every
$\sig\in[0,\tau(f)]$, its unique minimizer, denoted by $\ell(\sig)$,
is the unique solution in $[0,\sig]$ of
\begin{align}\label{eqn:ell-sig}
1-f(\sig)=e^{-\ell}\sbp{1-f(\ell)}.
\end{align}
In particular, $\ell(0)=0$.
\end{lemma}

\begin{proof}
The case $\sig=0$ is immediate. For $\sig>0$, continuity and
monotonicity give $f(\ell)\le1-1/e$ on $[0,\sig]$, and
\[
\frac{\sd\Psi_{f,\sig}(\ell)}{\sd\ell}
=1-f(\sig)-e^{-\ell}\sbp{1-f(\ell)}.
\]
Since $1-f$ is non-increasing and at least $1/e$ on this interval,
$e^{-\ell}(1-f(\ell))$ is continuous and strictly decreasing.
Thus the derivative is continuous and strictly increasing. Its
endpoint values satisfy
\begin{align*}
\left.\frac{\sd\Psi_{f,\sig}(\ell)}{\sd\ell}\right|_{\ell=0}
&=f(0)-f(\sig)\le0,\\
\left.\frac{\sd\Psi_{f,\sig}(\ell)}{\sd\ell}\right|_{\ell=\sig}
&=(1-f(\sig))(1-e^{-\sig})>0.
\end{align*}
Consequently, $\Psi_{f,\sig}$ is strictly convex, and its derivative
has a unique zero in $[0,\sig]$. This zero is the unique minimizer
and satisfies~\eqref{eqn:ell-sig}.
\end{proof}

By Lemma~\ref{lem:11-29-c},
\begin{align}\label{eqn:4-18-a}
\Phi_2(f)=\min_{\sig\in[0,\tau(f)]}\Psi_{f,\sig}(\ell(\sig)).
\end{align}
Write $\Psi_f(\sig):=\Psi_{f,\sig}(\ell(\sig))$ for the minimized
value at total load $\sig$. Explicitly,
\begin{align}\label{eqn:4-18-b}
\Psi_f(\sig)
&=\int_0^{\ell(\sig)}e^{-z}f(z)\,\sd z
 +\int_{\ell(\sig)}^\sig e^{-z}\,\sd z
 +(1-\sig+\ell(\sig))(1-f(\sig)).
\end{align}
Thus the one-argument notation $\Psi_f(\sig)$ denotes the minimum
over $\ell$, and
\begin{align}\label{eqn:4-18-c}
\Phi_2(f)=\min_{\sig\in[0,\tau(f)]}\Psi_f(\sig).
\end{align}

\xhdr{An exact functional optimization program.}
Consider the following program.
\begin{tcolorbox}[title=Functional Optimization Program for $\Gam(\cF_0)$]
\begin{align}\label{prog:var-gamma-F0}
\sup_{\gam\in\mathbb R,\,f\in\cF_0}\quad &\gam\\
\text{s.t.}\quad
&\Phi_1(f)\ge\gam,\nonumber\\
&\Phi_2(f)\ge\gam
\;\Leftrightarrow\;
\bp{\Psi_f(\sig)\ge\gam,\quad\forall\sig\in[0,\tau(f)]}.
\nonumber
\end{align}
\end{tcolorbox}
For a fixed $f$, the largest feasible $\gam$ is
$\min\{\Phi_1(f),\Phi_2(f)\}$. Each test restricts the adversarial
choices in Theorem~\ref{thm:structure}, so
\[
\cL[f]\le\min\{\Phi_1(f),\Phi_2(f)\},\qquad f\in\cF_0.
\]
Hence Program~\eqref{prog:var-gamma-F0} gives an upper bound on
$\Gam(\cF_0)$. The next theorem establishes equality of the optimal
values.

\begin{theorem}[Exact functional characterization of $\Gam(\cF_0)$]
\label{thm:main-exact}
Let $\gam^*$ denote the supremal value of
Program~\eqref{prog:var-gamma-F0}. Then
\[
\Gam(\cF_0)=\gam^*.
\]
\end{theorem}

The proof is given in Appendix~\ref{app:thm-exact-opt}. The upper
bound follows from the displayed inequality. For the matching lower
bound, the structural theorem below provides an optimal function on
its critical interval. Continuing that function constantly at $1-1/e$
beyond the interval yields a function in $\cF_2$ that attains
$\gam^*$ in the original optimization. The structural theorem is
proved directly for Program~\eqref{prog:var-gamma-F0}, independently
of this exactness statement.

Both constraint families are needed. Removing $\Phi_1(f)\ge\gam$
gives value $1-1/e$, attained by constant functions; under the full
program, constant functions are strictly suboptimal. These comparisons
are proved in Appendix~\ref{app:const-class}.

\begin{theorem}[Existence, flatness, and regularity of optimal solutions]
\label{thm:main-structural}
The supremal value $\gam^*$ of Program~\eqref{prog:var-gamma-F0}
is attained. For every optimal solution $(\gam^*,f^*)$, let
$\tau^*:=\tau(f^*)$. Then
\begin{align}
\Phi_1(f^*)&=\gam^*,\label{eqn:12-12-a}\\
\Psi_{f^*}(\sig)&\equiv\gam^*,
\qquad \sig\in[0,\tau^*].\label{eqn:12-12-b}
\end{align}
Moreover,
\[
0<\tau^*<1,\qquad f^*(\tau^*)=1-1/e.
\]
The restriction of $f^*$ to $[0,\tau^*]$ belongs to $C^1([0,\tau^*])$
and is strictly increasing, with
\[
1-f^*(z)<e^{-z},\qquad 0\le z\le\tau^*.
\]
All optimal solutions have the same critical interval and coincide
there. Endpoint derivatives are understood in the one-sided sense
from within the critical interval.
\end{theorem}

The proof in Appendix~\ref{app:gam-f} uses a comparison principle for
an integral equation relating a profile at each load to its values
at earlier loads. For the main argument, the key conclusion is
\emph{flatness}: after minimizing over the initial load $\ell$, every
total load $\sig$ in the critical interval yields the same value
$\gam^*$. This concerns $\Psi_{f^*}(\sig)$, not every individual
test $\Psi_{f^*}(\ell,\tpsi)$. The uniqueness conclusion concerns
only the critical interval, not the continuation beyond it.

The $C^1$ regularity allows us to differentiate the flatness identity.
Together with the minimizing-load equation, this gives the following
system for the optimal trajectory.

\begin{corollary}\label{cor:main-struc}
Let $(\gam^*,f^*)$ be an optimal solution to
Program~\eqref{prog:var-gamma-F0}, and let $\tau^*:=\tau(f^*)$.
Then
\begin{align}
1-f^*(z)
&=e^{-\ell^*(z)}\bigl(1-f^*(\ell^*(z))\bigr),
\qquad 0\le z\le\tau^*,
\label{eqn:12-13-c}\\
\frac{\sd f^*(z)}{\sd z}
&=\frac{e^{-z}-(1-f^*(z))}{1-z+\ell^*(z)}>0,
\qquad 0\le z\le\tau^*,
\label{eqn:12-13-b}\\
\gam^*
&=1-f^*(0)
=\int_0^{\tau^*}e^{-z}f^*(z)\,\sd z
 +(1-1/e)e^{-\tau^*},
\label{eqn:12-13-d}
\end{align}
where $\ell^*(z)$ is the unique solution of~\eqref{eqn:ell-sig}
for $f^*$. The function $\ell^*$ belongs to $C^1([0,\tau^*])$,
$\ell^*(0)=0$, and
\[
0<\ell^*(z)<z,\qquad 0<z\le\tau^*.
\]
The differential identity holds pointwise, with the right derivative
at $0$ and the left derivative at $\tau^*$.
\end{corollary}

The proof is given in Appendix~\ref{app:cor-main-struc}.
The identity $\gam^*=1-f^*(0)$ identifies the initial value of the
optimal trajectory, while the integral in~\eqref{eqn:12-13-d}
provides a consistency condition. The differential equation has a
strictly positive numerator, and
\[
1-z+\ell^*(z)\ge1-z\ge1-\tau^*>0,
\qquad 0\le z\le\tau^*.
\]
These properties lead to the shooting formulation in
Section~\ref{sec:numerical-gamma-star} and the endpoint verification
in Section~\ref{sec:validated-gamma-star}.

\subsection{Euler Shooting and Candidate Generation}
\label{sec:numerical-gamma-star}

We use the initial value to parameterize the system in
Corollary~\ref{cor:main-struc}. For a candidate value
$\alp\in(1/e,1)$, seek a non-decreasing trajectory
$f_\alp\in C^1([0,\tau_\alp])$ and a load function
$\ell_\alp\in C^1([0,\tau_\alp])$ satisfying
\begin{align}
1-f_\alp(z)
&=e^{-\ell_\alp(z)}\sbp{1-f_\alp(\ell_\alp(z))},
&&z\in[0,\tau_\alp],
\label{eq:num-alpha-a}\\
\frac{\sd f_\alp(z)}{\sd z}
&=\frac{e^{-z}-(1-f_\alp(z))}{1-z+\ell_\alp(z)},
&&z\in[0,\tau_\alp],
\label{eq:num-alpha-b}\\
f_\alp(0)&=1-\alp.
\label{eq:num-alpha-c}
\end{align}
Here $\ell_\alp(z)\in[0,z]$, and $\tau_\alp\in(0,1)$ is the first
time $f_\alp(\tau_\alp)=1-1/e$. We extend $f_\alp$ constantly at
$1-1/e$ for $z>\tau_\alp$, so $\tau(f_\alp)=\tau_\alp$.
Endpoint derivatives are taken from within the critical interval.
The verification in Section~\ref{sec:validated-gamma-star}
establishes these trajectories on the parameter interval used below.

For such a trajectory, define its \emph{return value}
\begin{align}\label{eq:def-q-alpha}
q(\alp):=\Phi_1(f_\alp)
=\int_0^{\tau_\alp}e^{-z}f_\alp(z)\,\sd z
 +(1-1/e)e^{-\tau_\alp}.
\end{align}
At $\alp=\gam^*$, Theorem~\ref{thm:main-structural} and
Corollary~\ref{cor:main-struc} supply the optimal critical trajectory,
and $q(\gam^*)=\gam^*$. The optimal value therefore solves
\begin{align}\label{eq:fixed-point-alpha}
q(\alp)=\alp.
\end{align}

\begin{lemma}[Monotonicity of the return value]
\label{lem:q-monotone}
Let $1/e<\alp_1<\alp_2<1$. For each $i\in\{1,2\}$, let
$f_i,\ell_i\in C^1([0,\tau_i])$ satisfy
\eqref{eq:num-alpha-a}--\eqref{eq:num-alpha-c} with parameter
$\alp_i$, where $f_i$ is non-decreasing, $\ell_i(z)\in[0,z]$,
and $\tau_i\in(0,1)$ is the first time $f_i(\tau_i)=1-1/e$.
Define $q_i$ by~\eqref{eq:def-q-alpha} using $f_i$ and $\tau_i$.
Then $q_1>q_2$.
\end{lemma}

The proof appears in Appendix~\ref{app:q-monotone}. Thus, on an
interval where the shooting trajectories are well-defined, $q$ is
strictly decreasing and has at most one fixed point.

\xhdr{Numerical candidates.}
We approximate~\eqref{eq:num-alpha-b} by forward Euler on a uniform
grid of step size $h$, solving the implicit load equation against a
piecewise-linear interpolant at each step. Interpolation of the
threshold crossing and trapezoidal quadrature give an exploratory
return estimate $Q_h(\alp)$. The scheme is described in
Appendix~\ref{app:dis-scheme}. With $h=1/3000$, the values in
Table~\ref{tab:numerical-alpha} suggest a sign change of
$Q_h(\alp)-\alp$ between $0.5803$ and $0.5804$. We therefore choose
\[
a:=0.5802,\qquad b:=0.5805
\]
as endpoints for a rigorous enclosure. A finer computation generates
explicit piecewise-linear candidates for $1-f_\alp$ at these two
parameters. The next subsection verifies the recorded candidates
independently of the accuracy of the Euler computation.

\subsection{Validated Endpoint Certification}
\label{sec:validated-gamma-star}

For each $\alp\in\{a,b\}$, let $u_\alp:[0,1]\to[1/e,1]$ be the
recorded candidate for $1-f_\alp$. Each $u_\alp$ is continuous and
non-increasing, is piecewise linear with mesh $h=1/96000$ on
$[0,\eta]$, where $\eta:=33/50$, and satisfies
\[
u_\alp(0)=\alp,\qquad u_\alp(z)=1/e\quad(\eta\le z\le1).
\]
Define its score by
\[
\widehat Q_\alp
:=1-e^{-(\eta+1)}-\int_0^\eta e^{-z}u_\alp(z)\,\sd z.
\]
The profile is specified exactly by its nodal data. Its score is
enclosed using exact cell-integral formulas and outward-rounded
arithmetic, so this calculation has no quadrature remainder.

\xhdr{Validation of the exact trajectories.}
Appendix~\ref{app:error} checks finitely many inequalities for each
candidate and bounds its residual in an integral-equation formulation
of the shooting system. Lemmas~\ref{lem:cert-residual}
and~\ref{lem:cert-return} establish an exact trajectory $f_\alp$
with $0<\tau_\alp<\eta$ and
\[
\|(1-f_\alp)-u_\alp\|_{\infty,[0,\eta]}<\frac1{5000}.
\]
They also establish flatness of its partial-load envelope:
\[
\Psi_{f_\alp}(\sig)=\alp,
\qquad 0\le\sig\le\tau_\alp,
\]
so $\Phi_2(f_\alp)=\alp$. Thus the partial-load constraint is
satisfied at each endpoint; the remaining question is whether
$\Phi_1(f_\alp)=q(\alp)$ is above or below that endpoint.

Because $f_\alp$ is constant at $1-1/e$ after
$\tau_\alp<\eta$, its return value can be written as
\[
q(\alp)
=1-e^{-(\eta+1)}-\int_0^\eta e^{-z}(1-f_\alp(z))\,\sd z.
\]
Consequently,
\[
|q(\alp)-\widehat Q_\alp|
\le(1-e^{-\eta})\|(1-f_\alp)-u_\alp\|_{\infty,[0,\eta]}
<10^{-4},\qquad \alp\in\{a,b\}.
\]
Combining this bound with the outward-rounded score enclosures in
Table~\ref{tab:cert-finite-results} gives
\[
q(a)\ge0.5803315>a,\qquad q(b)\le0.5803333<b.
\]
The certificate data and all verification details are given in
Appendix~\ref{app:error}.

\xhdr{Localization of the optimum.}
The lower endpoint is feasible in Program~\eqref{prog:var-gamma-F0}:
$\Phi_2(f_a)=a$ and $\Phi_1(f_a)=q(a)>a$. Hence $\gam^*\ge a$.
For the upper bound, the comparison principle of
Appendix~\ref{app:gam-f}, applied to the validated trajectory $f_b$
in Appendix~\ref{app:error}, bounds $\Phi_1(f)$ by $q(b)$ for every
function $f$ feasible at value $b$. If $\gam^*\ge b$, an optimal
solution $(\gam^*,f^*)$ from Theorem~\ref{thm:main-structural}
would be feasible at value $b$, giving
\[
\gam^*\le\Phi_1(f^*)\le q(b)<b,
\]
a contradiction. By Theorem~\ref{thm:main-exact}, we conclude that
\[
\Gam(\cF_0)=\gam^*\in[0.5802,\;0.5805],
\]
which proves Theorem~\ref{thm:main}.

Appendix~\ref{app:error} also validates the shooting trajectories
and establishes continuity of $q$ throughout $[a,b]$.
Together with Lemma~\ref{lem:q-monotone} and $q(\gam^*)=\gam^*$,
this identifies $\gam^*$ as the unique fixed point
of~\eqref{eq:fixed-point-alpha} on $[a,b]$.

\section{Conclusions and Future Directions}
\label{sec:con}

We studied gain-sharing optimization as a means of certifying the
limits of a specified algorithm--analysis template. Our methods combine
auxiliary LPs with explicit error bounds and monotone Volterra
comparison with independent residual-based verification. For the
\sto\ template in vertex-weighted \om\ with uniform vanishing success
probabilities, the canonical reduction and subsequent functional
analysis yield
\[
\Gam(\cF_0)=\Gam(\cF_1)=\Gam(\cF_2)
\in[0.5802,\;0.5805].
\]
The equality shows that the restrictions defining $\cF_2$ are
lossless at the optimum. The upper bound limits certificates obtained
from the specified configuration-LP benchmark and dual updates;
it does not bound actual algorithmic performance or exclude stronger
randomized primal--dual analyses.

Future work includes extending these methods to other functional and
factor-revealing programs by identifying suitable comparison structures
or refinement and coarsening maps with controlled error. Within
stochastic matching, natural settings include non-uniform success
probabilities, larger capacities, and alternative benchmarks.
More broadly, comparing function classes can distinguish a suboptimal
gain-sharing function from a restrictive admissible class or a
limitation of the template, helping determine whether progress calls
for refining the certificate or changing the algorithm or its analysis.

\clearpage
\appendix

\section{The Gain-Sharing Function of Huang and Zhang}
\label{app:prior-gain-sharing}

For reference, the function used in~\cite{huang2020online} to obtain
the guarantee of approximately $0.576$ for equal vanishing probabilities
is
\begin{align}\label{eqn:1-25-a}
\tilde f(x)
=
\begin{cases}
\displaystyle
\frac{1}{h(x)}
\left(
1-\frac1e+
\int_x^1(1-e^{-y})g(y)h(y)\,\sd y
\right),
&0\le x\le1,\\[1.2ex]
1-1/e,&x\ge1,
\end{cases}
\end{align}
where, locally in this appendix,
\[
g(x):=\frac{1}{2-x-e^{-x}},
\qquad
h(x):=\exp\!\left(\int_x^1g(z)\,\sd z\right),
\qquad 0\le x\le1.
\]


\section{Proof of Theorem~\ref{thm:structure}}
\label{app:thm-structure}

This appendix gives the stochastic-matching model and proves the canonical
representation in Theorem~\ref{thm:structure}.  We first reduce the
worst-case dual constraint for a fixed gain-sharing function to two canonical
arrival effects, and then show that these effects can be arranged in a
load-then-divert normal form.

\subsection{Model and Certificate Functional}

\xhdr{Online Matching with Stochastic Rewards (\om).}
We consider a bipartite graph $G=(U,V,E)$, where $U$ and $V$ are the
sets of offline and online agents, respectively.  Each offline agent
$u\in U$ has weight $w_u>0$ and unit capacity.  Online agents arrive
in an adversarial order.  Upon the arrival of $v\in V$, its neighborhood
$N_v$ is revealed, and the algorithm may assign $v$ to an available
neighbor.  Each assignment succeeds independently with a common
probability $p\in(0,1)$; success earns reward $w_u$ and makes $u$
unavailable, whereas failure leaves $u$ available.  In either case,
$v$ departs immediately after the assignment.  The graph, weights, and
arrival order are fixed independently of the assignment outcomes.  We
study the regime $p\downarrow0$, referred to in the main body as
uniform vanishing success probabilities.

\xhdr{Loads and the exponential-threshold representation.}
The load $\ell_u$ of an offline agent is the cumulative assignment
probability placed on it, so every assignment increases its load by
$p$.  The number of attempts before the first success is geometric.
After multiplication by $p$, this random threshold converges to an
$\Exp(1)$ random variable as $p\downarrow0$.  Accordingly, in the
continuous-load limit we represent the reward randomness by independent
latent thresholds
\[
\Theta_u\sim\Exp(1),
\]
where $u$ remains available precisely while $\ell_u<\Theta_u$.  In
particular, an agent that has accumulated load $\ell$ is available with
probability $e^{-\ell}$.  This is the standard limiting representation
used in the analysis of Stochastic Balance; see~\cite{huang2020online}.

\xhdr{Stochastic Balance and the gain-sharing function.}
We study the vertex-weighted Stochastic Balance template (\sto),
parameterized by
\[
f\in\cF_0:=C_{\uparrow}([0,\infty),[0,1]),
\]
the class of continuous non-decreasing functions from $[0,\infty)$ to
$[0,1]$.  For a fixed $f$, an arriving online agent $v$ is assigned to
an available neighbor maximizing
\[
w_up\bigl(1-f(\ell_u)\bigr),
\]
with ties resolved by a fixed deterministic rule.  Loads in this score
are evaluated immediately before the assignment.  If no neighbor is
available, $v$ is left unassigned.

The same function $f$ determines the dual updates below.  Thus, in the
vertex-weighted setting, changing $f$ may affect both the allocation rule
and its analysis; throughout the paper, optimization over $f$ refers to
this specified algorithm--analysis template.

\xhdr{Configuration-LP benchmark and dual updates.}
For $u\in U$, let $N_u$ denote its online neighbors.  The
configuration-LP benchmark used in~\cite{huang2020online} has
coefficient
\[
p_u(S):=\min\{|S|p,1\},\qquad S\subseteq N_u,
\]
and dual constraints
\[
\alp_u+\sum_{v\in S}\beta_v\ge w_up_u(S),
\qquad u\in U,\ S\subseteq N_u.
\]
The dual variables are nonnegative and the dual objective is
$\sum_{u\in U}\alp_u+\sum_{v\in V}\beta_v$.  Notice that $p_u(S)$
is the capped total assignment probability appearing in the benchmark;
it is not the probability $1-(1-p)^{|S|}$ that at least one of the
independent attempts in $S$ succeeds.

Initialize all dual variables to zero.  Whenever $v$ is assigned to
$u$, use the pre-assignment load to update
\begin{align}\label{eq:pd-update}
\alp_u&\gets\alp_u+w_upf(\ell_u),
&
\beta_v&\gets w_up\bigl(1-f(\ell_u)\bigr).
\end{align}
If $v$ is unassigned, set $\beta_v=0$.  The total dual increment of an
assignment is $w_up$, equal to its conditional expected reward.  Hence
the expected dual objective equals the algorithm's expected reward.
A uniform expected dual-feasibility factor $\gam>0$ therefore certifies
a $\gam$-competitive guarantee by scaling the expected dual variables
by $1/\gam$ and applying weak duality.

\xhdr{The certificate functional.}
For fixed $f$ and $p$, the normalized dual constraint depends on the
input instance as well as on the target offline vertex and selected
neighbor set.  Let $\mathcal I$ range over all finite input instances
with at least one edge, including their offline weights and arrival
orders.  The concrete specialization of~\eqref{def:cL-f} is
\begin{align}\label{def:cL-f-2}
\cL[f]
:=\liminf_{p\downarrow0}\ \inf_{\mathcal I}
\min_{\substack{u\in U,\ S\subseteq N_u\\p_u(S)>0}}
\frac{\E_{\mathcal I}\!\left[\alp_u+\sum_{v\in S}\beta_v\right]}
     {w_up_u(S)}.
\end{align}
For each $p$, the infimum ranges over instances with that common
success probability; all graph notation inside the minimum refers to
the chosen instance.  The expectation is over the assignment outcomes.
The condition $p_u(S)>0$ excludes the vacuous empty-set constraint.
The worst case is taken before the limit, so the instance size may vary
with $p$.

\subsection{Proof Overview}

Fix $f\in\cF_0$.  The main difficulty is the adversarial minimization
in~\eqref{def:cL-f-2}: the adversary may choose an initial load on the
target vertex, arbitrary competing neighborhoods for the selected
arrivals, and arbitrary pre-loads on those competing vertices.  The
resulting state space is therefore much larger than the three scalar
parameters appearing in Theorem~\ref{thm:structure}.

The reduction has two steps.  First, after projection onto the target
vertex, every generic selected arrival can be replaced by a randomized
mixture of two canonical effects: a \emph{Type~I} arrival, which loads
the target when it is available, and a \emph{Type~II} arrival, which is
diverted to a competing vertex and leaves the target load unchanged.
Second, among canonical sequences, exchanging an adjacent Type~II--Type~I
pair never increases the adversary's objective.  Repeated exchanges
therefore place all Type~I arrivals before all Type~II arrivals, yielding
the load-then-divert normal form.

\subsection{Reduction to Two Canonical Arrival Types}
\label{subsec:canonical-arrival-types}

We first fix a target offline vertex $u^*$ and analyze the dual
constraint associated with a nonempty set $S\subseteq N_{u^*}$.  Two
preliminary simplifications are useful.

First, it is enough to consider $|S|\le\lceil1/p\rceil$.  If
$|S|>\lceil1/p\rceil$, the denominator in~\eqref{def:cL-f-2} is already
saturated at $w^*$, while deleting arrivals from $S$ can only decrease
$\sum_{v\in S}\beta_v$.  Thus a subset of cardinality
$\lceil1/p\rceil$ gives a ratio no larger.  In the limit $p\downarrow0$,
this corresponds to a total selected load
\[
0<\psi\le1.
\]
Second, after the last arrival in $S$, the instance may be truncated:
subsequent arrivals do not change any $\beta_v$ with $v\in S$ and can
only increase $\alp_{u^*}$.  Hence no later arrival is useful to the
adversary for the dual constraint under consideration.

Fix the target $u^*$ and write $w^*:=w_{u^*}$.  For a fixed finite
instance and success probability, the relevant inner objective is
\begin{align}\label{eqn:Lf-c}
\frac{\E\!\left[\alp_{u^*}+\sum_{v\in S}\beta_v\right]}
     {w^*p_{u^*}(S)},
\qquad
\varnothing\ne S\subseteq N_{u^*},
\quad |S|\le\lceil1/p\rceil.
\end{align}
The outer infimum over the instance and the target vertex will be
restored when we pass to the canonical parameters at the end of the
proof.

The adversary may use arrivals preceding $S$ to create an initial
\emph{pre-load} $\ell\ge0$ on $u^*$.  In the exponential-threshold
representation, the target survives this pre-load with probability
$e^{-\ell}$.  Conditional on survival, its load at the beginning of the
selected block is $\ell$.

We next analyze a single selected arrival adjacent to the target.

\xhdr{Generic setting.}
Consider the setting illustrated in Figure~\ref{fig:ger}.  The target
$u^*$ has weight $w^*$ and pre-load $\ell^*\ge0$.  An online neighbor
$v\in N_{u^*}$ arrives with
\[
N_v=\{u_i:i\in[n]\}\cup\{u^*\}.
\]
Each competing neighbor $u_i$ has weight $w_i$ and pre-load
$\ell_i\ge0$.  For the local calculation, these pre-loads are fixed and
the remaining randomness is represented by the independent exponential
thresholds.

\tikzset{
  offv/.style={
    draw,
    circle,
    ultra thick,
    minimum width=8mm,
    minimum height=8mm,
    inner sep=0pt,
    outer sep=0pt,
    font=\small
  },
  onv/.style={
    draw,
    circle,
    ultra thick,
    minimum width=8mm,
    minimum height=8mm,
    inner sep=0pt,
    outer sep=0pt,
    font=\small
  },
  myedge/.style={very thick},
  gapedge/.style={dashed, thick}
}

\begin{figure}[ht!]
    \centering
    \begin{subfigure}[t]{0.45\textwidth}
        \centering
        \vspace{0pt}
        \begin{tikzpicture}[x=1cm,y=1cm]
            \node[offv] (ustar) at (0,0) {$u^*$};
            \node[offv] (u1)    at (0,-1) {$u_1$};
            \node[offv] (u2)    at (0,-2) {$u_2$};
            \node              (dots)  at (0,-3) {$\vdots$};
            \node[offv] (un)    at (0,-4) {$u_n$};

            \node[left=4mm of ustar] {$\ell^*$};
            \node[left=4mm of u1]    {$\ell_1$};
            \node[left=4mm of u2]    {$\ell_2$};
            \node[left=4mm of un]    {$\ell_n$};

            \node[onv] (v) at (3,0) {$v$};

            \draw[myedge]  (ustar) -- (v);
            \draw[myedge]  (u1)    -- (v);
            \draw[myedge]  (u2)    -- (v);
            \draw[gapedge] (dots.east) -- (v);
            \draw[myedge]  (un)    -- (v);

            \node at (0,-4.55) {\phantom{X}};
        \end{tikzpicture}
        \caption{A generic arrival $v$ adjacent to the target $u^*$ and to competing offline neighbors with adversarially chosen pre-loads.}
        \label{fig:ger}
    \end{subfigure}
    \hfill
    \begin{subfigure}[t]{0.45\textwidth}
        \centering
        \vspace{0pt}
        \begin{tikzpicture}[x=1cm,y=1cm]
            \node[offv] (ustar1) at (0,0) {$u^*$};
            \node[onv]  (v1)     at (3,0) {$v$};
            \draw[myedge] (ustar1) -- (v1);
            \node[left=4mm of ustar1] {$\ell^*$};
            \node[right=6mm of v1] {\textbf{Type I}};

            \node[offv] (ustar2) at (0,-2.2) {$u^*$};
            \node[offv] (utilde) at (0,-3.2) {$\tilde u$};
            \node[onv]  (v2)     at (3,-2.2) {$v$};

            \draw[myedge] (ustar2) -- (v2);
            \draw[myedge] (utilde) -- (v2);

            \node[left=4mm of ustar2] {$\ell^*$};
            \node[left=5mm of utilde] {$\tilde\ell$};
            \node[right=6mm of v2] {\textbf{Type II}};

            \node at (0,-4.55) {\phantom{X}};
        \end{tikzpicture}
        \caption{The two canonical local effects: assignment to $u^*$ (Type~I) and diversion to a competing vertex $\tilde u$ (Type~II).}
        \label{fig:ws-a2}
    \end{subfigure}
    \caption{Local reduction at a target vertex $u^*$ with pre-load $\ell^*$.  Lemma~\ref{lem:ws-a} bounds the expected contribution of a generic arrival by the smaller of the Type~I and Type~II canonical effects.}
    \label{fig:ws-a}
\end{figure}

\begin{lemma}[Local canonical lower bound]\label{lem:ws-a}
Let $u^*$ be a target offline vertex of weight $w^*$ and pre-load
$\ell^*\ge0$.  For every arriving online neighbor $v\in N_{u^*}$,
\begin{align}\label{ineq:ws-a}
\E[\Delta\alp_{u^*}+\beta_v]
\ge
w^*p\min\{e^{-\ell^*},\,1-f(\ell^*)\}.
\end{align}
\end{lemma}

\begin{proof}
Let
\begin{align}\label{ineq:9-a}
L:=\max_{i\in[n]:\,\ell_i<\Theta_i}
   w_ip\bigl(1-f(\ell_i)\bigr),
\end{align}
with the convention $L=0$ if no competing neighbor is available.
Define
\[
E_1:=\{\Theta_{u^*}>\ell^*\},
\qquad
E_2:=\left\{L<w^*p\bigl(1-f(\ell^*)\bigr)\right\}.
\]
Then
\[
q_1:=\Pr[E_1]=e^{-\ell^*},
\qquad
q_2:=\Pr[E_2].
\]
The event $E_1$ depends only on $\Theta_{u^*}$, whereas $E_2$ depends
only on the thresholds of the competing neighbors and their fixed
pre-loads.  Hence $E_1$ and $E_2$ are independent.

On $E_1\cap E_2$, the target is available and has strictly larger
score than every available competitor.  Therefore $v$ is assigned to
$u^*$ and
\[
\Delta\alp_{u^*}+\beta_v=w^*p.
\]
On $\neg E_2$, some available competitor has score at least
$w^*p(1-f(\ell^*))$.  If $v$ is assigned to a competitor, then its
online dual contribution is at least this quantity.  If instead the
fixed tie-breaking rule assigns $v$ to $u^*$, the total contribution is
$w^*p$, which is no smaller because $f\in[0,1]$.  Thus in either case
\[
\Delta\alp_{u^*}+\beta_v
\ge w^*p\bigl(1-f(\ell^*)\bigr)
\qquad\text{on }\neg E_2.
\]
Consequently,
\begin{align*}
\E[\Delta\alp_{u^*}+\beta_v]
&\ge
\Pr[E_1\cap E_2]w^*p
+
\Pr[\neg E_2]w^*p\bigl(1-f(\ell^*)\bigr)\\
&=w^*p\Bigl(q_1q_2+(1-q_2)\bigl(1-f(\ell^*)\bigr)\Bigr)\\
&\ge
w^*p\min\{q_1,1-f(\ell^*)\}\\
&=w^*p\min\{e^{-\ell^*},1-f(\ell^*)\},
\end{align*}
which proves~\eqref{ineq:ws-a}.
\end{proof}

\subsubsection{Two canonical arrival types}

Lemma~\ref{lem:ws-a} identifies two canonical local effects, each of
which is attainable exactly or arbitrarily closely.

\begin{itemize}
\item \tbf{Type I.}
The arrival $v$ has $u^*$ as its only offline neighbor.  It is assigned
to $u^*$ precisely when the target is available, and therefore
\[
\E[\Delta\alp_{u^*}+\beta_v]
=w^*pe^{-\ell^*}.
\]
Whenever the assignment occurs, the target load increases by $p$.

\item \tbf{Type II.}
The arrival is diverted to a competing offline vertex and does not
change the target load.  The ideal Type~II contribution is
\[
w^*p\bigl(1-f(\ell^*)\bigr).
\]
To implement it under an arbitrary fixed tie-breaking rule, suppose
first that $1-f(0)>0$.  Introduce a fresh zero-load competitor
$\tilde u$, which is available with probability one, and choose its
weight so that
\[
\tilde w\bigl(1-f(0)\bigr)
=w^*\bigl(1-f(\ell^*)\bigr)+\varepsilon
\]
for an arbitrarily small $\varepsilon>0$.  The competitor then has
strictly larger score than the target whenever the latter is available,
so $v$ is diverted and contributes
$p[w^*(1-f(\ell^*))+\varepsilon]$.  Letting $\varepsilon\downarrow0$
recovers the Type~II value.  If $1-f(0)=0$, monotonicity and the range
of $f$ imply $f\equiv1$, so the Type~II value is zero and is obtained
by leaving $v$ unassigned.
\end{itemize}

Type~I therefore affects both the current contribution and the future
state of the target, whereas Type~II contributes at the current load
without changing that load.

\xhdr{Local tightness versus global optimality.}
Local tightness does not imply that the adversary should always choose
the smaller of the two immediate contributions.  Even when
\[
w^*pe^{-\ell^*}
\ge
w^*p\bigl(1-f(\ell^*)\bigr),
\]
a Type~I arrival may be globally preferable because the additional load
reduces the target's future availability.  Type~II, in contrast,
preserves the current load.  The resulting intertemporal tradeoff is
what leads to the load-then-divert structure below.

\subsubsection{Canonical domination and adversarial mixing}

The preceding local calculation admits the following stronger
formulation.  After projection onto the target state, a generic arrival
can be replaced by a randomized mixture of the two canonical effects
without increasing the adversary's objective.

\begin{lemma}[Canonical domination of a single arrival]
\label{lem:10-a}
Fix the pre-load $\ell^*$ of the target vertex $u^*$.  For any
adversarial configuration of a single arrival $v\in N_{u^*}$, there
exists a randomized canonical arrival supported on Type~I and Type~II
such that, after projection onto $u^*$,
\begin{enumerate}
    \item the load increment of $u^*$ has the same distribution as in the
    original configuration; and
    \item the expected contribution to $\Delta\alp_{u^*}+\beta_v$ is no
    larger than in the original configuration.
\end{enumerate}
Consequently, in minimizing the objective in~\eqref{eqn:Lf-c}, it is
sufficient to consider strategies that randomize only over Type~I and
Type~II arrivals.
\end{lemma}

\begin{proof}
Consider the generic configuration in Figure~\ref{fig:ger}. Recall that
\[
L
=
\max_{i\in[n]:\,\ell_i<\Theta_i}
w_i\,p\,\bigl(1-f(\ell_i)\bigr)
\]
is the maximum score among available offline neighbors other than
$u^*$, with $L=0$ if no competing neighbor is available.  Let
\[
s^*
:=
w^*\,p\,\bigl(1-f(\ell^*)\bigr)
\]
denote the score of $u^*$ whenever it is available.

Define $E_2$ to be the event that, given the realized states of all
competing neighbors, the fixed deterministic tie-breaking rule would
select $u^*$ if $u^*$ were available.  Equivalently, on $E_2$ no
available competitor has score larger than $s^*$, and no competitor
with score exactly $s^*$ has tie-breaking priority over $u^*$.
Thus $E_2$ depends only on the competing neighbors and is independent
of the availability event
\[
E_1:=\{\Theta_{u^*}>\ell^*\}.
\]
Write
\[
q_1:=\Pr[E_1]=e^{-\ell^*},
\qquad
q_2:=\Pr[E_2].
\]

Construct an independent randomized canonical arrival as follows:
choose Type~I with probability $q_2$, and Type~II with probability
$1-q_2$.

We first compare the induced load increment on $u^*$.  In the original
configuration, $v$ is assigned to $u^*$ precisely on
$E_1\cap E_2$.  Hence the target load increases by $p$ with probability
\[
\Pr[E_1\cap E_2]=q_1q_2.
\]
For the randomized canonical arrival, the target load increases by
$p$ precisely when Type~I is selected and $u^*$ is available.  Since
the canonical type is sampled independently of $\Theta_{u^*}$, this
also occurs with probability $q_2q_1$.  In both cases the load
increment takes values only in $\{0,p\}$, so the two load-increment
distributions coincide.

It remains to compare the expected dual contribution.  On
$E_1\cap E_2$, the original arrival is assigned to $u^*$ and hence
\[
\Delta\alp_{u^*}+\beta_v=w^*p.
\]
On $E_2\cap\neg E_1$, the contribution of the original arrival is
nonnegative.  Therefore the contribution from the event $E_2$ is at
least
\[
q_1q_2\,w^*p,
\]
which is exactly the expected contribution of the Type~I part of the
canonical mixture.

On $\neg E_2$, some available competing neighbor either has score
strictly larger than $s^*$ or has score $s^*$ and higher priority than
$u^*$.  Thus the algorithm assigns $v$ to a competing neighbor whose
score is at least $s^*$, regardless of whether $u^*$ is available.
Consequently,
\[
\Delta\alp_{u^*}=0,
\qquad
\beta_v\ge s^*
=
w^*\,p\,\bigl(1-f(\ell^*)\bigr)
\qquad\text{on }\neg E_2.
\]
Hence the expected contribution from $\neg E_2$ is at least
\[
(1-q_2)\,
w^*\,p\,\bigl(1-f(\ell^*)\bigr),
\]
which is the contribution of the Type~II part of the canonical mixture.

Combining the two events,
\[
\begin{aligned}
\E[\Delta\alp_{u^*}+\beta_v]
&\ge
q_2q_1\,w^*p
+
(1-q_2)\,w^*p\bigl(1-f(\ell^*)\bigr),
\end{aligned}
\]
which is exactly the expected contribution of the randomized canonical
mixture.  Thus the mixture preserves the distribution of the load
increment on $u^*$ while weakly decreasing the expected contribution
to $\Delta\alp_{u^*}+\beta_v$.

The ideal Type~II effect may be implemented arbitrarily closely by the
strict-score construction described above; letting its score slack tend
to zero does not change the infimum defining the adversarial objective.
Therefore the replacement is valid for the minimization under
consideration, proving the lemma.
\end{proof}

\subsection{Load-Then-Divert Structure of the Optimal Adversary}
\label{subsec:prefix-adversary}

We next characterize the order in which the two canonical types need to
appear.  Fix the target $u^*$ and the cardinality $|S|$.  After the
canonical reduction, the adversary minimizes
\[
\E\!\left[\alp_{u^*}+\sum_{v\in S}\beta_v\right]
\]
over sequences of Type~I and Type~II effects.

Lemma~\ref{lem:10-a} may produce a randomized canonical strategy, but
randomization is unnecessary for minimizing an expectation: the value
of the randomized strategy is an average of the values of deterministic
canonical type sequences, so at least one sequence in its support has
no larger value.  We may therefore restrict attention to deterministic
canonical strategies.

Such a strategy is specified by
\begin{enumerate}
    \item an initial pre-load $\ell\ge0$ on $u^*$ before the first
    arrival in $S$; and
    \item a binary vector $\bfq\in\{0,1\}^{|S|}$, where $q_i=1$
    represents Type~I and $q_i=0$ represents Type~II.
\end{enumerate}
We denote the corresponding strategy by $\pi(\ell,\bfq)$.

Let $\cQ\subseteq\{0,1\}^{|S|}$ be the set of prefix vectors
\[
\bfq=
(\underbrace{1,\ldots,1}_{m},
 \underbrace{0,\ldots,0}_{|S|-m}),
\qquad 0\le m\le|S|.
\]
Thus $\cQ$ consists precisely of sequences in which every Type~I
arrival precedes every Type~II arrival.

\begin{lemma}[Load-then-divert normal form]
\label{lem:12-a}
Fix the target vertex $u^*$ and the cardinality $|S|$.  Let
$\kappa^*$ denote the infimum value of
\[
\E\!\left[\alp_{u^*}+\sum_{v\in S}\beta_v\right]
\]
achievable by the adversary over all configurations with these fixed
choices.  For $\ell\ge0$ and
$\bfq\in\{0,1\}^{|S|}$, let $\kappa(\ell,\bfq)$ denote the value of
the same objective under the deterministic canonical strategy
$\pi(\ell,\bfq)$.  Then
\[
\kappa^*
=
\inf_{\ell\ge0,\ \bfq\in\{0,1\}^{|S|}}
\kappa(\ell,\bfq)
=
\inf_{\ell\ge0,\ \bfq\in\cQ}
\kappa(\ell,\bfq).
\]
In particular, for every $\varepsilon>0$, there exists a
load-then-divert canonical strategy whose value is at most
$\kappa^*+\varepsilon$.
\end{lemma}

\begin{proof}
We first eliminate randomization among the canonical arrival types.
By Lemma~\ref{lem:10-a}, any adversarial strategy can be replaced,
without increasing its objective value, by a strategy whose arrivals
are randomized mixtures of Type~I and Type~II effects.  Condition on
all random choices used in these canonical replacements.  For a fixed
initial pre-load $\ell$, each realization determines a deterministic
type sequence
\[
\bfq\in\{0,1\}^{|S|},
\]
and hence a deterministic canonical strategy $\pi(\ell,\bfq)$.
Since the objective is linear in the distribution over these
realizations, the value of the randomized strategy is a convex
combination of the corresponding values $\kappa(\ell,\bfq)$.
Therefore at least one deterministic realization has value no larger
than the randomized one.  Consequently,
\[
\kappa^*
=
\inf_{\ell\ge0,\ \bfq\in\{0,1\}^{|S|}}
\kappa(\ell,\bfq).
\]

It remains to show that, for every fixed $\ell$, the type sequence may
be chosen from $\cQ$.  Consider any
$\bfq\notin\cQ$.  Then $\bfq$ contains an adjacent inversion $(0,1)$:
for some $i$,
\begin{align}
\bfq
&=
(*,\ldots,*,
\underbrace{0}_{i},
\underbrace{1}_{i+1},
*,\ldots,*),
\label{str:10-a}\\
\bfq'
&=
(*,\ldots,*,
\underbrace{1}_{i},
\underbrace{0}_{i+1},
*,\ldots,*),
\label{str:10-b}
\end{align}
where $\bfq'$ is obtained by swapping the entries in positions
$i$ and $i+1$.

Let $r$ be the target pre-load immediately before this adjacent pair.
Under $\pi(\ell,\bfq)$, the Type~II arrival occurs first.  It contributes
\[
w^*p\bigl(1-f(r)\bigr)
\]
and leaves the target pre-load unchanged.  The following Type~I arrival
therefore sees the same pre-load $r$ and contributes, in expectation,
\[
w^*p e^{-r}.
\]
Thus the expected contribution of the pair is
\[
w^*p\bigl(1-f(r)\bigr)+w^*p e^{-r}.
\]

Under $\pi(\ell,\bfq')$, the Type~I arrival occurs first.  Its expected
contribution is again
\[
w^*p e^{-r},
\]
and the canonical load parameter subsequently increases from $r$ to
$r+p$.  Hence the following Type~II arrival contributes
\[
w^*p\bigl(1-f(r+p)\bigr).
\]
The total expected contribution of the swapped pair is therefore
\[
w^*p e^{-r}
+
w^*p\bigl(1-f(r+p)\bigr).
\]
Since $f$ is non-decreasing,
\[
1-f(r+p)\le1-f(r),
\]
and consequently
\[
w^*p e^{-r}
+
w^*p\bigl(1-f(r+p)\bigr)
\le
w^*p e^{-r}
+
w^*p\bigl(1-f(r)\bigr).
\]
Thus swapping an adjacent inversion $(0,1)$ into $(1,0)$ cannot
increase the adversary's objective.

The swap also leaves the state seen by every arrival outside the pair
unchanged.  Indeed, Type~II does not modify the target pre-load, while
both orderings contain exactly one Type~I arrival.  Hence the target
pre-load immediately after the pair is the same under $\bfq$ and
$\bfq'$, and all subsequent canonical contributions are identical.
Therefore
\[
\kappa(\ell,\bfq')
\le
\kappa(\ell,\bfq).
\]

Repeatedly applying this adjacent exchange removes every inversion
of the form $(0,1)$.  Since each swap weakly decreases the objective,
the procedure terminates at a prefix vector
\[
(1,\ldots,1,0,\ldots,0)\in\cQ
\]
whose value is no larger than that of the original sequence.  Hence,
for every fixed $\ell\ge0$,
\[
\min_{\bfq\in\{0,1\}^{|S|}}\kappa(\ell,\bfq)
=
\min_{\bfq\in\cQ}\kappa(\ell,\bfq).
\]
Taking the infimum over $\ell$ gives
\[
\inf_{\ell\ge0,\ \bfq\in\{0,1\}^{|S|}}
\kappa(\ell,\bfq)
=
\inf_{\ell\ge0,\ \bfq\in\cQ}
\kappa(\ell,\bfq),
\]
which proves the lemma.
\end{proof}

\subsection{Proof of Theorem~\ref{thm:structure}}
\label{subsec:proof-structure}

We now pass to the vanishing-probability limit.  Let
$\psi\in(0,1]$ denote the limiting total load of the selected arrivals.
By Lemma~\ref{lem:12-a}, it is sufficient to consider canonical
load-then-divert strategies.  Such a strategy is described by an
initial pre-load $\ell\ge0$, followed by a Type~I block of total load
$\tpsi\in[0,\psi]$, and then a Type~II block of total load
$\psi-\tpsi$.  Denote it by $\pi(\ell,\psi,\tpsi)$.

Let $\kappa(\ell,\psi,\tpsi)$ denote the limiting value of
\[
\E\!\left[\alp_{u^*}+\sum_{v\in S}\beta_v\right]
\]
for this strategy.  Then
\begin{align}\label{eq:kappa-three-param}
\kappa(\ell,\psi,\tpsi)
&=w^*\left(
\int_0^\ell e^{-z}f(z)\,\sd z
+
\int_\ell^{\ell+\tpsi}e^{-z}\,\sd z
+
(\psi-\tpsi)\bigl(1-f(\ell+\tpsi)\bigr)
\right).
\end{align}
Indeed, the three terms have the following interpretations.

First, during the preprocessing that creates the initial load $\ell$,
an infinitesimal load increment at level $z$ is realized only if the
target survives to $z$, which occurs with probability $e^{-z}$; its
offline-dual contribution is therefore
$w^*e^{-z}f(z)\,\sd z$.  Integrating over $z\in[0,\ell]$ gives the
first term.

Second, throughout the Type~I block, an infinitesimal assignment at
load $z$ occurs precisely when the target survives to $z$.  Whenever it
occurs, the offline and online increments sum to $w^*\,\sd z$.
Consequently, the expected contribution while the target load moves
from $\ell$ to $\ell+\tpsi$ is
\[
w^*\int_\ell^{\ell+\tpsi}e^{-z}\,\sd z.
\]

Finally, the Type~II block leaves the target load fixed at
$\ell+\tpsi$.  Each unit of selected load in this block contributes,
in the limiting construction,
$w^*(1-f(\ell+\tpsi))$.  Since the block has total load
$\psi-\tpsi$, this gives the last term in
\eqref{eq:kappa-three-param}.

The canonical-domination lemma and the load-then-divert lemma show that
no general adversarial configuration yields a smaller limiting value
than the infimum over these canonical parameters, while the canonical
constructions above realize every such parameter triple arbitrarily
closely.  Dividing~\eqref{eq:kappa-three-param} by the limiting
denominator $w^*\psi$, restoring the outer infimum in
\eqref{def:cL-f-2}, and taking the infimum over
\[
\ell\ge0,
\qquad
0<\psi\le1,
\qquad
0\le\tpsi\le\psi,
\]
gives
\[
\cL[f]
=
\inf_{\substack{\ell\ge0,\;0<\psi\le1\\0\le\tpsi\le\psi}}
\frac{
\displaystyle\int_0^\ell e^{-z}f(z)\,\sd z
+
\displaystyle\int_\ell^{\ell+\tpsi}e^{-z}\,\sd z
+
(\psi-\tpsi)\bigl(1-f(\ell+\tpsi)\bigr)
}{\psi},
\]
which is exactly~\eqref{eqn:Lf-d-1}.  This completes the proof of
Theorem~\ref{thm:structure}.


\section{Proof of Proposition~\ref{pro:2-17-a}}
\label{app:17-b}

The proof has two steps.  First, Lemma~\ref{lem:2-16-a} shows that, for any fixed $\ell\ge0$ and $\tpsi\in[0,1]$, the minimization in~\eqref{eqn:Lf-d-1} is attained at $\psi=1$.  Second, we show that after this reduction the remaining minimum is attained either in the region $\ell+\tpsi\le1$ or at the point $\ell=\tpsi=1$.

\begin{lemma}\label{lem:2-16-a}
For any $f \in \cF_2$ defined in~\eqref{eq:def-F2}, the value of $\cL[f]$ in~\eqref{eqn:Lf-d-1} can be simplified as
\begin{align}
\cL[f]
=
\min_{\ell \geq 0,\ \tpsi \in [0,1]}
\left\{
\int_0^\ell e^{-z} f(z)\,\sd z
+
\int_\ell^{\ell+\tpsi} e^{-z}\,\sd z
+
(1-\tpsi)\bigl(1-f(\ell+\tpsi)\bigr)
\right\}.
\end{align}
\end{lemma}

We defer the proof of Lemma~\ref{lem:2-16-a} to Appendix~\ref{app:17-a}.  We now prove Proposition~\ref{pro:2-17-a} assuming the lemma.

\begin{proof}[\textbf{Proof of Proposition~\ref{pro:2-17-a}}.]
By Lemma~\ref{lem:2-16-a}, for every $f\in\cF_2$ the minimization in~\eqref{eqn:Lf-d-1} reduces to the case $\psi=1$.  Define
\begin{align}\label{eqn:18-a}
g(\ell,\tpsi)
:=
\int_0^\ell e^{-z} f(z)\,\sd z
+
\int_\ell^{\ell+\tpsi} e^{-z}\,\sd z
+
(1-\tpsi)\bigl(1-f(\ell+\tpsi)\bigr).
\end{align}
Then
\[
\cL[f]
=
\min_{\ell \ge 0,\;\tpsi \in [0,1]} g(\ell,\tpsi)
=
\min\biggl\{
     \underbrace{\min_{\ell,\,\tpsi\in[0,1],\,\ell+\tpsi\le 1} g(\ell,\tpsi)}_{=W_2},
     \;\;
     \underbrace{\min_{\ell\ge 0,\,\tpsi\in[0,1],\,\ell+\tpsi\ge 1} g(\ell,\tpsi)}_{\text{to be shown equal to }W_1}
   \biggr\}.
\]

\medskip
\noindent\textbf{Step 1: Optimize over $\tpsi$ when $\ell+\tpsi \ge 1$.}
If $\ell+\tpsi\ge1$, then $f(\ell+\tpsi)=1-1/e$ by the definition of $\cF_2$.  Hence
\begin{align*}
g(\ell,\tpsi)
&=
\int_0^\ell e^{-z} f(z)\,\sd z
+
\int_\ell^{\ell+\tpsi} e^{-z}\,\sd z
+
(1-\tpsi)\,e^{-1},\\
\frac{\partial g(\ell,\tpsi)}{\partial \tpsi}
&=
e^{-(\ell+\tpsi)}-e^{-1}
\le 0.
\end{align*}
Thus, for any fixed $\ell$, the function $\tpsi\mapsto g(\ell,\tpsi)$ is non-increasing on the feasible interval
\[
\{\tpsi\in[0,1]:\ell+\tpsi\ge1\}
=
[\max\{0,1-\ell\},1].
\]
Consequently, the minimum over this interval is attained at $\tpsi=1$, and
\[
g(\ell,\tpsi)\ge g(\ell,1).
\]

\medskip
\noindent\textbf{Step 2: Optimize $g(\ell,1)$ over $\ell \ge 0$.}
We have
\[
g(\ell,1)
=
\int_0^\ell e^{-z} f(z)\,\sd z
+
e^{-\ell}(1-e^{-1}).
\]
Differentiating with respect to $\ell$ gives
\[
\frac{\partial g(\ell,1)}{\partial \ell}
=
e^{-\ell}\bigl(f(\ell)-(1-1/e)\bigr).
\]
For $\ell\in[0,1]$, monotonicity and the constraint $f(1)\le1-1/e$ imply
\[
f(\ell)\le1-1/e.
\]
For $\ell\ge1$, the definition of $\cF_2$ gives \(f(\ell)=1-1/e\).  Hence $g(\ell,1)$ is non-increasing on $[0,1]$ and constant on $[1,\infty)$.  Therefore
\[
\min_{\ell \ge 0} g(\ell,1)=g(1,1).
\]

\medskip
\noindent\textbf{Step 3: Conclude.}
Combining the previous two steps,
\[
\min_{\ell \ge 0,\, \tpsi \in [0,1],\, \ell+\tpsi \ge 1} g(\ell,\tpsi)
=
g(1,1)
=
W_1.
\]
Together with the region $\ell+\tpsi\le1$, this gives
\[
\cL[f]=\min\{W_1,W_2\},
\]
as claimed.
\end{proof}

\subsection{Proof of Lemma~\ref{lem:2-16-a}}
\label{app:17-a}

Before proving Lemma~\ref{lem:2-16-a}, we record the following elementary criterion.

\begin{proposition}\label{pro:2-16-b}  
Let $g(\psi)\ge0$ be a non-decreasing function on $\psi\in[0,1]$ with a continuous first derivative.  Then $g(\psi)/\psi$ is non-increasing on $(0,1]$ if and only if
\[
\psi\,\dot g(\psi) \le g(\psi),
\qquad
\forall \psi\in(0,1].
\]
\end{proposition}  

\begin{proof}
For $\psi\in(0,1]$,
\[
\frac{\sd}{\sd\psi}\left(\frac{g(\psi)}{\psi}\right)
=
\frac{\psi\,\dot g(\psi)-g(\psi)}{\psi^2}.
\]
Since $\psi^2>0$, the derivative is non-positive if and only if $\psi\,\dot g(\psi)\le g(\psi)$.
\end{proof}

\begin{proof}[\textbf{Proof of Lemma~\ref{lem:2-16-a}}.]
We show that for every fixed $\ell\ge0$ and $\tpsi\in[0,1]$, the ratio in~\eqref{eqn:Lf-d-1}, viewed as a function of $\psi$, is minimized at $\psi=1$.  This allows the minimization to be restricted to $\psi=1$.

Fix $\ell\ge0$ and $\tpsi\in[0,1]$, and define
\[
g(\psi)
:=
\int_0^\ell e^{-z} f(z)\,\sd z
+
\int_\ell^{\ell+\tpsi} e^{-z}\,\sd z
+
(\psi-\tpsi)\bigl(1-f(\ell+\tpsi)\bigr),
\qquad
\psi\in[\tpsi,1].
\]
Then $g(\psi)\ge0$ and
\[
\dot g(\psi)=1-f(\ell+\tpsi)\ge0,
\]
so $g$ is non-decreasing on $[\tpsi,1]$.  By Proposition~\ref{pro:2-16-b}, applied on this interval, it suffices to verify
\[
\psi\,\dot g(\psi)\le g(\psi),
\qquad
\forall \psi\in[\tpsi,1]\cap(0,1].
\]
Since $\dot g(\psi)=1-f(\ell+\tpsi)$, this inequality is equivalent to
\begin{align}
\int_0^\ell e^{-z} f(z)\,\sd z
+
\int_\ell^{\ell+\tpsi} e^{-z}\,\sd z
-
\tpsi\bigl(1-f(\ell+\tpsi)\bigr)
\ge0.
\label{eqn:2-16-a}
\end{align}

We verify~\eqref{eqn:2-16-a} by considering three cases.

\xhdr{Case 1.} $\ell+\tpsi\in[0,1]$.

By~\eqref{eq:def-F2}, we have $1-f(\ell+\tpsi)\le e^{-(\ell+\tpsi)}$.  Hence
\begin{align*}
&\int_0^\ell e^{-z} f(z)\,\sd z
+
\int_\ell^{\ell+\tpsi} e^{-z}\,\sd z
-
\tpsi\bigl(1-f(\ell+\tpsi)\bigr) \\
&\ge
\int_0^\ell e^{-z} f(z)\,\sd z
+
\int_\ell^{\ell+\tpsi}
\bigl(e^{-z}-e^{-(\ell+\tpsi)}\bigr)\,\sd z
\ge0.
\end{align*}

\xhdr{Case 2.} $\ell\ge1$.

In this case, $f(\ell+\tpsi)=1-1/e$, and hence $1-f(\ell+\tpsi)=e^{-1}$.  We first claim that
\begin{align}
&\int_0^\ell e^{-z} f(z)\,\sd z
+
\int_\ell^{\ell+\tpsi} e^{-z}\,\sd z
-
\tpsi e^{-1}
\ge
\int_0^1 e^{-z} f(z)\,\sd z
+
\int_1^{1+\tpsi} e^{-z}\,\sd z
-
\tpsi e^{-1}.
\label{ineq:2-16-c}
\end{align}
To see this, define
\[
h(\ell)
:=
\int_0^\ell e^{-z} f(z)\,\sd z
+
\int_\ell^{\ell+\tpsi} e^{-z}\,\sd z
-
\tpsi e^{-1}.
\]
For $\ell\ge1$, $f(\ell)=1-1/e$, and therefore
\[
\dot h(\ell)
=
e^{-\ell}f(\ell)+e^{-(\ell+\tpsi)}-e^{-\ell}
=
-e^{-(\ell+1)}+e^{-(\ell+\tpsi)}
\ge0.
\]
Thus $h$ is non-decreasing on $[1,\infty)$, proving~\eqref{ineq:2-16-c}.

Using $f(z)\ge1-e^{-z}$ for $z\in[0,1]$, the right-hand side of~\eqref{ineq:2-16-c} is at least
\begin{align*}
&\int_0^1 e^{-z}(1-e^{-z})\,\sd z
+
e^{-1}-e^{-(1+\tpsi)}
-
\tpsi e^{-1} \\
&=
\frac12(1+e^{-2})
-
\bigl(e^{-(1+\tpsi)}+\tpsi e^{-1}\bigr) \ge
\frac12(1+e^{-2})-(e^{-2}+e^{-1})
>0,
\end{align*}
where the last inequality uses $\tpsi\in[0,1]$.

\xhdr{Case 3.} $\ell\in[0,1]$ and $\tpsi\in[1-\ell,1]$.

In this regime, $\ell+\tpsi\ge1$, so $f(\ell+\tpsi)=1-1/e$.  Hence
\begin{align*}
&\int_0^\ell e^{-z} f(z)\,\sd z
+
\int_\ell^{\ell+\tpsi} e^{-z}\,\sd z
-
\tpsi e^{-1} \ge
\int_0^\ell e^{-z}(1-e^{-z})\,\sd z
+
\int_\ell^{\ell+\tpsi} e^{-z}\,\sd z
-
\tpsi e^{-1}.
\end{align*}
For fixed $\ell\in[0,1]$, define
\[
h(\tpsi)
:=
\int_0^\ell e^{-z}(1-e^{-z})\,\sd z
+
\int_\ell^{\ell+\tpsi} e^{-z}\,\sd z
-
\tpsi e^{-1}.
\]
Then
\[
\dot h(\tpsi)=e^{-(\ell+\tpsi)}-e^{-1}\le0
\qquad
\text{for }\tpsi\in[1-\ell,1].
\]
Thus $h$ is non-increasing on $[1-\ell,1]$, and therefore
\begin{align}
&\int_0^\ell e^{-z}(1-e^{-z})\,\sd z
+
\int_\ell^{\ell+\tpsi} e^{-z}\,\sd z
-
\tpsi e^{-1}\ge
\int_0^\ell e^{-z}(1-e^{-z})\,\sd z
+
\int_\ell^{\ell+1} e^{-z}\,\sd z
-
e^{-1}.
\label{ineq:2-16-1}
\end{align}

Now define
\[
\tilde h(\ell)
:=
\int_0^\ell e^{-z}(1-e^{-z})\,\sd z
+
\int_\ell^{\ell+1} e^{-z}\,\sd z
-
e^{-1}.
\]
A direct differentiation gives
\[
\frac{\sd \tilde h}{\sd\ell}
=
e^{-(\ell+1)}-e^{-2\ell}
\le0
\qquad
\text{for }\ell\in[0,1].
\]
Thus $\tilde h$ is non-increasing on $[0,1]$, and so
\begin{align}
\tilde h(\ell)
&\ge
\tilde h(1)
=
\int_0^1 e^{-z}(1-e^{-z})\,\sd z
+
\int_1^2 e^{-z}\,\sd z
-
e^{-1}
\nonumber\\
&=
\frac12\bigl(1-2e^{-1}-e^{-2}\bigr)
>0.
\label{ineq:2-16-2}
\end{align}
Combining~\eqref{ineq:2-16-1} and~\eqref{ineq:2-16-2} proves~\eqref{eqn:2-16-a} in Case~3.

Since~\eqref{eqn:2-16-a} holds in all cases, Proposition~\ref{pro:2-16-b} implies that $g(\psi)/\psi$ is non-increasing on $[\tpsi,1]$.  Therefore the minimum over $\psi\in[\tpsi,1]$ is attained at $\psi=1$, completing the proof.
\end{proof}


\section{Proof of Lemma~\ref{lem:20-lb}}
\label{app:20-lb}

\begin{proof}
Let $\LP(n)$ and $\LP(2n)$ denote \textbf{AUG-LP}~\eqref{lp:obj} with discretization parameters $n$ and $2n$, respectively.  Suppose $(y,\bx)$ is feasible for $\LP(n)$, where $\bx=(x_t)_{t\in(n)}$.  We verify that the pair $(\ty,\tbx)$ defined in~\eqref{def:1}, with $\tbx=(\tx_t)_{t\in(2n)}$, is feasible for $\LP(2n)$.

\medskip
\noindent
\textbf{Constraints~\eqref{lp:cons-1}--\eqref{lp:cons-2}.}
By construction, $\tbx$ is non-decreasing, $\tx_0=x_0\ge0$, and $\tx_{2n}=x_n=1-1/e$.  It remains to verify the lower envelope.  For $1\le t\le 2n-1$,
\[
\tx_t =
\begin{cases}
x_{t/2}, & t \text{ even},\\
x_{(t+1)/2}, & t \text{ odd}.
\end{cases}
\]
If $t$ is even, then
\[
\tx_t=x_{t/2}\ge 1-e^{-t/(2n)}.
\]
If $t$ is odd, then
\[
\tx_t=x_{(t+1)/2}\ge 1-e^{-(t+1)/(2n)} > 1-e^{-t/(2n)}.
\]
Thus Constraints~\eqref{lp:cons-1}--\eqref{lp:cons-2} hold for $\LP(2n)$.

\medskip
\noindent
\textbf{Constraint~\eqref{lp:cons-3}.}
Let $W_1(2n)$ and $W_1(n)$ denote the right-hand side of Constraint~\eqref{lp:cons-3} under $\tbx$ and $\bx$, respectively.  Since $(y,\bx)$ is feasible for $\LP(n)$, we have $W_1(n)\ge y$.  Moreover,
\begin{align*}
W_1(2n)
&=
\frac{1}{2n}\sum_{t=1}^{2n}\tx_t e^{-t/(2n)}
+
e^{-1}(1-1/e)=
\frac{1}{2n}\sum_{k=1}^{n} x_k
\sbp{e^{-(2k-1)/(2n)}+e^{-2k/(2n)}}
+
e^{-1}(1-1/e)\\
&\ge
\frac{1}{2n}\sum_{k=1}^{n} x_k\,2e^{-2k/(2n)}
+
e^{-1}(1-1/e)\\
&=
\frac{1}{n}\sum_{k=1}^{n}x_k e^{-k/n}
+
e^{-1}(1-1/e)
=
W_1(n)
\ge y>\ty.
\end{align*}

\medskip
\noindent
\textbf{Constraint~\eqref{lp:cons-4}.}
For each pair $i\in(2n)$ and $j\in(2n-i)$, let $W_2(i,j,2n)$ denote the right-hand side of Constraint~\eqref{lp:cons-4} in $\LP(2n)$ under $\tbx$.  Similarly, for $\hi\in(n)$ and $\hj\in(n-\hi)$, let $W_2(\hi,\hj,n)$ denote the corresponding right-hand side in $\LP(n)$ under $\bx$.  Since $(y,\bx)$ is feasible for $\LP(n)$,
\[
W_2(\hi,\hj,n)\ge y,
\qquad
\forall\,\hi\in(n),\ \hj\in(n-\hi).
\]
We verify Constraint~\eqref{lp:cons-4} for $\LP(2n)$ by considering the following cases.

\begin{enumerate}
\item[\textbf{Case 1}:] $i\ge1$, $j\ge1$, and both $i$ and $j$ are even.  Then
\begin{align*}
W_2(i,j,2n)
&=
\frac{1}{2n}\sum_{t=1}^{i}\tx_t e^{-t/(2n)}
+
\frac{1}{2n}\sum_{t=i+1}^{i+j} e^{-t/(2n)}
+
\bp{1-\frac{j}{2n}}\bp{1-\tx_{i+j}}\\
&=
\frac{1}{2n}\sum_{k=1}^{i/2}x_k
\sbp{e^{-(2k-1)/(2n)}+e^{-2k/(2n)}}+
\frac{1}{2n}\sum_{k=i/2+1}^{i/2+j/2}
\sbp{e^{-(k-1/2)/n}+e^{-k/n}}
+
\bp{1-\frac{j/2}{n}}\bp{1-x_{(i+j)/2}}\\
&\ge
\frac{1}{n}\sum_{k=1}^{i/2}x_k e^{-k/n}
+
\frac{1}{n}\sum_{k=i/2+1}^{i/2+j/2}e^{-k/n}
+
\bp{1-\frac{j/2}{n}}\bp{1-x_{(i+j)/2}}\\
&=
W_2(i/2,j/2,n)
\ge y>\ty.
\end{align*}

The remaining cases reduce to this even-grid case, up to an additive loss of at most $(1-1/e)/(2n)$.

\medskip
\item[\textbf{Case 2}:] $i\ge1$, $j\ge1$, $i$ is odd, and $j$ is even.  Then
\begingroup
\allowdisplaybreaks
\begin{align*}
W_2(i,j,2n)
&=
\frac{1}{2n}\sum_{t=1}^{i}\tx_t e^{-t/(2n)}
+
\frac{1}{2n}\sum_{t=i+1}^{i+j} e^{-t/(2n)}
+
\bp{1-\frac{j}{2n}}\bp{1-\tx_{i+j}}\\
&=
\frac{1}{2n}\sum_{t=1}^{i+1}\tx_t e^{-t/(2n)}
-
\frac{1}{2n}\,\tx_{i+1}e^{-(i+1)/(2n)}+
\frac{1}{2n}\sum_{t=i+2}^{i+1+j}e^{-t/(2n)}
+
\frac{1}{2n}\sbp{e^{-(i+1)/(2n)}-e^{-(i+j+1)/(2n)}}\\
&\quad+
\bp{1-\frac{j}{2n}}\bp{1-\tx_{i+j+1}}\\
&=
W_2(i+1,j,2n)
-
\frac{1}{2n}\,\tx_{i+1}e^{-(i+1)/(2n)}
+
\frac{1}{2n}\sbp{e^{-(i+1)/(2n)}-e^{-(i+j+1)/(2n)}}\\
&\ge
y-\frac{1-1/e}{2n}
=
\ty.
\end{align*}
\endgroup
Here we used Case~1 for $W_2(i+1,j,2n)$, since both $i+1$ and $j$ are even, and the bound $\tx_{i+1}\le 1-1/e$.

\medskip
\item[\textbf{Case 3}:] $i\ge1$, $j\ge1$, and both $i$ and $j$ are odd.  Then
\begin{align*}
W_2(i,j,2n)
&=
\frac{1}{2n}\sum_{t=1}^{i}\tx_t e^{-t/(2n)}
+
\frac{1}{2n}\sum_{t=i+1}^{i+j} e^{-t/(2n)}
+
\bp{1-\frac{j}{2n}}\bp{1-\tx_{i+j}}\\
&=
\frac{1}{2n}\sum_{t=1}^{i-1}\tx_t e^{-t/(2n)}
+
\frac{1}{2n}\,\tx_i e^{-i/(2n)}+
\frac{1}{2n}\sum_{t=i}^{i+j} e^{-t/(2n)}
-
\frac{1}{2n}e^{-i/(2n)}\\
&\quad+
\bp{1-\frac{j+1}{2n}}\bp{1-\tx_{i+j}}
+
\frac{1}{2n}\bp{1-\tx_{i+j}}\\
&=
W_2(i-1,j+1,2n)
+
\frac{1}{2n}\,\tx_i e^{-i/(2n)}
-
\frac{1}{2n}e^{-i/(2n)}
+
\frac{1}{2n}\bp{1-\tx_{i+j}}\\
&=
W_2(i-1,j+1,2n)
+
\frac{1}{2n}\bp{1-\tx_{i+j}}
-
\frac{1}{2n}e^{-i/(2n)}\bp{1-\tx_i}\\
&\ge
W_2(i-1,j+1,2n)
+
\frac{1}{2n}\bp{\tx_i-\tx_{i+j}}\ge
y-\frac{1-1/e}{2n}
=
\ty.
\end{align*}
In the last line, we used Case~1 for $W_2(i-1,j+1,2n)$ and the monotonicity of $\tbx$, which gives $\tx_i-\tx_{i+j}\ge -(1-1/e)$.

\medskip
\item[\textbf{Case 4}:] $i\ge1$, $j\ge1$, $i$ is even, and $j$ is odd.  Then
\begin{align*}
W_2(i,j,2n)
&=
\frac{1}{2n}\sum_{t=1}^{i}\tx_t e^{-t/(2n)}
+
\frac{1}{2n}\sum_{t=i+1}^{i+j} e^{-t/(2n)}
+
\bp{1-\frac{j}{2n}}\bp{1-\tx_{i+j}}\\
&=
\frac{1}{2n}\sum_{t=1}^{i}\tx_t e^{-t/(2n)}
+
\frac{1}{2n}\sum_{t=i+1}^{i+j+1} e^{-t/(2n)}
-
\frac{1}{2n}e^{-(i+j+1)/(2n)}\\
&\quad+
\bp{1-\frac{j+1}{2n}}\bp{1-\tx_{i+j+1}}
+
\frac{1}{2n}\bp{1-\tx_{i+j+1}}\\
&=
W_2(i,j+1,2n)
+
\frac{1}{2n}\bp{1-\tx_{i+j+1}-e^{-(i+j+1)/(2n)}}\\
&\ge
y-\frac{1-1/e}{2n}
=
\ty.
\end{align*}
Here we used Case~1 for $W_2(i,j+1,2n)$, since both $i$ and $j+1$ are even.

\medskip
\item[\textbf{Case 5}:] $i=0$ and $j$ is even.  Then
\begin{align*}
W_2(0,j,2n)
&=
\frac{1}{2n}\sum_{t=1}^{j}e^{-t/(2n)}
+
\bp{1-\frac{j}{2n}}\bp{1-\tx_j}\\
&=
\frac{1}{2n}\sum_{k=1}^{j/2}
\sbp{e^{-(2k-1)/(2n)}+e^{-2k/(2n)}}
+
\bp{1-\frac{j/2}{n}}\bp{1-x_{j/2}}\\
&\ge
\frac{1}{n}\sum_{k=1}^{j/2}e^{-k/n}
+
\bp{1-\frac{j/2}{n}}\bp{1-x_{j/2}}\\
&=
W_2(0,j/2,n)
\ge y>\ty.
\end{align*}

\medskip
\item[\textbf{Case 6}:] $i=0$ and $j$ is odd.  Since $j\le2n-1$, we have $(j+1)/2\le n$.  Then
\begin{align*}
W_2(0,j,2n)
&=
\frac{1}{2n}\sum_{t=1}^{j}e^{-t/(2n)}
+
\bp{1-\frac{j}{2n}}\bp{1-\tx_j}\\
&=
\frac{1}{2n}\sum_{t=1}^{j+1}e^{-t/(2n)}
+
\bp{1-\frac{j+1}{2n}}\bp{1-x_{(j+1)/2}}+
\frac{1}{2n}\bp{1-x_{(j+1)/2}-e^{-(j+1)/(2n)}}\\
&\ge
W_2(0,(j+1)/2,n)
-
\frac{1-1/e}{2n}\ge
y-\frac{1-1/e}{2n}
=
\ty.
\end{align*}

\medskip
\item[\textbf{Case 7}:] $j=0$ and $i$ is even.  Then
\begin{align*}
W_2(i,0,2n)
&=
\frac{1}{2n}\sum_{t=1}^{i}\tx_t e^{-t/(2n)}
+
\bp{1-\tx_i}\ge
\frac{1}{n}\sum_{k=1}^{i/2}x_k e^{-k/n}
+
\bp{1-x_{i/2}}=
W_2(i/2,0,n)
\ge y>\ty.
\end{align*}

\medskip
\item[\textbf{Case 8}:] $j=0$ and $i$ is odd.  Then
\begin{align*}
W_2(i,0,2n)
&=
\frac{1}{2n}\sum_{t=1}^{i}\tx_t e^{-t/(2n)}
+
\bp{1-\tx_i}\\
&=
\frac{1}{2n}\sum_{t=1}^{i+1}\tx_t e^{-t/(2n)}
+
\bp{1-x_{(i+1)/2}}
-
\frac{1}{2n}\,\tx_{i+1}e^{-(i+1)/(2n)}\\
&\ge
W_2((i+1)/2,0,n)
-
\frac{1-1/e}{2n}\ge
y-\frac{1-1/e}{2n}
=
\ty.
\end{align*}
\end{enumerate}

Combining the eight cases above, Constraint~\eqref{lp:cons-4} holds for every pair $(i,j)$ in $\LP(2n)$.  Hence $(\ty,\tbx)$ is feasible for $\LP(2n)$.
\end{proof}

\section{Proof of Lemma~\ref{lem:20-ub}}
\label{app:20-ub}

\begin{proof}
Let $\LP(n)$ and $\LP(2n)$ denote \textbf{AUG-LP}~\eqref{lp:obj} with discretization parameters $n$ and $2n$, respectively.
Suppose $(\ty,\tbx)$ is feasible for $\LP(2n)$, where $\tbx=(\tx_t)_{t\in(2n)}$.
We verify the feasibility of the solution $(y,\bx)$ in $\LP(n)$, where $\bx=(x_t)_{t\in(n)}$, as defined in~\eqref{def:2}.

\medskip
\noindent
\textbf{Constraints~\eqref{lp:cons-1}--\eqref{lp:cons-2}.}
By definition~\eqref{def:2}, the vector $\bx$ is non-decreasing, with $x_0=\tx_0\ge 0$ and $x_n=\tx_{2n}=1-1/e$.  It remains to show that $x_t\ge 1-e^{-t/n}$ for $1\le t\le n-1$.  Indeed,
\[
x_t
=
\frac{\tx_{2t}+\tx_{2t+1}}{2}
\ge
\frac12\sbp{1-e^{-2t/(2n)}+1-e^{-(2t+1)/(2n)}}
>
1-e^{-t/n}.
\]
Thus Constraints~\eqref{lp:cons-1}--\eqref{lp:cons-2} hold for $\LP(n)$.

\medskip
\noindent
\textbf{Constraint~\eqref{lp:cons-3}.}
Let $W_1(2n)$ and $W_1(n)$ denote the right-hand side of Constraint~\eqref{lp:cons-3} in $\LP(2n)$ under $\tbx$ and in $\LP(n)$ under $\bx$, respectively.  Since $(\ty,\tbx)$ is feasible for $\LP(2n)$, we have $W_1(2n)\ge \ty$.  We compute
\begingroup
\allowdisplaybreaks
\begin{align*}
W_1(n)
&= \frac{1}{n}\sum_{t=1}^{n}x_t e^{-t/n}+e^{-1}(1-1/e)\\
&= \frac{1}{n}\sum_{t=1}^{n-1}\frac{\tx_{2t}+\tx_{2t+1}}{2}e^{-t/n}
+\frac{1}{n}\tx_{2n}e^{-1}+e^{-1}(1-1/e)\\
&= \frac{1}{2n}\tx_1e^{-1/(2n)}
+\frac{1}{2n}\sum_{t=1}^{n-1}\sbp{\tx_{2t}e^{-2t/(2n)}+\tx_{2t+1}e^{-2t/(2n)}}
+\frac{1}{2n}\tx_{2n}e^{-1}\\
&\quad
+\frac{1}{2n}\tx_{2n}e^{-1}+e^{-1}(1-1/e)
-\frac{1}{2n}\tx_1e^{-1/(2n)}\\
&\ge
\frac{1}{2n}\tx_1e^{-1/(2n)}
+\frac{1}{2n}\sum_{t=1}^{n-1}\sbp{\tx_{2t}e^{-2t/(2n)}+\tx_{2t+1}e^{-(2t+1)/(2n)}}
+\frac{1}{2n}\tx_{2n}e^{-1}\\
&\quad
+\frac{1}{2n}\tx_{2n}e^{-1}+e^{-1}(1-1/e)
-\frac{1}{2n}\tx_1e^{-1/(2n)}\\
&= W_1(2n)+\frac{1}{2n}\sbp{\tx_{2n}e^{-1}-\tx_1e^{-1/(2n)}}\\
&\ge
\ty+\frac{1}{2n}\sbp{\tx_{2n}e^{-1}-\tx_1e^{-1/(2n)}}
>
\ty-\frac{1-1/e}{2n}
=
y,
\end{align*}
\endgroup
where the last inequality uses $\tx_1\le 1-1/e$.

\medskip
\noindent
\textbf{Constraint~\eqref{lp:cons-4}.}
For each pair $i\in(n)$ and $j\in(n-i)$, let $W_2(i,j,n)$ denote the right-hand side of Constraint~\eqref{lp:cons-4} in $\LP(n)$ under $\bx$.  Similarly, let $W_2(\hi,\hj,2n)$ denote the corresponding right-hand side in $\LP(2n)$ under $\tbx$.  Since $(\ty,\tbx)$ is feasible for $\LP(2n)$,
\[
W_2(\hi,\hj,2n)\ge \ty,
\qquad
\forall\, \hi\in(2n),\ \hj\in(2n-\hi).
\]

For $(i,j)=(0,0)$, both sums are empty and $x_0=\tx_0$, so
\[
W_2(0,0,n)=1-x_0=1-\tx_0=W_2(0,0,2n)\ge\ty\ge y.
\]

Now suppose $0<i+j<n$. Then $(2i+1,2j)$ is an admissible fine-grid pair, and~\eqref{def:2} expresses $x_{i+j}$ as the average of $\tx_{2i+2j}$ and $\tx_{2i+2j+1}$. Thus
\begingroup
\allowdisplaybreaks
\begin{align*}
W_2(i,j,n)
&= \frac{1}{n}\sum_{t=1}^{i}x_t e^{-t/n}
+\frac{1}{n}\sum_{t=i+1}^{i+j}e^{-t/n}
+\bp{1-\frac{j}{n}}\bp{1-x_{i+j}}\\
&= \frac{1}{n}\sum_{t=1}^{i}\frac{\tx_{2t}+\tx_{2t+1}}{2}e^{-t/n}
+\frac{1}{n}\sum_{t=i+1}^{i+j}e^{-t/n}
+\bp{1-\frac{2j}{2n}}\bp{1-\frac{\tx_{2i+2j}+\tx_{2i+2j+1}}{2}}\\
&= \frac{1}{2n}\sum_{t=1}^{i}\sbp{\tx_{2t}e^{-2t/(2n)}+\tx_{2t+1}e^{-2t/(2n)}}
+\frac{1}{2n}\sum_{t=i+1}^{i+j}\sbp{e^{-2t/(2n)}+e^{-2t/(2n)}}\\
&\quad
+\bp{1-\frac{2j}{2n}}\bp{1-\frac{\tx_{2i+2j}+\tx_{2i+2j+1}}{2}}\\
&\ge
\frac{1}{2n}\sum_{t=1}^{i}\sbp{\tx_{2t}e^{-2t/(2n)}+\tx_{2t+1}e^{-(2t+1)/(2n)}}
+\frac{1}{2n}\sum_{t=i+1}^{i+j}\sbp{e^{-2t/(2n)}+e^{-(2t+1)/(2n)}}\\
&\quad
+\bp{1-\frac{2j}{2n}}\bp{1-\tx_{2i+2j+1}}\\
&=
\frac{1}{2n}\sum_{k=1}^{2i+1}\tx_k e^{-k/(2n)}
+\frac{1}{2n}\sum_{k=2i+2}^{2i+1+2j}e^{-k/(2n)}
+\bp{1-\frac{2j}{2n}}\bp{1-\tx_{2i+2j+1}}
-\frac{1}{2n}\tx_1e^{-1/(2n)}\\
&= W_2(2i+1,2j,2n)-\frac{1}{2n}\tx_1e^{-1/(2n)}\\
&\ge
\ty-\frac{1}{2n}\tx_1e^{-1/(2n)}
\ge
\ty-\frac{1-1/e}{2n}
=
y.
\end{align*}
\endgroup

Finally, suppose $i+j=n$. Compare $W_2(i,j,n)$ with
$W_2(2i,2j,2n)$: the pair $(2i,2j)$ is admissible, and the
terminal terms agree because $2j/(2n)=j/n$ and $x_n=\tx_{2n}$.
Set $h:=1/(2n)$ and $a_k:=e^{-kh}$ for $0\le k\le 2n$.
For $1\le t\le i$, monotonicity and~\eqref{def:2} give
$x_t\ge\tx_{2t}\ge\tx_{2t-1}$, also at $t=n$ if $i=n$.
Hence the fine-grid contribution from indices $2t-1,2t$ minus
its coarse-grid counterpart is bounded by
\[
\begin{aligned}
h\bigl(\tx_{2t-1}a_{2t-1}+\tx_{2t}a_{2t}-2x_ta_{2t}\bigr)
&\le h\tx_{2t-1}(a_{2t-1}-a_{2t})\le h(a_{2t-1}-a_{2t}),
\end{aligned}
\]
since $\tx_{2t-1}\le 1$. For $i<t\le n$, the coefficients in
the second sum are all one, so the difference is exactly
$h(a_{2t-1}-a_{2t})$. Summing over all $n$ pairs and using that
$(a_k)_{k=0}^{2n}$ is decreasing, we obtain
\[
\begin{aligned}
W_2(2i,2j,2n)-W_2(i,j,n)
&\le h\sum_{t=1}^n(a_{2t-1}-a_{2t})\le h\sum_{k=1}^{2n}(a_{k-1}-a_k)
 =h(a_0-a_{2n})
 =\frac{1-1/e}{2n}.
\end{aligned}
\]
Fine-grid feasibility now gives
$W_2(i,j,n)\ge W_2(2i,2j,2n)-(1-1/e)/(2n)\ge y$.
This verifies Constraint~\eqref{lp:cons-4} for every admissible pair.

Combining the three constraint checks, $(y,\bx)$ is feasible for $\LP(n)$.
\end{proof}



\section{Auxiliary-LP Upper Certification for
\texorpdfstring{$\Gam(\cF_0)$}{Γ(𝓕₀)} and
\texorpdfstring{$\Gam(\cF_1)$}{Γ(𝓕₁)}}
\label{sec:ub}

This section proves the $\cF_0$ and $\cF_1$ upper-certification parts of Theorem~\ref{thm:aux-lp-cert}.  Unlike the two-sided certification over $\cF_2$ in Section~\ref{sec:lb}, here we only need valid upper certificates.  Thus we do not require an exact reduction of $\cL[f]$ for every $f$ in the larger classes.  It suffices to identify adversarial tests that upper bound $\cL[f]$ uniformly over the class.

Recall the quantities $W_1$ and $W_2$ defined in~\eqref{int:23-a} and~\eqref{int:23-b}.  For any admissible function $f \in \cF_0$, the canonical representation in Theorem~\ref{thm:structure} implies
\[
\cL[f]\le W_1,
\qquad
\cL[f]\le W_2.
\]
Indeed, $W_1$ is obtained by fixing the adversarial parameters
\[
\ell=1,\qquad \psi=\tpsi=1,
\]
while $W_2$ is obtained by restricting to strategies with
\[
\psi=1,\qquad \ell,\tpsi\in[0,1],\qquad \ell+\tpsi\le1.
\]
Since $\cL[f]$ is defined by minimizing over all adversarial strategies, restricting the adversary to these special choices yields valid upper bounds.  Therefore
\begin{align}\label{ineq:Lf-W12-upper}
\cL[f]\le \min\{W_1,W_2\}.
\end{align}

Consequently, for the classes $\cF_0$ and $\cF_1$ defined in~\eqref{eq:def-F0}--\eqref{eq:def-F1},
\begin{align}
\Gam(\cF_0)
&\le
\sup_{f\in\cF_0}\min\{W_1,W_2\},
\label{prog:2-25-2}\\
\Gam(\cF_1)
&\le
\sup_{f\in\cF_1}\min\{W_1,W_2\}.
\label{prog:2-25-7}
\end{align}
We now certify these two upper envelopes by finite-dimensional LPs.

\xhdr{Auxiliary upper-bound LP.}
Both $W_1$ and $W_2$ depend only on the restriction of $f$ to $[0,1]$.  As before, for a grid size $n\ge1$, introduce variables
\[
x_t:=f(t/n),
\qquad
t\in(n):=\{0,1,\ldots,n\}.
\]
For the full class $\cF_0$, the finite upper certificate is the following LP.

\smallskip
\begin{tcolorbox}[title=(\textbf{AUG-UB-LP}) Auxiliary upper-bound LP for~\eqref{prog:2-25-2}]
\begingroup
\allowdisplaybreaks
\begin{align}
&\max y \label{lp3:obj} \\
& 0 \le x_t \le x_{t+1} \le 1, && \forall t\in(n-1), \label{lp3:cons-1}\\
& y \le \frac{1}{n}\sum_{t=1}^{n} x_t e^{-t/n} + e^{-1}(1-e^{-1}), \label{lp3:cons-3}\\
& y \le \frac{1}{n}\sum_{t=1}^{i} x_t e^{-t/n}
      + \frac{1}{n}\sum_{t=i+1}^{i+j} e^{-t/n}
      + (1-j/n)\bigl(1-x_{i+j}\bigr),
      && \forall i\in(n),\ j\in(n-i). \label{lp3:cons-4}
\end{align}
\endgroup
\end{tcolorbox}

For $\cF_1$, we use the analogous upper-bound LP obtained by replacing Constraint~\eqref{lp3:cons-1} with
\[
0\le x_t\le x_{t+1}\le 1-1/e,
\qquad \forall t\in(n-1).
\]
These LPs are upper-certification devices: their limiting values certify the right-hand sides of~\eqref{prog:2-25-2} and~\eqref{prog:2-25-7}, which in turn upper bound $\Gam(\cF_0)$ and $\Gam(\cF_1)$.

{Write $\zeta_0(n):=\zeta(n)$ for the exact optimum of
\textbf{AUG-UB-LP}~\eqref{lp3:obj}, and let $\zeta_1(n)$ be the
exact optimum of its $\cF_1$ variant.  These are distinct
optimization problems from \textbf{AUG-LP}~\eqref{lp:obj}, whose
exact optimum is $\eta(n)$.  The argument below uses separately
verified upper bounds on $\zeta_0(n)$ and $\zeta_1(n)$; it does
not assume equality of either optimum with $\eta(n)$.}

\begin{proposition}[Auxiliary-LP upper certificates]\label{prop:2-25-1}
The following bounds hold:
\[
 {\Gam(\cF_0)\le 0.5842},
\qquad
\Gam(\cF_1)\le 0.5809.
\]
\end{proposition}

\begin{proof}
 {For $k\in\{0,1\}$, set $\delta_0:=1$ and
$\delta_1:=1-1/e$.  The same refinement and coarsening argument
used in Proposition~\ref{pro:2-18-a}, together with consistency
with the corresponding upper envelope, gives}
{\begin{align}\label{ineq:2-26-1}
\bigl|\zeta_k(n)-\zeta_k(\infty)\bigr|
\leq\frac{\delta_k}{n},
\qquad n\in\mathbb N_{\geq1},\quad k\in\{0,1\},
\end{align}}
{where $\zeta_k(\infty)$ is the supremal value of the
upper envelope over $\cF_k$.  Consequently, any verified finite
LP upper bound $\zeta_k(n)\leq\overline{\zeta}_{k,n}$ yields}
{\[
\Gam(\cF_k)\leq\zeta_k(\infty)
\leq\zeta_k(n)+\frac{\delta_k}{n}
\leq\overline{\zeta}_{k,n}+\frac{\delta_k}{n}.
\]}

{The finite certificates described below establish}
{\[
\zeta_0(1000)\leq0.5831180,
\qquad
\zeta_1(1000)\leq0.5802564.
\]}
{Therefore}
{\[
\Gam(\cF_0)
\leq0.5831180+\frac1{1000}
=0.5841180<0.5842,
\]}
{and}
{\[
\Gam(\cF_1)
\leq0.5802564+\frac{1-1/e}{1000}
<0.5809
\]}
{These inequalities use certified upper endpoints, not
rounded approximations to the finite LP optima.}
\end{proof}

\xhdr{Finite LP certificates.}
{Table~\ref{table:nume} distinguishes the exact finite LP
optima from their certified bounds.  Its second column gives
$[\underline{\eta}_n,\overline{\eta}_n]$, an enclosure of
$\eta(n)$, and its fifth column gives a verified upper bound
$\overline{\zeta}_{0,n}$ on $\zeta_0(n)$.  The remaining columns
are obtained by first adding the discretization-error bounds and
then rounding outward to four decimal places.  Specifically,}
{\[
\begin{aligned}
L_{2,n}
&:=10^{-4}\left\lfloor
  10^4\left(\underline{\eta}_n-\frac{1-1/e}{n}\right)
\right\rfloor,\\
U_{2,n}
&:=10^{-4}\left\lceil
  10^4\left(\overline{\eta}_n+\frac{1-1/e}{n}\right)
\right\rceil,\\
U_{0,n}
&:=10^{-4}\left\lceil
  10^4\left(\overline{\zeta}_{0,n}+\frac1n\right)
\right\rceil.
\end{aligned}
\]}
{Thus $\Gam(\cF_2)\in[L_{2,n},U_{2,n}]$ and
$\Gam(\cF_0)\leq U_{0,n}$.  The $\cF_1$ certificate is
verified separately and gives
$\overline{\zeta}_{1,1000}=0.5802564$.
Coincidence of displayed bounds does not assert equality of the
underlying LP optima.}

\begin{table}[ht!]
\caption{{Verified finite LP bounds and the resulting
functional bounds.  The interval in the second column encloses
the exact optimum of \textbf{AUG-LP}~\eqref{lp:obj}; the fifth
column upper bounds the exact optimum of
\textbf{AUG-UB-LP}~\eqref{lp3:obj}.  The columns
$L_{2,n},U_{2,n},U_{0,n}$ include discretization error and are
rounded outward.}}
\label{table:nume}
\centering
\small
\setlength{\tabcolsep}{5pt}
\renewcommand{\arraystretch}{1.45}
{\begin{tabular}{c c >{\columncolor{gray!15}}c >{\columncolor{gray!35}}c c >{\columncolor{gray!55}}c}
\thickhline
$n$ & $[\underline{\eta}_n,\overline{\eta}_n]$
& $L_{2,n}$ & $U_{2,n}$
& $\overline{\zeta}_{0,n}$ & $U_{0,n}$ \\
\thickhline
$10$ & $[0.5712752,0.5712753]$ & $0.5080$ & $0.6345$ & $0.5735947$ & $0.6736$ \\
$100$ & $[0.5795033,0.5795034]$ & $0.5731$ & $0.5859$ & $0.5823409$ & $0.5924$ \\
$500$ & $[0.5801732,0.5801733]$ & $0.5789$ & $0.5815$ & $0.5830325$ & $0.5851$ \\
$1000$ & $[0.5802563,0.5802564]$ & $0.5796$ & $0.5809$ & $0.5831180$ & $0.5842$ \\
\thickhline
\end{tabular}}
\end{table}

{\xhdr{Independent verification of the finite LP bounds.}
Floating-point LP solves are used only to propose node values and
nonnegative multipliers.  The accompanying file
\texttt{lp\_candidates.json} records these candidates as integers.
The standalone verifier \texttt{verify\_lp\_certificates.py}
uses only exact Python integers and standard-library rational
arithmetic; it does not invoke an LP solver or trust the
candidate-generation procedure.  It verifies each LP family
separately, including the lower-envelope and terminal-value
conditions of \textbf{AUG-LP}, and produces
\texttt{lp\_certificate\_report.json}.}

{For a primal certificate, let the recorded nodal candidates
be $v_t\in10^{-15}\mathbb Z$ and set
$m_t:=\max\{0,v_0,\ldots,v_t\}$.  Put
$\delta=1-1/e$ for the $\cF_2$ and $\cF_1$ LPs, and
$\delta=1$ for the $\cF_0$ LP.  For the $\cF_2$ LP, define}
{\[
x_t:=\min\{\delta,\max\{m_t,1-e^{-t/n}\}\}
\quad(0\leq t<n),
\qquad x_n:=\delta.
\]}
{For either upper-envelope LP, instead set
$x_t:=\min\{\delta,m_t\}$ for $0\leq t\leq n$.
These definitions impose the required shape conditions exactly.
A downward-rounded lower bound on the minimum of all the
right-hand sides of the scalar and partial-load inequalities
then supplies a feasible objective value.  The verifier checks
all $(n+1)(n+2)/2$ partial-load inequalities, including those
omitted from the solver's working constraint set.}

{For an upper certificate, take a nonnegative integer
combination of the scalar, partial-load, and monotonicity
inequalities.  After collecting terms, it has the form}
{\[
W y+\sum_{t=0}^n a_t x_t\leq b,
\qquad W>0.
\]}
{Here $W$ is the total weight on the scalar and partial-load
inequalities.  If $\lambda_t\leq x_t\leq\upsilon_t$ are the
exact box bounds for the chosen LP, then every feasible
solution satisfies}
{\[
y\leq\frac1W\left(
 b+\sum_{t=0}^n
 \max\{-a_t\lambda_t,-a_t\upsilon_t\}
\right).
\]}
{For the $\cF_2$ LP these box bounds are
$\lambda_t=1-e^{-t/n}$ and $\upsilon_t=1-1/e$; for the
$\cF_k$ upper-envelope LP they are $0$ and $\delta_k$.
The verifier encloses the entire right-hand side from above.
Thus it accounts for residual coefficients explicitly instead
of treating approximately cancelling coefficients as zero.}

{All interval endpoints lie on the lattice $10^{-50}\mathbb Z$,
and every arithmetic operation is rounded outward by integer
floor and ceiling operations.  The constants $e^{-1}$ and
$e^{-1/n}$ are enclosed using exact rational Taylor polynomials
of degrees $81$ and $80$ as lower and upper bounds,
respectively; the enclosures of $e^{-t/n}$ are propagated by
outward interval multiplication.  Final finite LP endpoints are
rounded outward to seven decimal places.  All twelve checks,
for $n\in\{10,100,500,1000\}$ and the three LP families, pass.
The numerical bounds in Table~\ref{table:nume} therefore include
finite-arithmetic error independently of the discretization
bounds used in Propositions~\ref{pro:2-18-a}
and~\ref{prop:2-25-1}.}

{The finite-LP certificate data, verification code, and exploratory MATLAB scripts are available in the supplementary repository:}
\begin{center}
\href{https://drive.google.com/drive/folders/1ZuRxkxKS-7HyNTcOYBQq8b80ORtnO3Dc?usp=sharing}
{\texttt{https://drive.google.com/drive/folders/1ZuRxkxKS-7HyNTcOYBQq8b80ORtnO3Dc}}
\end{center}

\section{Closed-Form Certification over
\texorpdfstring{$\cF_3$}{𝓕₃}}
\label{app:analy}

This appendix proves the $\cF_3$ part of
Theorem~\ref{thm:aux-lp-cert}. Recall from~\eqref{eq:def-F3} that
$f\in\cF_3$ belongs to $\cF_2$, is absolutely continuous on $[0,1]$,
and satisfies
\[
0\le\dot f(z)\le e^{-z}-(1-f(z))
\quad\text{for almost every }z\in[0,1].
\]
In particular, $f(1)=1-1/e$ and $f(z)=1-1/e$ for $z\ge1$.
We first show that the partial-load minimum equals $1-f(0)$.
An integrating-factor comparison then bounds this initial value and
identifies an explicit function attaining the bound.

By Proposition~\ref{pro:2-17-a}, every $f\in\cF_3\subseteq\cF_2$
satisfies
\[
\cL[f]=\min\{W_1(f),W_2(f)\},
\]
where
\begin{align*}
W_1(f)
&:=\int_0^1e^{-z}f(z)\,\sd z+\frac1e(1-1/e),\\
W_2(f)
&:=\min_{\substack{\ell,\tpsi\ge0\\\ell+\tpsi\le1}}
\Psi_f(\ell,\tpsi),
\end{align*}
and, consistently with the notation in Section~\ref{sec:thm2},
\[
\Psi_f(\ell,\tpsi)
:=\int_0^\ell e^{-z}f(z)\,\sd z
 +\int_\ell^{\ell+\tpsi}e^{-z}\,\sd z
 +(1-\tpsi)\bigl(1-f(\ell+\tpsi)\bigr).
\]

\xhdr{Step 1: Evaluating the partial-load minimum.}
Fix $\ell\in[0,1]$. Absolute continuity of $f$ implies that
$\tpsi\mapsto\Psi_f(\ell,\tpsi)$ is absolutely continuous on
$[0,1-\ell]$. For almost every $\tpsi$ in this interval, writing
$z=\ell+\tpsi$, we have
\begin{align*}
\frac{\partial\Psi_f(\ell,\tpsi)}{\partial\tpsi}
&=e^{-z}-(1-f(z))-(1-\tpsi)\dot f(z)\\
&\ge\dot f(z)-(1-\tpsi)\dot f(z)
=\tpsi\dot f(z)\ge0.
\end{align*}
Thus $\Psi_f(\ell,\tpsi)$ is non-decreasing in $\tpsi$ and
\[
\Psi_f(\ell,\tpsi)\ge\Psi_f(\ell,0)
=\int_0^\ell e^{-z}f(z)\,\sd z+1-f(\ell).
\]
The function $\ell\mapsto\Psi_f(\ell,0)$ is also absolutely
continuous, and for almost every $\ell\in[0,1]$,
\begin{align*}
\frac{\sd}{\sd\ell}\Psi_f(\ell,0)
&=e^{-\ell}f(\ell)-\dot f(\ell)\\
&\ge e^{-\ell}f(\ell)-\bigl(e^{-\ell}-(1-f(\ell))\bigr)\\
&=(1-f(\ell))(1-e^{-\ell})\ge0.
\end{align*}
It is therefore non-decreasing in $\ell$. Consequently,
\[
\Psi_f(\ell,\tpsi)\ge\Psi_f(\ell,0)
\ge\Psi_f(0,0)=1-f(0)
\]
for every feasible pair $(\ell,\tpsi)$. Equality is attained at
$\ell=\tpsi=0$, so
\[
W_2(f)=1-f(0),\qquad f\in\cF_3.
\]
Hence
\begin{align}\label{eqn:max_ff}
\Gam(\cF_3)
=\sup_{f\in\cF_3}
\min\left\{
\int_0^1e^{-z}f(z)\,\sd z+\frac1e(1-1/e),\;
1-f(0)
\right\}.
\end{align}

\xhdr{Step 2: A pointwise comparison and an upper bound.}
For any $f\in\cF_3$, the function $z\mapsto e^{-z}(1-f(z))$ is
absolutely continuous. The derivative constraint gives
\[
\frac{\sd}{\sd z}\bigl(e^{-z}(1-f(z))\bigr)
=-e^{-z}\bigl(1-f(z)+\dot f(z)\bigr)
\ge-e^{-2z}
\quad\text{almost everywhere on }[0,1].
\]
Integrating from $z$ to $1$ and using $1-f(1)=1/e$ yields
\[
e^{-z}(1-f(z))
\le\frac1{e^2}+\int_z^1e^{-2s}\,\sd s
=\frac{e^{-2z}+e^{-2}}2,
\qquad 0\le z\le1.
\]
Equivalently,
\[
1-f(z)\le\frac{e^{-z}+e^{z-2}}2,
\qquad 0\le z\le1.
\]
In particular,
\[
\cL[f]\le W_2(f)=1-f(0)\le\frac{1+e^{-2}}2.
\]
Taking the supremum over $\cF_3$ proves the required upper bound.

\xhdr{Step 3: An attaining function.}
Define
\begin{align}\label{eqn:2-28-a}
f_3^*(z):=
\begin{cases}
\displaystyle 1-\frac{e^{-z}+e^{z-2}}2,&0\le z\le1,\\[6pt]
1-1/e,&z>1.
\end{cases}
\end{align}
On $[0,1]$, this function is continuously differentiable and satisfies
\[
\frac{\sd f_3^*(z)}{\sd z}
=\frac{e^{-z}-e^{z-2}}2
=e^{-z}-(1-f_3^*(z))\ge0,
\]
with one-sided derivatives at the endpoints. Moreover,
$f_3^*(1)=1-1/e$ and
\[
f_3^*(z)-(1-e^{-z})
=\frac{e^{-z}-e^{z-2}}2\ge0,
\qquad 0\le z\le1.
\]
Monotonicity and the terminal value give $f_3^*(z)\le1-1/e$ on
$[0,1]$. The constant continuation is therefore admissible in $\cF_2$,
and the differential identity above proves $f_3^*\in\cF_3$.

By Step~1,
\[
W_2(f_3^*)=1-f_3^*(0)=\frac{1+e^{-2}}2.
\]
The other functional has the value
\begin{align*}
W_1(f_3^*)
&=\int_0^1
\left(e^{-z}-\frac12e^{-2z}-\frac12e^{-2}\right)\,\sd z
 +\frac1e(1-1/e)=\frac34-\frac{5}{4e^2}.
\end{align*}
In particular,
\[
W_1(f_3^*)-W_2(f_3^*)
=\frac{1-7e^{-2}}4>0,
\]
since $e^2>7$. Thus the minimum is attained by $W_2$, and
\[
\cL[f_3^*]=\frac{1+e^{-2}}2.
\]
Combining this equality with Step~2 shows that the supremum is attained
and establishes
\[
\Gam(\cF_3)=\frac{1+e^{-2}}2\approx0.56766764.
\]

\xhdr{Role of the two tests.}
The pointwise comparison in Step~2 shows that every $f\in\cF_3$
satisfies $f\ge f_3^*$ on $[0,1]$. Consequently,
\[
W_1(f)\ge W_1(f_3^*)
>\frac{1+e^{-2}}2
\ge W_2(f).
\]
Hence, throughout this restricted class,
\[
\cL[f]=1-f(0),\qquad f\in\cF_3.
\]
The closed-form optimum is therefore determined by the differential
constraint and the terminal value, not by balancing $W_1$ and $W_2$.
Its value matches the classical stochastic-matching guarantee
of Mehta and Panigrahi~\cite{mehta2013online}.\footnote{The function
in~\eqref{eqn:2-28-a} also appears in the explicit
$(1+e^{-2})/2$-competitive analysis for unequal vanishing probabilities
in~\cite{huang2020online}. Here we establish its optimality for the
restricted functional optimization problem over $\cF_3$.}


\section{The Role of the Scalar Constraint}
\label{app:const-class}

This appendix examines the role of the scalar constraint
$\Phi_1(f)\ge\gam$ in Functional Optimization
Program~\eqref{prog:var-gamma-F0}. We show that functions constant
on $[0,1]$ are strictly suboptimal for the full program, but become
optimal when this constraint is removed.

\begin{proposition}[Constant functions are strictly suboptimal]
\label{prop:const-class}
Suppose $f\in\cF_0$ satisfies $f(z)\equiv c$ on $[0,1]$, where
$c\in[0,1]$.

\begin{itemize}
\item[(i)] If $c<1-1/e$, then $\tau(f)=1$ and
\[
\Phi_1(f)=(1-e^{-1})(c+e^{-1}),
\qquad
\Phi_2(f)=\min\{1-c,\,1-e^{-1}\}.
\]
Moreover,
\[
\max_{0\le c<1-1/e}
\min\{\Phi_1(f),\Phi_2(f)\}
=
\frac{1-e^{-2}}{2-e^{-1}},
\]
attained at
\[
c^*=\frac{1-e^{-1}+e^{-2}}{2-e^{-1}}.
\]

\item[(ii)] If $c\ge1-1/e$, then $\tau(f)=0$ and
\[
\Phi_1(f)=1-e^{-1},
\qquad
\Phi_2(f)=1-c.
\]
Hence
\[
\max_{1-1/e\le c\le1}
\min\{\Phi_1(f),\Phi_2(f)\}
=
\frac1e,
\]
attained at $c=1-1/e$.
\end{itemize}

Consequently,
\[
\max_{\substack{f\in\cF_0\\
                 f\text{ is constant on }[0,1]}}
\min\{\Phi_1(f),\Phi_2(f)\}
=
\frac{1-e^{-2}}{2-e^{-1}}
<0.53.
\]
These functions are strictly suboptimal for both
Program~\eqref{prog:var-gamma-F0} and the maximization of $\cL[f]$
over $\cF_0$.
\end{proposition}

\begin{proof}
Both functionals depend only on the restriction of $f$ to $[0,1]$,
so its continuation beyond this interval does not affect the
calculation. We consider the two cases separately.

\xhdr{Case 1: $c<1-1/e$.}
Then $\tau(f)=1$, and~\eqref{eqn:ell-sig} reduces to
\[
1-c=e^{-\ell}(1-c).
\]
Since $1-c>0$, the minimizing load is $\ell(\sig)=0$ for every
$\sig\in[0,1]$. Thus
\[
\Phi_1(f)
=
\int_0^1 c e^{-z}\,\sd z+e^{-1}(1-e^{-1})
=
(1-e^{-1})(c+e^{-1}),
\]
and
\[
\Psi_f(\sig)
=
1-e^{-\sig}+(1-\sig)(1-c).
\]
This expression is concave in $\sig$, so its minimum over $[0,1]$
is attained at an endpoint. Hence
\[
\Phi_2(f)=\min\{1-c,\,1-e^{-1}\}.
\]

Since $c+e^{-1}<1$, we have $\Phi_1(f)<1-e^{-1}$ and therefore
\[
\min\{\Phi_1(f),\Phi_2(f)\}
=
\min\{(1-e^{-1})(c+e^{-1}),\,1-c\}.
\]
The first expression is strictly increasing in $c$, whereas the
second is strictly decreasing. Their intersection,
\[
(1-e^{-1})(c+e^{-1})=1-c,
\]
occurs at
\[
c^*=\frac{1-e^{-1}+e^{-2}}{2-e^{-1}}
\in(e^{-1},\,1-e^{-1}).
\]
Consequently, the maximum in this case is attained at $c^*$, with
value
\[
1-c^*=\frac{1-e^{-2}}{2-e^{-1}}.
\]

\xhdr{Case 2: $c\ge1-1/e$.}
Then $\tau(f)=0$, so $\sig=0$ is the only admissible total load.
The definitions give
\[
\Phi_1(f)=1-e^{-1},
\qquad
\Phi_2(f)=1-c.
\]
Since $1-c\le1/e<1-1/e$, it follows that
\[
\min\{\Phi_1(f),\Phi_2(f)\}=1-c.
\]
This is maximized at $c=1-1/e$, giving value $1/e$.

Because
\[
\frac{1-e^{-2}}{2-e^{-1}}>\frac1e,
\]
the maximum over all constant restrictions is attained in Case~1
and equals $(1-e^{-2})/(2-e^{-1})<0.53$.

Finally, Lemma~\ref{lem:12-2-a} gives
$\cL[f]\le\min\{\Phi_1(f),\Phi_2(f)\}$ for every $f\in\cF_0$.
Thus
\[
\sup_{\substack{f\in\cF_0\\
                f\text{ is constant on }[0,1]}}
\cL[f]
\le
\frac{1-e^{-2}}{2-e^{-1}}
<0.53.
\]
Appendix~\ref{app:analy} gives an admissible function in $\cF_3$
with certificate value $(1+e^{-2})/2>0.53$. The same function is
feasible at that value in Program~\eqref{prog:var-gamma-F0}, since
$\cL[f]\le\min\{\Phi_1(f),\Phi_2(f)\}$. Therefore functions constant
on $[0,1]$ are strictly suboptimal for both optimizations.
\end{proof}

\begin{proposition}[Constant functions are optimal without the scalar constraint]
\label{prop:no-Phi1}
Consider the relaxation of Program~\eqref{prog:var-gamma-F0}
obtained by removing the constraint $\Phi_1(f)\ge\gam$:
\begin{tcolorbox}[
title={\normalfont Relaxed Functional Optimization Program without
the Constraint $\Phi_1(f)\ge\gam$}
]
\begin{align}
\sup_{\gam\in\mathbb R,\,f\in\cF_0}\quad
&\gam
\label{prog:relaxed-var-gamma-F0}\\
\text{s.t.}\quad
&\Phi_2(f)\ge\gam
\quad\Longleftrightarrow\quad
\Bigl(
\Psi_f(\sig)\ge\gam,\quad
\forall\sig\in[0,\tau(f)]
\Bigr).
\nonumber
\end{align}
\end{tcolorbox}
Its optimal value is
\[
\sup_{f\in\cF_0}\Phi_2(f)=1-\frac1e,
\]
attained by every $f\in\cF_0$ satisfying $f(z)\equiv c$ on
$[0,1]$ for some $c\in[0,1/e]$.
\end{proposition}

\begin{proof}
Fix $f\in\cF_0$ and write $\tau:=\tau(f)$. If
$f(0)\ge1-1/e$, then $\tau=0$ and
\[
\Phi_2(f)=1-f(0)\le\frac1e<1-\frac1e.
\]
We may therefore assume $f(0)<1-1/e$, so $\tau>0$.

Define
\[
g(z):=1-f(z),
\qquad z\in[0,\tau].
\]
The function $g$ is continuous and non-increasing, and
$1/e\le g(z)\le1$ on this interval. For
$\sig\in[0,\tau]$, recall that
\[
\Psi_f(\sig)
=
\int_0^{\ell(\sig)}e^{-z}f(z)\,\sd z
+
\int_{\ell(\sig)}^\sig e^{-z}\,\sd z
+
(1-\sig+\ell(\sig))g(\sig),
\]
where the minimizing load satisfies
\[
g(\sig)=e^{-\ell(\sig)}g(\ell(\sig)).
\]
Writing $\ell=\ell(\sig)$ and using $f=1-g$, we obtain
\[
\Psi_f(\sig)
=
1-e^{-\sig}
+(1-\sig+\ell)g(\sig)
-\int_0^\ell e^{-z}g(z)\,\sd z.
\]
Since $g$ is non-increasing,
\[
\int_0^\ell e^{-z}g(z)\,\sd z
\ge
g(\ell)\int_0^\ell e^{-z}\,\sd z
=
g(\ell)(1-e^{-\ell}).
\]
Consequently,
\begin{align*}
\Psi_f(\sig)
&\le
1-e^{-\sig}
+(1-\sig+\ell)g(\sig)
-g(\ell)(1-e^{-\ell})\\
&=
1-e^{-\sig}
+(1-\sig+\ell)g(\sig)
-(e^\ell-1)g(\sig)\\
&\le
1-e^{-\sig}
+(1-\sig+\ell)g(\sig)
-\ell g(\sig)\\
&=
1-e^{-\sig}+(1-\sig)g(\sig),
\end{align*}
where we used $g(\sig)=e^{-\ell}g(\ell)$,
$e^\ell-1\ge\ell$, and $g(\sig)\ge0$.

We now distinguish whether the threshold is reached on $[0,1]$.

\xhdr{Case 1: $f(1)<1-1/e$.}
Here $\tau=1$. Evaluating the preceding bound at $\sig=1$ gives
\[
\Phi_2(f)\le\Psi_f(1)\le1-\frac1e.
\]

\xhdr{Case 2: $f(1)\ge1-1/e$.}
Since $f(0)<1-1/e$, continuity implies
$\tau\in(0,1]$ and $f(\tau)=1-1/e$. Thus $g(\tau)=1/e$, and
\[
\Phi_2(f)
\le
\Psi_f(\tau)
\le
1-e^{-\tau}+(1-\tau)g(\tau)
=
1-e^{-\tau}+\frac{1-\tau}{e}.
\]
The function $t\mapsto1-e^{-t}+(1-t)/e$ is non-decreasing on
$[0,1]$, since its derivative is $e^{-t}-1/e\ge0$.
Its value at $t=1$ is $1-1/e$, so
\[
\Phi_2(f)\le1-\frac1e.
\]

Together with the case $\tau=0$, these bounds establish
\[
\Phi_2(f)\le1-\frac1e
\qquad\text{for every }f\in\cF_0.
\]
It remains to show attainment. If $f(z)\equiv c$ on $[0,1]$ with
$c\in[0,1/e]$, then Proposition~\ref{prop:const-class} gives
\[
\Phi_2(f)
=
\min\{1-c,\,1-e^{-1}\}
=
1-\frac1e.
\]
Therefore the supremum in
Program~\eqref{prog:relaxed-var-gamma-F0} is attained and equals
$1-1/e$, as claimed.
\end{proof}

\section{Proof of Theorem~\ref{thm:main-exact}}
\label{app:thm-exact-opt}

We prove the two inequalities separately. The upper bound follows from
restricting the adversarial choices in the canonical representation of
$\cL[f]$. The matching lower bound uses the structural conclusions of
Theorem~\ref{thm:main-structural}, proved directly for the functional 
program in Appendix~\ref{app:gam-f}. Neither direction uses the numerical
certification in Appendix~\ref{app:num}.

\begin{proposition}[Upper bound]\label{pro:upp-bound}
$\Gam(\cF_0)\le\gam^*$.
\end{proposition}

\begin{proposition}[Matching lower bound]\label{pro:low-bound}
$\Gam(\cF_0)\ge\gam^*$.
\end{proposition}

\subsection{Proof of Proposition~\ref{pro:upp-bound} (Upper bound)}

\begin{lemma}\label{lem:12-2-a}
For every $f\in\cF_0$, the quantity $\cL[f]$ defined
in~\eqref{eqn:Lf-d-1} satisfies
\[
\cL[f]\le\min\{\Phi_1(f),\Phi_2(f)\}.
\]
\end{lemma}

\begin{proof}
In the canonical representation~\eqref{eqn:Lf-d-1}, choose
$\psi=1$, $\ell=\tau(f)$, and $\tpsi=1$. This gives
\begin{align*}
\cL[f]
&\le\int_0^{\tau(f)}e^{-z}f(z)\,\sd z
   +\int_{\tau(f)}^{\tau(f)+1}e^{-z}\,\sd z\\
&=\int_0^{\tau(f)}e^{-z}f(z)\,\sd z
   +(1-1/e)e^{-\tau(f)}
 =\Phi_1(f).
\end{align*}
Likewise, fixing $\psi=1$ and restricting to
$\ell,\tpsi\ge0$ with $\ell+\tpsi\le\tau(f)$ yields
\[
\cL[f]
\le\min_{\substack{\ell,\tpsi\ge0\\\ell+\tpsi\le\tau(f)}}
\Psi_f(\ell,\tpsi)
=\Phi_2(f).
\]
Both choices are valid for every $f\in\cF_0$, including the endpoint
cases $\tau(f)=0$ and $\tau(f)=1$. Combining the bounds proves the lemma.
\end{proof}

\begin{proof}[Proof of Proposition~\ref{pro:upp-bound}]
By Lemma~\ref{lem:12-2-a} and the definition of
Program~\eqref{prog:var-gamma-F0},
\[
\Gam(\cF_0)
=\sup_{f\in\cF_0}\cL[f]
\le\sup_{f\in\cF_0}\min\{\Phi_1(f),\Phi_2(f)\}
=\gam^*.
\]
\end{proof}

\subsection{Proof of Proposition~\ref{pro:low-bound} (Matching lower bound)}

Choose an optimal solution $(\gam^*,f^*)$ to
Program~\eqref{prog:var-gamma-F0}, whose existence follows from
Theorem~\ref{thm:main-structural}, and put $\tau^*:=\tau(f^*)$.
The same theorem gives $0<\tau^*<1$ and $f^*(\tau^*)=1-1/e$.
Although $f^*$ is defined on $[0,\infty)$, both constraints
depend only on its critical restriction to $[0,\tau^*]$. We choose the
constant continuation of that restriction:
\begin{align}\label{def:12-14-a}
\widetilde f^*(z):=
\begin{cases}
f^*(z),&0\le z\le\tau^*,\\[3pt]
1-1/e,&z>\tau^*.
\end{cases}
\end{align}
The threshold identity makes $\widetilde f^*$ continuous, and monotonicity
of $f^*$ implies $\widetilde f^*\in\cF_0$. For the remainder of this
subsection, write
\[
g^*(z):=1-\widetilde f^*(z),\qquad 0\le z\le1.
\]
Thus $g^*\in\mathcal G$ is the normalized transformed profile from
Appendix~\ref{app:gam-f}, and $g^*=1/e$ on $[\tau^*,1]$.
We use the same functional $\mathcal H_{g^*}$ as in
\eqref{eq:flat-fixed-test}.

\begin{lemma}\label{lem:12-14-a}
The constant extension~\eqref{def:12-14-a} satisfies
$\cL[\widetilde f^*]=\gam^*$.
\end{lemma}

The proof appears in Section~\ref{sec:proof-aid-9-13} below and uses
the post-threshold comparisons in Lemma~\ref{lem:11-29-a}.

\begin{proof}[Proof of Proposition~\ref{pro:low-bound}]
Since $\widetilde f^*\in\cF_0$, Lemma~\ref{lem:12-14-a} gives
\[
\Gam(\cF_0)
=\sup_{f\in\cF_0}\cL[f]
\ge\cL[\widetilde f^*]
=\gam^*.
\]
\end{proof}

\subsubsection{Proof of Lemma~\ref{lem:12-14-a}}
\label{sec:proof-aid-9-13}

\begin{proof}
We first check that $\widetilde f^*$ belongs to the subclass $\cF_2$.
Theorem~\ref{thm:main-structural} and monotonicity of $f^*$ give
\[
1-e^{-z}<f^*(z)\le1-1/e,
\qquad 0\le z\le\tau^*.
\]
For $\tau^*\le z\le1$, the constant continuation satisfies
$\widetilde f^*(z)=1-1/e\ge1-e^{-z}$, and
$\widetilde f^*(z)=1-1/e$ for all $z\ge1$. Together with
$\widetilde f^*\in\cF_0$, these properties establish
$\widetilde f^*\in\cF_2$.

Proposition~\ref{pro:2-17-a} therefore gives
\[
\cL[\widetilde f^*]
=\min\{W_1(\widetilde f^*),W_2(\widetilde f^*)\},
\]
where
\begin{align}
W_1(\widetilde f^*)
&=\int_0^1e^{-z}\widetilde f^*(z)\,\sd z
  +\frac1e(1-1/e),
\label{int:12-14-a}\\
W_2(\widetilde f^*)
&=\min_{0\le\ell\le\sig\le1}\mathcal H_{g^*}(\sig,\ell).
\nonumber
\end{align}
The second identity follows by setting $\sig=\ell+\tpsi$ and using
\[
\Psi_{\widetilde f^*}(\ell,\sig-\ell)
=\mathcal H_{g^*}(\sig,\ell),
\qquad 0\le\ell\le\sig\le1.
\]

\xhdr{Part 1: $W_2(\widetilde f^*)=\gam^*$.}
For $0\le\sig\le\tau^*$, the functions $\widetilde f^*$ and $f^*$
coincide on $[0,\sig]$. Flatness from
Theorem~\ref{thm:main-structural} therefore gives
\[
\min_{0\le\ell\le\sig}\mathcal H_{g^*}(\sig,\ell)
=\Psi_{f^*}(\sig)
=\gam^*,
\qquad 0\le\sig\le\tau^*.
\]
For $\sig\in[\tau^*,1]$, Lemma~\ref{lem:11-29-a} yields
\[
\min_{0\le\ell\le\sig}\mathcal H_{g^*}(\sig,\ell)
\ge\min_{0\le\ell\le\tau^*}\mathcal H_{g^*}(\tau^*,\ell)
=\gam^*.
\]
Indeed, a load $\ell\le\tau^*$ is compared directly with the same
load at $\sig=\tau^*$; a load $\ell>\tau^*$ is first reduced to
$\tau^*$ and then compared at $\sig=\tau^*$.
Thus extending the load domain beyond the critical interval cannot
lower the minimum, and $W_2(\widetilde f^*)=\gam^*$.

\xhdr{Part 2: $W_1(\widetilde f^*)=\gam^*$.}
Integrating the constant continuation in~\eqref{def:12-14-a} gives
\begin{align*}
W_1(\widetilde f^*)
&=\int_0^{\tau^*}e^{-z}f^*(z)\,\sd z
 +(1-1/e)\int_{\tau^*}^1e^{-z}\,\sd z
 +\frac1e(1-1/e)\\
&=\int_0^{\tau^*}e^{-z}f^*(z)\,\sd z
 +(1-1/e)e^{-\tau^*}\\
&=\Phi_1(f^*)=\gam^*,
\end{align*}
where the last equality is~\eqref{eqn:12-12-a} from
Theorem~\ref{thm:main-structural}. Consequently,
\[
\cL[\widetilde f^*]
=\min\{W_1(\widetilde f^*),W_2(\widetilde f^*)\}
=\gam^*.
\]
\end{proof}

\begin{lemma}[Post-threshold monotonicity]\label{lem:11-29-a}
For the constant extension~\eqref{def:12-14-a} and its transformed
profile $g^*=1-\widetilde f^*$ on $[0,1]$,
\begin{align*}
\mathcal H_{g^*}(\tau^*,\ell)
&\le\mathcal H_{g^*}(\sig,\ell),
&&0\le\ell\le\tau^*\le\sig\le1,\\
\mathcal H_{g^*}(\sig,\tau^*)
&\le\mathcal H_{g^*}(\sig,\ell),
&&\tau^*\le\ell\le\sig\le1.
\end{align*}
\end{lemma}

\begin{proof}
Since $g^*=1/e$ on $[\tau^*,1]$, the definition
of $\mathcal H_{g^*}$ gives, for $0\le\ell\le\tau^*\le\sig\le1$,
\[
\mathcal H_{g^*}(\sig,\ell)-\mathcal H_{g^*}(\tau^*,\ell)
=\int_{\tau^*}^{\sig}(e^{-z}-1/e)\,\sd z
\ge0.
\]
For $\tau^*\le\ell\le\sig\le1$, the same definition gives
\[
\mathcal H_{g^*}(\sig,\ell)-\mathcal H_{g^*}(\sig,\tau^*)
=\frac1e\int_{\tau^*}^{\ell}(1-e^{-z})\,\sd z
\ge0.
\]
These are the post-threshold comparisons used in the normalization
argument in Appendix~\ref{app:gam-f}; they require no differentiation
across $\tau^*$.
\end{proof}

\begin{proof}[Proof of Theorem~\ref{thm:main-exact}]
Propositions~\ref{pro:upp-bound} and~\ref{pro:low-bound} give
$\Gam(\cF_0)=\gam^*$. Moreover, Lemma~\ref{lem:12-14-a} shows that
the supremum defining $\Gam(\cF_0)$ is attained by the constant
extension $\widetilde f^*\in\cF_2\subseteq\cF_0$.
\end{proof}

\section{Proof of Theorem~\ref{thm:main-structural}}
\label{app:gam-f}

We prove the theorem by a comparison argument for
Program~\eqref{prog:var-gamma-F0}. The argument establishes attainment,
flatness, and regularity directly, without assuming a regular optimal
extremal. It does not use Theorem~\ref{thm:main-exact} or any numerical
bounds on the optimal value.

Derivatives at the endpoints of a critical interval are understood as
one-sided derivatives taken from within that interval. All inequalities
between profiles are understood pointwise.

\subsection{Normalization and a Fixed-Domain Formulation}

Define
\[
\mathcal G
:=
\{g\in C([0,1]): g\text{ is non-increasing and }1/e\le g\le1\}.
\]
For $g\in\mathcal G$, write
\begin{align}\label{def:12-12-a}
\tau(g):=\inf\{z\in[0,1]:g(z)=1/e\},
\end{align}
with the convention $\tau(g)=1$ if this set is empty. For
$0\le\ell\le\sig\le1$, define
\begin{align}
\mathcal A(g)
&:=1-1/e^2-\int_0^1 e^{-z}g(z)\,\sd z,
\label{eq:flat-fixed-scalar}\\
\mathcal H_g(\sig,\ell)
&:=1-e^{-\sig}
 +(1-\sig+\ell)g(\sig)
 -\int_0^\ell e^{-z}g(z)\,\sd z.
\label{eq:flat-fixed-test}
\end{align}
Consider the fixed-domain program
\begin{tcolorbox}[title=Fixed-Domain Formulation of
Program~\eqref{prog:var-gamma-F0}]
\begin{align}\label{pro:12-10-a}
\sup_{\gam\in\mathbb R,\,g\in\mathcal G}\quad &\gam\\
\text{s.t.}\quad
&\mathcal A(g)\ge\gam,\nonumber\\
&\mathcal H_g(\sig,\ell)\ge\gam,
\qquad 0\le\ell\le\sig\le1.\nonumber
\end{align}
\end{tcolorbox}

\begin{lemma}[Normalization]\label{lem:flat-normalization}
Programs~\eqref{prog:var-gamma-F0} and~\eqref{pro:12-10-a} have the same
supremal value $\gam^*$, and
\[
\frac12\le \gam^*\le 1-1/e<1.
\]
For every original feasible pair with value greater than $1/e$, replacing
$f$ by $\min\{f,1-1/e\}$ and setting
$g=1-\min\{f,1-1/e\}$ on $[0,1]$ preserves both constraint values.
Conversely, every $g\in\mathcal G$ corresponds to an admissible original
function with the same constraint values.
\end{lemma}

\begin{proof}
The constant function $f\equiv1/2$ is feasible in
Program~\eqref{prog:var-gamma-F0} at value $1/2$. Indeed,
\begin{align*}
\Phi_1(f)
&=\frac12+\frac{1}{2e}-\frac{1}{e^2}>\frac12,\\
\Phi_2(f)
&=\min_{0\le\sig\le1}
\left(1-e^{-\sig}+\frac{1-\sig}{2}\right)
=\frac12.
\end{align*}
Here the minimizing load is $\ell=0$, and the displayed function of
$\sig$ is concave, with endpoint values $1/2$ and $1-1/e$.

If $f(0)\ge 1-1/e$, then $\tau(f)=0$ and
$\Phi_2(f)=1-f(0)\le 1/e$. Such functions cannot improve the supremum.
Otherwise, truncating $f$ at $1-1/e$ preserves its critical interval and
its values on that interval, and hence preserves both functionals.
Set $g=1-\min\{f,1-1/e\}\in\mathcal G$ and
$\tau:=\tau(f)=\tau(g)$. Then
\[
\Phi_1(f)
=
1-e^{-(\tau+1)}-\int_0^\tau e^{-z}g(z)\,\sd z
=
\mathcal A(g).
\]
The last equality follows from $g=1/e$ after $\tau$ when $\tau<1$;
when $\tau=1$, it holds directly, whether or not the threshold is
attained. The expression minimized in $\Phi_2(f)$ is precisely
$\mathcal H_g(\sig,\ell)$ for
$0\le\ell\le\sig\le\tau$.

We verify that adding tests with $\sig>\tau$ does not decrease this
minimum. If $\tau<1$, then $g=1/e$ on $[\tau,1]$. For
$\tau\le\ell\le\sig\le1$,
\[
\frac{\partial}{\partial\ell}\mathcal H_g(\sig,\ell)
=(1/e)(1-e^{-\ell})\ge0.
\]
For $0\le\ell\le\tau\le\sig\le1$,
\begin{align}\label{eq:flat-post-threshold}
\mathcal H_g(\sig,\ell)-\mathcal H_g(\tau,\ell)
=
\int_\tau^\sig(e^{-z}-1/e)\,\sd z
\ge0.
\end{align}
Thus every additional test is at least the minimum of the tests at
$\sig=\tau$. Consequently,
\[
\Phi_2(f)
=
\min_{0\le\ell\le\sig\le1}\mathcal H_g(\sig,\ell).
\]

Conversely, given $g\in\mathcal G$, define $f=1-g$ on $[0,1]$ and
extend it constantly beyond $1$. The same identities apply, including
the case $g\equiv 1/e$, where $\tau(f)=0$. This proves the equivalence
of the supremal values. Finally,
$\mathcal H_g(1,0)=1-1/e$ for every $g\in\mathcal G$, proving
the upper bound.
\end{proof}

\subsection{A Volterra Comparison Principle}

Fix $\alpha\in(1/e,1)$. For $0<\eta<1$ and $u\in C([0,\eta])$, define
\begin{align}\label{eq:flat-volterra-operator}
(\mathcal T_\alpha u)(t)
:=
\max\left\{
1/e,\;
\max_{0\le s\le t}
\frac{\alpha-1+e^{-t}+\int_0^s e^{-z}u(z)\,\sd z}
     {1-t+s}
\right\},
\qquad 0\le t\le \eta.
\end{align}
The denominator is positive because $\eta<1$. The same formula is used
on different intervals: its value at $t$ depends only on the profile
on $[0,t]$. A feasible profile $g$ at value $\alpha$ satisfies
\[
g\ge\mathcal T_\alpha g
\quad\text{on every }[0,\eta],\ \eta<1,
\]
since this inequality is equivalent to the partial-load constraints
with $t\le \eta$, together with the lower bound $g\ge 1/e$.
The operator is order-preserving in both its profile and its parameter.
We use it only for $t<1$; constraints at $t=1$ are handled separately.

\begin{lemma}[Fixed point, comparison, and parameter continuity]
\label{lem:flat-volterra}
For every $\alpha\in(1/e,1)$, there is a unique continuous function
$y_\alpha:[0,1)\to\mathbb R$ whose restriction to each $[0,\eta]$, $\eta<1$,
is a fixed point of~\eqref{eq:flat-volterra-operator}. It satisfies
$y_\alpha(0)=\alpha$ and $y_\alpha\ge 1/e$.
If $g\in C([0,\eta])$ satisfies $g\ge\mathcal T_\alpha g$, then
$g\ge y_\alpha$ on $[0,\eta]$. Moreover, for $\alpha,\beta\in(1/e,1)$,
\begin{align}\label{eq:flat-parameter-continuity}
\|y_\alpha-y_\beta\|_{\infty,[0,\eta]}
\le
\frac{\exp(\eta/(1-\eta))}{1-\eta}\,|\alpha-\beta|.
\end{align}
The fixed points are also ordered in the parameter:
$y_\alpha\le y_\beta$ whenever $1/e<\alpha\le\beta<1$.
\end{lemma}

\begin{proof}
The substitution $s=tv$ expresses the inner maximum
in~\eqref{eq:flat-volterra-operator} over the fixed compact set
$v\in[0,1]$. Thus $\mathcal T_\alpha$ maps $C([0,\eta])$ into itself.
For $u,v\in C([0,\eta])$,
\[
|(\mathcal T_\alpha u)(t)-(\mathcal T_\alpha v)(t)|
\le (1-\eta)^{-1}\int_0^t|u(z)-v(z)|\,\sd z.
\]
For $k>0$, equip $C([0,\eta])$ with the weighted norm
\[
\|u\|_k:=\sup_{0\le t\le \eta}e^{-kt}|u(t)|.
\]
It is complete because
$e^{-k\eta}\|u\|_\infty\le\|u\|_k\le\|u\|_\infty$
and $C([0,\eta])$ is complete in the uniform norm. The preceding
Volterra estimate gives
\[
\|\mathcal T_\alpha u-\mathcal T_\alpha v\|_k
\le\frac{(1-\eta)^{-1}}{k}\|u-v\|_k.
\]
Taking $k=2(1-\eta)^{-1}$, Banach's fixed-point theorem gives a unique fixed
point on $[0,\eta]$. Restrictions of fixed points to smaller intervals
are fixed points there, so uniqueness makes these constructions
consistent. The initial value and lower bound follow directly from
the operator definition.

For comparison, positivity of the integral kernel gives the one-sided
estimate
\[
\bigl((\mathcal T_\alpha u)(t)
      -(\mathcal T_\alpha v)(t)\bigr)_+
\le (1-\eta)^{-1}\int_0^t (u(z)-v(z))_+\,\sd z,
\]
where $x_+:=\max\{x,0\}$. If $g\ge\mathcal T_\alpha g$, then
$d(t):=(y_\alpha(t)-g(t))_+$ consequently satisfies
\[
d(t)\le (1-\eta)^{-1}\int_0^t d(z)\,\sd z.
\]
Writing $D(t):=\int_0^t d(z)\,\sd z$, we have $D(0)=0$ and
$\frac{\sd D(t)}{\sd t}\le (1-\eta)^{-1}D(t)$. Hence $D=d=0$, proving comparison.

Finally, the fixed-point equations at parameters $\alpha$ and $\beta$
give
\[
|y_\alpha(t)-y_\beta(t)|
\le (1-\eta)^{-1}|\alpha-\beta|
+(1-\eta)^{-1}\int_0^t|y_\alpha(z)-y_\beta(z)|\,\sd z.
\]
The integral Gronwall inequality yields
\eqref{eq:flat-parameter-continuity}.
If $\alpha\le\beta$, then
$\mathcal T_\alpha y_\beta\le\mathcal T_\beta y_\beta=y_\beta$.
Applying the supersolution comparison on each $[0,\eta]$ gives
$y_\alpha\le y_\beta$ on $[0,1)$.
\end{proof}

\begin{lemma}[Regularity of the decreasing flat branch]
\label{lem:flat-regularity}
For every $\alpha\in(1/e,1)$, the fixed point $y_\alpha$ is locally
Lipschitz on $[0,1)$. On any interval where $y_\alpha(t)<e^{-t}$,
it is non-increasing, including at points where $y_\alpha=1/e$.
It is strictly decreasing on every subinterval where additionally
$y_\alpha>1/e$.

Suppose $b\in(0,1]$ and
\[
1/e<y_\alpha(t)<e^{-t},\qquad 0\le t<b.
\]
Then $y_\alpha\in C^1([0,b))$. For every $t\in[0,b)$ there is a
unique minimizing load $\ell_\alpha(t)$, with $\ell_\alpha(0)=0$ and
$0<\ell_\alpha(t)<t$ for $t>0$, such that
\begin{align}
y_\alpha(t)
&=e^{-\ell_\alpha(t)}y_\alpha(\ell_\alpha(t)),
\label{eq:flat-branch-implicit}\\
\min_{0\le s\le t}\mathcal H_{y_\alpha}(t,s)
&=\alpha,
\label{eq:flat-branch-envelope}\\
\frac{\sd y_\alpha(t)}{\sd t}
&=\frac{y_\alpha(t)-e^{-t}}{1-t+\ell_\alpha(t)}.
\label{eq:flat-branch-ode}
\end{align}
Here $\mathcal H$ is interpreted by its defining formula on the
interval where the profile is defined. If $b<1$, these identities
extend to $t=b$, and the restriction to $[0,b]$ belongs to $C^1$
with the left derivative at $b$. This does not assert differentiability
across $b$ for a subsequent continuation.
\end{lemma}

\begin{proof}
Fix $0<\eta<1$ and abbreviate $y=y_\alpha$. Define
\[
R(t,s):=
\frac{\alpha-1+e^{-t}+\int_0^s e^{-z}y(z)\,\sd z}{1-t+s},
\qquad 0\le s\le t\le \eta.
\]
We first establish regularity without assuming that $y$ is monotone
or differentiable. For $\theta\in[0,1]$, put
$F(t,\theta):=R(t,t\theta)$. Since $y$ is continuous, the integral
primitive appearing in $R$ is $C^1$, and
\[
\partial_t F(t,\theta)
=
\frac{-e^{-t}+\theta e^{-t\theta}y(t\theta)
      +(1-\theta)F(t,\theta)}{1-t+t\theta}.
\]
This derivative is continuous and uniformly bounded on
$[0,\eta]\times[0,1]$, because the denominator is at least $1-\eta$.
Thus the functions $F(\cdot,\theta)$ have a common Lipschitz bound.
Their maximum together with the constant $1/e$ is Lipschitz as well.
By the fixed-point equation, this maximum is $y$. In particular,
$y$ is absolutely continuous on $[0,\eta]$ and differentiable almost
everywhere there.

Let $t\in(0,\eta)$ satisfy $y(t)>1/e$. A maximizing load $s=t$ is
impossible: if $R(t,t)=y(t)$, then
\[
\partial_sR(t,t)=(e^{-t}-1)y(t)<0,
\]
so decreasing $s$ slightly would increase $R(t,s)$.
Hence every maximizing load satisfies $s<t$.
At a differentiability point of $y$, fix one such load $s$.
Because $s<t$, it remains admissible for every $v$ in a neighborhood
of $t$. The function $y(v)-R(v,s)$ is nonnegative there and vanishes
at $v=t$. Differentiating at this local minimum gives
\begin{align}\label{eq:flat-ae-slope}
\frac{\sd y(t)}{\sd t}=\partial_tR(t,s)
=\frac{y(t)-e^{-t}}{1-t+s}.
\end{align}
This argument holds the maximizing load fixed; it does not
differentiate a choice of maximizer.
At an interior differentiability point where $y(t)=1/e$, we instead
have $\frac{\sd y(t)}{\sd t}=0$, since $y\ge 1/e$.

Consequently, on any interval where $y(t)<e^{-t}$, we have $\frac{\sd y}{\sd t}\le0$
almost everywhere. Absolute continuity then makes $y$ non-increasing
on that interval. On every subinterval where also $y>1/e$, the derivative
in~\eqref{eq:flat-ae-slope} is strictly negative almost everywhere;
integrating it shows that $y$ is strictly decreasing.

Now consider the initial branch $[0,b)$ in the statement. For $t>0$
on this branch, the function $s\mapsto e^{-s}y(s)$ is continuous and
strictly decreasing on $[0,t]$. Moreover,
\[
\partial_s\mathcal H_y(t,s)=y(t)-e^{-s}y(s)
\]
is negative at $s=0$ and positive at $s=t$. Thus it has a unique zero
$\ell_\alpha(t)\in(0,t)$, which is the unique minimizer and satisfies
\eqref{eq:flat-branch-implicit}. At $t=0$, the only load is $0$.

The relation
\[
\mathcal H_y(t,s)-\alpha
=(1-t+s)\bigl(y(t)-R(t,s)\bigr)
\]
and the fixed-point equation imply that every partial-load test is at
least $\alpha$. Since $y(t)>1/e$, at least one maximizing load attains
$R(t,s)=y(t)$. Hence the minimum test equals $\alpha$, proving
\eqref{eq:flat-branch-envelope}. It also follows that every maximizing
load in~\eqref{eq:flat-ae-slope} equals $\ell_\alpha(t)$.

Fix $B\in(0,b)$ and put $p(s):=e^{-s}y(s)$ on $[0,B]$.
This function has a continuous inverse on its image, and
\[
\ell_\alpha(t)=p^{-1}(y(t)),\qquad 0\le t\le B.
\]
Thus $\ell_\alpha$ is continuous. The function
\[
v(t):=\frac{y(t)-e^{-t}}{1-t+\ell_\alpha(t)}
\]
is continuous on $[0,B]$. We have already proved that $y$ is
absolutely continuous and that $\frac{\sd y}{\sd t}=v$ almost everywhere. Therefore
\[
y(t)=\alpha+\int_0^t v(z)\,\sd z,\qquad 0\le t\le B.
\]
The fundamental theorem of calculus now gives $y\in C^1([0,B])$
and~\eqref{eq:flat-branch-ode} pointwise, including the right derivative
$\left.\frac{\sd y_\alpha(t)}{\sd t}\right|_{t=0}=\alpha-1$. Since $B<b$ was arbitrary, this establishes
the claimed regularity on $[0,b)$.

If $b<1$, continuity gives $y(b)\ge 1/e>0$ and $y(b)<\alpha$.
The latter strict inequality follows, for example, by fixing any
$t_0\in(0,b)$ and using $y(b)\le y(t_0)<y(0)=\alpha$.
The function $p(s)=e^{-s}y(s)$ on $[0,b]$ remains continuous and
strictly decreasing, and
$p(b)<y(b)<p(0)$.
Its inverse therefore defines a unique $\ell_\alpha(b)\in(0,b)$,
and $\ell_\alpha(t)\to\ell_\alpha(b)$ as $t\uparrow b$.
Continuity gives the implicit and flatness identities at $b$.
The right-hand side of~\eqref{eq:flat-branch-ode} extends continuously
to $b$, with denominator at least $1-b>0$.
Passing to the limit in the integral identity above proves the
one-sided $C^1$ extension to $[0,b]$.
\end{proof}

\subsection{Feasibility Forces an Early Threshold Crossing}

\begin{lemma}[Canonical profile of a feasible value]
\label{lem:flat-canonical-profile}
Suppose $\alpha>1/e$ is feasible for~\eqref{pro:12-10-a}. Then the
decreasing branch of $y_\alpha$ reaches $1/e$ at a time
$\tau_\alpha\in(0,1)$ before reaching the curve $e^{-t}$.
Define
\[
g_\alpha(t):=
\begin{cases}
y_\alpha(t),&0\le t\le\tau_\alpha,\\
1/e,&\tau_\alpha\le t\le1.
\end{cases}
\]
Then $g_\alpha\in\mathcal G$, and
\begin{align}
\mathcal H_{g_\alpha}(t,s)&\ge\alpha,
&&0\le s\le t\le1,
\label{eq:flat-canonical-feasible}\\
\min_{0\le s\le t}\mathcal H_{g_\alpha}(t,s)&=\alpha,
&&0\le t\le\tau_\alpha.
\label{eq:flat-canonical-envelope}
\end{align}
For every feasible pair $(\alpha,g)$,
\[
g_\alpha\le g,
\qquad
\mathcal A(g_\alpha)\ge\mathcal A(g)\ge\alpha.
\]
Moreover, $g_\alpha=y_\alpha$ on $[0,1)$,
$g_\alpha(t)<e^{-t}$ for $t<1$, and its restriction to
$[0,\tau_\alpha]$ is $C^1$ with strictly negative derivative,
including the one-sided derivatives at the endpoints.
\end{lemma}

\begin{proof}
Fix a feasible pair $(\alpha,g)$. By
Lemma~\ref{lem:flat-normalization}, $\alpha\le 1-1/e<1$.
Lemma~\ref{lem:flat-volterra} gives $y:=y_\alpha\le g$ on $[0,1)$.
Since $1/e<y(0)=\alpha<1$, there is a nonempty initial branch with
$1/e<y(t)<e^{-t}$. Let $b\in(0,1]$ be its maximal endpoint.
If $b<1$, continuity implies that $y(b)=1/e$ or $y(b)=e^{-b}$.
If $b=1$, the bounds $1/e<y(t)<e^{-t}$ force $y(t)\to 1/e$ as
$t\uparrow1$; in this case, extend the branch continuously by
$y(1)=1/e$.

For $0\le t\le b$, define
\[
\mathcal A_t(y)
:=1-e^{-(t+1)}-\int_0^t e^{-z}y(z)\,\sd z.
\]
Comparison on $[0,t]$ and the lower bound $g\ge 1/e$ on the remaining
interval give
\begin{align}\label{eq:flat-scalar-comparison}
\alpha\le\mathcal A(g)\le\mathcal A_t(y).
\end{align}
More explicitly,
\[
\mathcal A_t(y)-\mathcal A(g)
=
\int_0^t e^{-z}(g(z)-y(z))\,\sd z
+
\int_t^1 e^{-z}(g(z)-1/e)\,\sd z
\ge0.
\]
This identity also holds at $t=b$, using the continuous extension of
$y$ when $b=1$. In particular, it does not require evaluating the
Volterra operator at $t=1$.

We exclude the possibility $y(b)=e^{-b}$, including the case $b=1$.
The branch is strictly decreasing before $b$, so $y(b)<\alpha$.
The continuous strictly decreasing function $s\mapsto e^{-s}y(s)$
on $[0,b]$ therefore has a unique solution $\ell_b\in(0,b)$ to
\[
e^{-\ell_b}y(\ell_b)=y(b).
\]
Indeed, with $p(s):=e^{-s}y(s)$ on $[0,b]$, we have
$p(b)<y(b)<p(0)$, so the root lies strictly inside $(0,b)$.
For $t<b$, the minimizing load equals $p^{-1}(y(t))$.
The continuity of this inverse gives
$\ell_\alpha(t)\to\ell_b$ as $t\uparrow b$, also when $b=1$.
The defining expression for $\mathcal H_y$ is continuous on the closed
triangle $0\le s\le t\le b$. Passing to the limit in
\eqref{eq:flat-branch-envelope} therefore yields
\[
\alpha=\mathcal H_y(b,\ell_b).
\]
Using $y(b)=e^{-b}$ and monotonicity of $y$, we obtain
\begin{align*}
\alpha-\mathcal A_b(y)
&=e^{-b}(1/e+\ell_b-b)
  +\int_{\ell_b}^{b}e^{-z}y(z)\,\sd z\\
&\ge e^{-b}\bigl(1/e+\ell_b-b+e^{-\ell_b}-e^{-b}\bigr)\\
&=e^{-b}\bigl[
   (\ell_b+e^{-\ell_b}-1)
   +(1+1/e-b-e^{-b})
   \bigr]
>0.
\end{align*}
Indeed, the first parenthesis is strictly positive because
$\ell_b>0$, whereas the second is nonnegative for $0\le b\le1$:
the function $t\mapsto1+1/e-t-e^{-t}$ is non-increasing and vanishes
at $t=1$. This contradicts~\eqref{eq:flat-scalar-comparison}.

Thus $b=\tau_\alpha<1$ and $y(\tau_\alpha)=1/e<e^{-\tau_\alpha}$.
Extend this branch constantly by $1/e$. The fixed-point equation and
continuity give~\eqref{eq:flat-canonical-feasible} and
\eqref{eq:flat-canonical-envelope} up to $\tau_\alpha$.
The post-threshold comparisons in the proof of
Lemma~\ref{lem:flat-normalization} give all the remaining tests.
Consequently, $g_\alpha\in\mathcal G$ satisfies every partial-load
constraint at level $\alpha$.

Before the threshold, $g_\alpha\le g$ follows from the Volterra
comparison. Afterward it follows from $g_\alpha=1/e\le g$.
The comparison of the scalar functionals follows
from~\eqref{eq:flat-fixed-scalar}.

We also identify the constant extension with the full fixed point.
Before the threshold, $\mathcal T_\alpha g_\alpha=g_\alpha$ by the
fixed-point equation and its dependence only on earlier values.
After the threshold, all partial constraints hold and $g_\alpha=1/e$,
so every ratio in~\eqref{eq:flat-volterra-operator} is at most $1/e$;
therefore $\mathcal T_\alpha g_\alpha=1/e=g_\alpha$ there as well.
Uniqueness in Lemma~\ref{lem:flat-volterra} gives
$g_\alpha=y_\alpha$ on $[0,1)$.

Finally, $g_\alpha(t)<e^{-t}$ holds on the critical branch by its
construction, at the threshold because $\tau_\alpha<1$, and afterward
because $g_\alpha=1/e<e^{-t}$ for $t<1$. The regularity and strict
derivative sign on the critical interval follow from
Lemma~\ref{lem:flat-regularity} and
\eqref{eq:flat-branch-ode}, whose denominator is at least
$1-\tau_\alpha>0$.
\end{proof}

\subsection{Attainment, Saturation, and Uniqueness}

\begin{lemma}[Attainment]\label{lem:flat-attainment}
Program~\eqref{pro:12-10-a} attains its supremal value $\gam^*$.
\end{lemma}

\begin{proof}
Choose feasible values $\alpha_n\ge1/2$ tending to $\gam^*$ and replace
their profiles by the canonical profiles $g_{\alpha_n}$ from
Lemma~\ref{lem:flat-canonical-profile}. Since
$\gam^*\in[1/2,1-1/e]\subset(1/e,1)$, the fixed point $y_{\gam^*}$ exists by
Lemma~\ref{lem:flat-volterra}; its existence does not presuppose
feasibility at value $\gam^*$.
By~\eqref{eq:flat-parameter-continuity},
\[
g_{\alpha_n}\longrightarrow y_{\gam^*}
\quad\text{uniformly on every }[0,\eta],\ \eta<1.
\]
The limit is continuous and non-increasing on $[0,1)$.
Moreover, $1/e\le g_{\alpha_n}(t)\le e^{-t}$ on $[0,1]$, so
$1/e\le y_{\gam^*}(t)\le e^{-t}$ for $t<1$. Hence $y_{\gam^*}(t)\to 1/e$ as
$t\uparrow1$. Define $g=y_{\gam^*}$ on $[0,1)$ and $g(1)=1/e$.
Then $g\in\mathcal G$.

In fact the convergence is uniform on the full interval $[0,1]$.
For every $\eta<1$ and $t\in[\eta,1]$, both $g_{\alpha_n}(t)$ and $g(t)$
belong to $[1/e,e^{-t}]$. Thus
\[
\|g_{\alpha_n}-g\|_{\infty,[0,1]}
\le
\max\left\{
\|g_{\alpha_n}-g\|_{\infty,[0,\eta]},\ e^{-\eta}-1/e
\right\}.
\]
First choose $\eta$ sufficiently close to $1$, then let $n\to\infty$.
This proves uniform convergence on $[0,1]$.

The scalar functional and every partial-load test now pass to the
limit. Explicitly, for $u,v\in C([0,1])$,
\begin{align*}
|\mathcal A(u)-\mathcal A(v)|
&\le(1-1/e)\|u-v\|_\infty,\\
\sup_{0\le s\le t\le1}
|\mathcal H_u(t,s)-\mathcal H_v(t,s)|
&\le(2-1/e)\|u-v\|_\infty.
\end{align*}
Consequently,
\[
\mathcal A(g)\ge \gam^*,
\qquad
\mathcal H_g(t,s)\ge \gam^*,
\quad 0\le s\le t\le1.
\]
Thus $(\gam^*,g)$ is feasible, proving attainment. No compactness
assumption on the entire class $\mathcal G$ is needed.
\end{proof}

\begin{lemma}[Stability near a feasible value]\label{lem:flat-stability}
Let $\alpha>1/e$ be feasible for~\eqref{pro:12-10-a}. For every
$\beta\in(1/e,1)$ sufficiently close to $\alpha$, the decreasing branch
of $y_\beta$ also reaches $1/e$ before time $1$ and before reaching
$e^{-t}$. Its constant extension $g_\beta\in\mathcal G$ satisfies
all partial-load constraints at level $\beta$. Moreover, the return value
\[
q(\beta):=\mathcal A(g_\beta)
=
\int_0^{\tau_\beta}e^{-z}(1-y_\beta(z))\,\sd z
+(1-1/e)e^{-\tau_\beta}
\]
is continuous at $\alpha$.
No feasibility assumption on the scalar constraint at value $\beta$
is needed for this construction.
\end{lemma}

\begin{proof}
Choose $\eta\in(\tau_\alpha,1)$. By
Lemma~\ref{lem:flat-canonical-profile},
\[
\delta:=\min_{0\le t\le \eta}\bigl(e^{-t}-g_\alpha(t)\bigr)>0.
\]
Parameter continuity~\eqref{eq:flat-parameter-continuity} and the
identity $g_\alpha=y_\alpha$ on $[0,1)$ imply that, for $\beta$
sufficiently close to $\alpha$,
\[
y_\beta(t)<e^{-t}-\delta/2,
\qquad 0\le t\le \eta.
\]
The derivative information in the proof of
Lemma~\ref{lem:flat-regularity} shows that $y_\beta$ is non-increasing
on $[0,\eta]$. At almost every point where $y_\beta(t)>1/e$,
\[
\frac{\sd y_\beta(t)}{\sd t}
=\frac{y_\beta(t)-e^{-t}}{1-t+s}
\le-\delta/2,
\]
where $s\in[0,t)$ is a maximizing load and $1-t+s\le1$.
At differentiability points where $y_\beta(t)=1/e$, its derivative
is zero.

Since $y_\beta(\tau_\alpha)\to 1/e$, we may also require
\[
y_\beta(\tau_\alpha)
<1/e+\frac{\delta}{4}(\eta-\tau_\alpha).
\]
If the level $1/e$ were not reached before $\eta$, integrating the slope
bound on $[\tau_\alpha,\eta]$ would give
\[
y_\beta(\eta)
\le y_\beta(\tau_\alpha)
   -\frac{\delta}{2}(\eta-\tau_\alpha)
<1/e,
\]
a contradiction. Thus $y_\beta$ reaches $1/e$ at some
$\tau_\beta\in(0,\eta)$ while remaining strictly below $e^{-t}$.

Extend this branch constantly by $1/e$. Flatness on the branch and
the post-threshold comparisons give all partial constraints at value
$\beta$, independently of the scalar constraint. The extension lies
in $\mathcal G$, since it is non-increasing and starts at
$\beta\in(1/e,1)$. As in the proof of
Lemma~\ref{lem:flat-canonical-profile}, the extension is a fixed point
of $\mathcal T_\beta$ on every $[0,\eta']$, $\eta'<1$, and hence agrees
with $y_\beta$ on $[0,1)$.

All these profiles equal $1/e$ beyond $\eta$. Therefore
\[
|q(\beta)-q(\alpha)|
\le
(1-e^{-\eta})\,
\|y_\beta-y_\alpha\|_{\infty,[0,\eta]}
\longrightarrow0,
\]
proving continuity. The identity with the stopping-time expression
for $q(\beta)$ follows by integrating the constant tail of $g_\beta$.
\end{proof}

\begin{proof}[Proof of Theorem~\ref{thm:main-structural}]
Lemmas~\ref{lem:flat-normalization} and~\ref{lem:flat-attainment}
show that Program~\eqref{prog:var-gamma-F0} attains its supremal
value $\gam^*$. Let $(\gam^*,g^*)$ be any optimal
normalized solution of~\eqref{pro:12-10-a}.
Since $\gam^*\ge1/2>1/e$, its canonical profile $g_{\gam^*}$ exists.
By Lemma~\ref{lem:flat-canonical-profile},
\[
g_{\gam^*}\le g^*,
\qquad
q(\gam^*):=\mathcal A(g_{\gam^*})
\ge\mathcal A(g^*)\ge\gam^*.
\]

Suppose $q(\gam^*)>\gam^*$. By
Lemma~\ref{lem:flat-stability}, for $\beta>\gam^*$ sufficiently close
to $\gam^*$ there is a profile $g_\beta\in\mathcal G$ satisfying
every partial-load constraint at level $\beta$. Such parameters
remain in $(1/e,1)$ because $\gam^*\le 1-1/e<1$.
Since $q(\gam^*)-\gam^*>0$ and $q$ is continuous there,
we may also choose $\beta$ so that $q(\beta)>\beta$.
Thus $(\beta,g_\beta)$ satisfies the scalar constraint as well,
contradicting optimality. Consequently,
\[
\mathcal A(g_{\gam^*})=\gam^*.
\]
It follows that
\[
\gam^*
\le\mathcal A(g^*)
\le\mathcal A(g_{\gam^*})
=\gam^*.
\]
Using~\eqref{eq:flat-fixed-scalar}, we obtain
\[
0
=
\int_0^1e^{-z}\bigl(g^*(z)-g_{\gam^*}(z)\bigr)\,\sd z.
\]
The integrand is continuous and nonnegative. Hence
\[
g^*=g_{\gam^*}
\quad\text{pointwise on }[0,1].
\]
Thus the normalized optimal profile is unique, and its scalar
constraint is tight.

Now let $f^*$ be any optimal solution of the original program.
Since its value exceeds $1/e$, Lemma~\ref{lem:flat-normalization}
applies, and
\[
1-\min\{f^*(z),1-1/e\}=g_{\gam^*}(z),
\qquad 0\le z\le1.
\]
The canonical profile reaches $1/e$ for the first time at
$\tau_{\gam^*}\in(0,1)$, so
$\tau(f^*)=\tau_{\gam^*}=: \tau^*$.
For $z<\tau^*$ we have $g_{\gam^*}(z)>1/e$, so the normalization
identity gives $f^*(z)=1-g_{\gam^*}(z)<1-1/e$.
Continuity then yields $f^*(\tau^*)=1-1/e$ and
$f^*=1-g_{\gam^*}$ throughout $[0,\tau^*]$.
The scalar identity and~\eqref{eq:flat-canonical-envelope} therefore
translate into
\[
\Phi_1(f^*)=\gam^*,
\qquad
\Psi_{f^*}(\sig)=\gam^*,
\qquad 0\le\sig\le\tau^*.
\]
These are~\eqref{eqn:12-12-a} and~\eqref{eqn:12-12-b}.

Finally, Lemma~\ref{lem:flat-canonical-profile} gives
\[
0<\tau^*<1,
\qquad
f^*(\tau^*)=1-1/e,
\qquad
1-f^*(z)<e^{-z}\quad(0\le z\le\tau^*).
\]
It also shows that the critical restriction of $f^*$ belongs to
$C^1([0,\tau^*])$ and has strictly positive derivative throughout
that interval, with the one-sided interpretation at its endpoints.
All optimal solutions have the same critical interval and coincide
there, because their normalized profiles are identical. No uniqueness
or differentiability across $\tau^*$ is asserted for the continuation
of the original function beyond its critical interval.
\end{proof}


\section{Proof of Corollary~\ref{cor:main-struc}}
\label{app:cor-main-struc}


\begin{proof}
For ease of notation, throughout this proof we omit the superscript $*$
and write $f$, $\ell$, and $\gam$ for the objects associated with an
optimal solution to Program~\eqref{prog:var-gamma-F0}. Set
$\tau:=\tau(f)$.

By Theorem~\ref{thm:main-structural}, we have $0<\tau<1$, the restriction
of $f$ to $[0,\tau]$ belongs to $C^1([0,\tau])$, and $f$ is strictly
increasing on this interval. Moreover,
\[
f(\tau)=1-e^{-1},
\qquad
1-f(z)<e^{-z},
\qquad \forall z\in[0,\tau].
\]
Throughout the proof, derivatives at $0$ and $\tau$ are understood as
the right and left derivatives, respectively, taken from within the
critical interval.

Equality~\eqref{eqn:12-13-c} is a restatement
of~\eqref{eqn:ell-sig} with $\sig$ replaced by $z$. In particular,
$\ell(0)=0$.

For~\eqref{eqn:12-13-d}, flatness~\eqref{eqn:12-12-b} gives
\[
\Psi_f(0)=\gam.
\]
Since $\ell(0)=0$, the definition of $\Psi_f$ yields
\[
\Psi_f(0)=1-f(0),
\]
and hence $\gam=1-f(0)$. The second equality
in~\eqref{eqn:12-13-d} follows from~\eqref{eqn:12-12-a}, namely,
\[
\gam=\Phi_1(f)
=
\int_0^\tau e^{-z}f(z)\,\sd z
+
e^{-\tau}(1-e^{-1}).
\]

We now prove~\eqref{eqn:12-13-b} as a pointwise differential identity.
First, we establish the regularity of $\ell$. Define
\[
H(z):=e^{-z}\bigl(1-f(z)\bigr),
\qquad z\in[0,\tau].
\]
Since $f\in C^1([0,\tau])$ is non-decreasing, $f'(z)\ge0$ throughout
the interval. Also, $f(z)\le f(\tau)=1-e^{-1}$, so
$1-f(z)\ge e^{-1}>0$. Consequently, $H\in C^1([0,\tau])$ and
\[
\frac{\sd H}{\sd z}
=
-e^{-z}\bigl(1-f(z)+\dot f(z)\bigr)
<0,
\qquad \forall z\in[0,\tau].
\]
Thus $H$ is strictly decreasing, and its inverse is continuously
differentiable on its image, with one-sided derivatives at the
endpoints. By~\eqref{eqn:12-13-c},
\[
1-f(z)=H(\ell(z)),
\]
and therefore
\[
\ell(z)=H^{-1}\bigl(1-f(z)\bigr),
\qquad z\in[0,\tau].
\]
It follows that $\ell\in C^1([0,\tau])$. In particular,
\[
\frac{\sd \ell}{\sd z}
=
\frac{\dot{f}(z)}
{e^{-\ell(z)}
 \bigl(1-f(\ell(z))+\dot{f}(\ell(z))\bigr)},
\qquad z\in[0,\tau],
\]
whose denominator is strictly positive.

We also verify the asserted bounds on $\ell$. For every
$z\in(0,\tau]$, strict monotonicity of $f$ and positivity of $1-f(z)$
give
\[
H(z)
=
e^{-z}\bigl(1-f(z)\bigr)
<
1-f(z)
<
1-f(0)
=
H(0).
\]
Since $H$ is strictly decreasing and $H(\ell(z))=1-f(z)$, these
inequalities imply
\[
0<\ell(z)<z,
\qquad \forall z\in(0,\tau].
\]

Both $f$ and $\ell$ are now known to be continuously differentiable
on the critical interval. Hence we may differentiate the flatness
identity
\[
\Psi_f(z)
=
\int_0^{\ell(z)}e^{-u}f(u)\,\sd u
+
\int_{\ell(z)}^z e^{-u}\,\sd u
+
\bigl(1-z+\ell(z)\bigr)\bigl(1-f(z)\bigr)
\equiv\gam
\]
at every $z\in(0,\tau)$. This yields
\begin{align*}
0
&=
e^{-\ell(z)}f(\ell(z))\,\frac{\sd\ell(z)}{\sd z}
+
e^{-z}
-
e^{-\ell(z)}\,\frac{\sd\ell(z)}{\sd z}
\\
&\quad
+
\left(-1+\frac{\sd\ell(z)}{\sd z}\right)
\bigl(1-f(z)\bigr)
-
\bigl(1-z+\ell(z)\bigr)\frac{\sd f(z)}{\sd z}.
\end{align*}
Using~\eqref{eqn:12-13-c}, namely,
\[
1-f(z)
=
e^{-\ell(z)}\bigl(1-f(\ell(z))\bigr),
\]
the terms involving $\sd\ell(z)/\sd z$ cancel. Therefore,
\[
0
=
e^{-z}
-
\bigl(1-f(z)\bigr)
-
\bigl(1-z+\ell(z)\bigr)\frac{\sd f(z)}{\sd z}.
\]
Since $\tau<1$ and $\ell(z)\ge0$, we have
\[
1-z+\ell(z)\ge1-\tau>0,
\qquad \forall z\in[0,\tau].
\]
Rearranging gives
\[
\frac{\sd f(z)}{\sd z}
=
\frac{e^{-z}-(1-f(z))}{1-z+\ell(z)},
\qquad z\in(0,\tau).
\]
Both sides are continuous on $[0,\tau]$, so the identity extends to
the endpoints, with the right derivative at $0$ and the left
derivative at $\tau$. Finally,
$1-f(z)<e^{-z}$ throughout $[0,\tau]$ by
Theorem~\ref{thm:main-structural}. Thus
\[
\frac{\sd f(z)}{\sd z}
=
\frac{e^{-z}-(1-f(z))}{1-z+\ell(z)}
>0,
\qquad \forall z\in[0,\tau],
\]
which proves~\eqref{eqn:12-13-b} and completes the proof.
\end{proof}


\section{Proof of Lemma~\ref{lem:q-monotone}}
\label{app:q-monotone}

\begin{proof}[Proof of Lemma~\ref{lem:q-monotone}]
For each $i\in\{1,2\}$, let $f_i$, $\ell_i$, and $\tau_i$ be as in
the statement. In particular, $f_i,\ell_i\in C^1([0,\tau_i])$,
$0<\tau_i<1$, and
\[
f_i(0)=1-\alp_i,
\qquad
f_i(\tau_i)=1-1/e,
\qquad
f_i(z)<1-1/e\quad(0\le z<\tau_i).
\]
When applying Lemma~\ref{lem:11-29-c}, we regard $f_i$ as extended
constantly by $1-1/e$ after $\tau_i$. This extension belongs to
$\cF_0$ and satisfies $\tau(f_i)=\tau_i$. All differentiation below
is confined to the interior of the original critical interval; no
differentiability across $\tau_i$ is needed.

For $\sig\in[0,\tau_i]$, equation~\eqref{eq:num-alpha-a} and
Lemma~\ref{lem:11-29-c} identify $\ell_i(\sig)$ as the unique minimizer
of $\Psi_{f_i,\sig}$ over $[0,\sig]$. Define
\begin{align*}
\Psi_i(\sig)
&:=\Psi_{f_i,\sig}(\ell_i(\sig))
=\min_{0\le\ell\le\sig}\Psi_{f_i,\sig}(\ell)\\
&=\int_0^{\ell_i(\sig)}e^{-z}f_i(z)\,\sd z
 +\int_{\ell_i(\sig)}^{\sig}e^{-z}\,\sd z
 +(1-\sig+\ell_i(\sig))(1-f_i(\sig)).
\end{align*}

\medskip
\noindent
\textbf{Step 1: Constancy of $\Psi_i$.}
We claim that
\begin{equation}\label{eq:psi-alpha-const}
\Psi_i(\sig)\equiv\alp_i,
\qquad
\sig\in[0,\tau_i],\quad i\in\{1,2\}.
\end{equation}
Because $f_i$ and $\ell_i$ are $C^1$ on their critical interval, we
may differentiate at every $\sig\in(0,\tau_i)$ to obtain
\begin{align*}
\frac{\sd\Psi_i(\sig)}{\sd\sig}
&=e^{-\sig}-(1-f_i(\sig))
 -(1-\sig+\ell_i(\sig))\frac{\sd f_i(\sig)}{\sd\sig}\\
&\quad+
 \frac{\sd\ell_i(\sig)}{\sd\sig}
 \bigl[1-f_i(\sig)
       -e^{-\ell_i(\sig)}(1-f_i(\ell_i(\sig)))\bigr]\\
&=0.
\end{align*}
The bracket multiplying $\sd\ell_i/\sd\sig$ vanishes
by~\eqref{eq:num-alpha-a}, and the remaining terms vanish
by~\eqref{eq:num-alpha-b}. Hence $\Psi_i$ is constant on
$(0,\tau_i)$ and, by continuity, on $[0,\tau_i]$.
Since $\ell_i(0)=0$ and $f_i(0)=1-\alp_i$, we have
$\Psi_i(0)=1-f_i(0)=\alp_i$, proving~\eqref{eq:psi-alpha-const}.

\medskip
\noindent
\textbf{Step 2: Strict ordering of trajectories.}
We show that
\begin{equation}\label{eq:trajectory-order}
f_1(z)>f_2(z),
\qquad
0\le z\le\min\{\tau_1,\tau_2\}.
\end{equation}
The initial values satisfy
$f_1(0)=1-\alp_1>1-\alp_2=f_2(0)$.
If~\eqref{eq:trajectory-order} failed, continuity would give a first
contact point $\bar z\in(0,\min\{\tau_1,\tau_2\}]$ such that
\[
f_1(\bar z)=f_2(\bar z),
\qquad
f_1(z)>f_2(z)\quad(0\le z<\bar z).
\]
For every $\ell\in[0,\bar z]$, the equality at $\bar z$ makes the
last two terms in the definitions of $\Psi_{f_1,\bar z}(\ell)$ and
$\Psi_{f_2,\bar z}(\ell)$ identical. Therefore
\[
\Psi_{f_1,\bar z}(\ell)-\Psi_{f_2,\bar z}(\ell)
=\int_0^\ell e^{-z}\bigl(f_1(z)-f_2(z)\bigr)\,\sd z
\ge0.
\]
Taking minima over the common load interval and
using~\eqref{eq:psi-alpha-const}, we obtain
\begin{align*}
\alp_1
&=\min_{0\le\ell\le\bar z}\Psi_{f_1,\bar z}(\ell)\ge\min_{0\le\ell\le\bar z}\Psi_{f_2,\bar z}(\ell)
=\alp_2,
\end{align*}
contradicting $\alp_1<\alp_2$. This proves~\eqref{eq:trajectory-order}.

\medskip
\noindent
\textbf{Step 3: Ordering of stopping times.}
We next prove
\begin{equation}\label{eq:stopping-order}
\tau_1<\tau_2.
\end{equation}
If $\tau_2\le\tau_1$, then~\eqref{eq:trajectory-order}, monotonicity
of $f_1$, and the threshold identities give
\[
1-1/e=f_2(\tau_2)
<f_1(\tau_2)
\le f_1(\tau_1)=1-1/e,
\]
a contradiction. Hence~\eqref{eq:stopping-order} holds.

\medskip
\noindent
\textbf{Step 4: Comparison of $q_1$ and $q_2$.}
Using~\eqref{eq:stopping-order}, split the integral defining $q_2$
at $\tau_1$ and combine the terminal terms to obtain
\begin{align*}
q_1-q_2
&=\int_0^{\tau_1}e^{-z}\bigl(f_1(z)-f_2(z)\bigr)\,\sd z+
\int_{\tau_1}^{\tau_2}
 e^{-z}\bigl((1-1/e)-f_2(z)\bigr)\,\sd z.
\end{align*}
The first integral is strictly positive
by~\eqref{eq:trajectory-order} and $\tau_1>0$.
The second is also strictly positive: $\tau_1<\tau_2$, and
$f_2(z)<1-1/e$ for every $z\in[\tau_1,\tau_2)$ because $\tau_2$
is the first threshold-crossing time. Consequently, $q_1-q_2>0$,
which proves $q_1>q_2$.
\end{proof}

\section{Fixed-Point Numerical Certification of \texorpdfstring{$\gam^*$}{γ*}}
\label{app:num}

Recall the shooting system~\eqref{eq:num-alpha-a}--\eqref{eq:num-alpha-c}
and its return value $q(\alp)$ in~\eqref{eq:def-q-alpha}. This appendix
separates the exploratory computations used to select endpoint parameters
from their rigorous certification. We first describe a forward-Euler
scheme and its numerical return estimates. Appendix~\ref{app:error} then
validates explicit candidate profiles independently, using exact-arithmetic
checks and the Volterra comparison principle from
Appendix~\ref{app:gam-f}.

\subsection{Exploratory Shooting Computations}
\label{app:dis-scheme}
\label{app:num-q-alp}
\label{app:num-res-1}

\xhdr{Euler scheme.}
Fix $h:=1/N$ for an integer $N\ge1$, set $z_k:=kh$, and initialize
\[
f_0:=1-\alp,\qquad \ell_0:=0,
\qquad \alp\in(1/e,1).
\]
At step $k<N$ with $f_k<1-1/e$, let $\tilde f_h$ be the piecewise-linear
interpolant of $\{(z_j,f_j)\}_{j=0}^k$. The scheme is
\begin{align}
e^{-\ell_k}\sbp{1-\tilde f_h(\ell_k)}
&=1-f_k,
\qquad \ell_k\in[0,z_k],
\label{eq:num-root}\\
f_{k+1}
&=f_k+h\,\frac{e^{-z_k}-(1-f_k)}{1-z_k+\ell_k}.
\label{eq:num-euler}
\end{align}
For $k\ge1$, whenever $0\le f_0\le\cdots\le f_k<1-1/e$, the left-hand
side of~\eqref{eq:num-root} is continuous and strictly decreasing in
$\ell_k$, and its endpoint values bracket $1-f_k$. Thus the root is
unique; in the implementation, it is approximated by one-dimensional
root finding. Equation~\eqref{eq:num-euler} is the forward-Euler update
for~\eqref{eq:num-alpha-b}.

\xhdr{Stopping and return evaluation.}
Let $k^*\in\{1,\ldots,N\}$ be the first index such that
$f_{k^*}\ge1-1/e$ or $k^*=N$. Using the interpolant through the final
computed value, set
\[
\tau_h(\alp):=
\begin{cases}
\displaystyle
z_{k^*-1}
+h\,\frac{(1-1/e)-f_{k^*-1}}{f_{k^*}-f_{k^*-1}},
& f_{k^*}\ge1-1/e,\\[6pt]
1,& f_{k^*}<1-1/e.
\end{cases}
\]
The interpolated terminal value $\tilde f_h(\tau_h(\alp))$ equals
$1-1/e$ when the threshold is reached and equals $f_N$ when the
computation terminates at $1$ without reaching it. Define the
trapezoidal sum and exploratory return value by
\begin{align}
S_h(\alp)
&:=
\sum_{j=0}^{k^*-2}\frac h2
\sbp{e^{-z_j}f_j+e^{-z_{j+1}}f_{j+1}}
+
\frac{\tau_h(\alp)-z_{k^*-1}}2
\sbp{e^{-z_{k^*-1}}f_{k^*-1}
      +e^{-\tau_h(\alp)}\tilde f_h(\tau_h(\alp))},
\label{eqn:trap}\\
Q_h(\alp)
&:=S_h(\alp)+(1-1/e)e^{-\tau_h(\alp)}.
\label{eqn:Q-h-alp}
\end{align}
An empty sum is interpreted as zero. In particular, $Q_h(\alp)$ is
\emph{defined by the finite quadrature formula}; it is not identified
with an exact integral of the interpolated trajectory.

\xhdr{Exploratory values.}
We evaluated $Q_h(\alp)$ for parameters in $[0.5800,0.5810]$ with
$h=1/3000$, using the \textsc{Mathematica}~14.1 notebook
\texttt{Computing\_gamma\_star.nb}.\footnote{The notebook and exploratory
output data are available in the
\href{https://drive.google.com/drive/folders/1ZuRxkxKS-7HyNTcOYBQq8b80ORtnO3Dc?usp=sharing}
{\emph{supplementary repository}}. This implementation is separate from the
exact-arithmetic verifier described in Appendix~\ref{app:error}.}
Table~\ref{tab:numerical-alpha} reports representative values.
The observed sign change of $Q_h(\alp)-\alp$ between
$0.5803$ and $0.5804$ motivates the endpoint choices
\[
a:=0.5802,\qquad b:=0.5805.
\]

\begin{table}[ht]
\centering
\caption{Exploratory shooting values for $h=1/3000$.
The entries are numerical approximations, not certified enclosures of
$q(\alp)$.}
\label{tab:numerical-alpha}
\setlength{\tabcolsep}{4pt}
\resizebox{\textwidth}{!}{%
\begin{tabular}{@{}l*{8}{c}@{}}
\toprule
$\alp$
& 0.5800 & 0.5801 & 0.5802
& \cellcolor{lightgray} 0.5803
& \cellcolor{lightgray} 0.5804
& 0.5805 & 0.5806 & 0.5807 \\
\midrule
$Q_h(\alp)$
& 0.580571 & 0.580505 & 0.580439
& \cellcolor{lightgray} 0.580373
& \cellcolor{lightgray} 0.580307
& 0.580241 & 0.580175 & 0.580108 \\
\bottomrule
\end{tabular}%
}
\end{table}

\xhdr{Transition to validated certification.}
Appendix~\ref{app:error} works instead with explicitly specified
transformed profiles $u_\alp$, for $\alp\in\{a,b\}$, on a finer mesh
of size $1/96000$. Each profile is piecewise linear on $[0,\eta]$,
where $\eta=33/50$, and is continued constantly by $1/e$ on
$[\eta,1]$. The quantity certified there is the exact candidate score
\[
\widehat Q_\alp
:=\mathcal A(u_\alp)
=1-e^{-(\eta+1)}-\int_0^\eta e^{-z}u_\alp(z)\,\sd z,
\]
not the exploratory value $Q_h(\alp)$. The score is enclosed using exact
cell-integral formulas and outward-rounded interval arithmetic.
Independently verified finite tests yield a uniform Volterra residual
bound, which validates the endpoint trajectories and proves
\[
|q(\alp)-\widehat Q_\alp|<10^{-4},
\qquad \alp\in\{a,b\}.
\]
Together with the score enclosures, this establishes $q(a)>a$ and
$q(b)<b$ and hence $\gam^*\in[a,b]$ by the comparison argument in
Appendix~\ref{app:error}. This localization does not rely on the exploratory
quadrature values or the preliminary auxiliary-LP enclosure.

\subsection{Validated Endpoint Certification by a Volterra Residual}
\label{app:error}

We certify the endpoint inequalities using explicit piecewise-linear
profiles and the Volterra comparison principle from
Appendix~\ref{app:gam-f}. The Euler scheme may be used to generate
these profiles, but its accuracy is not an assumption of the certificate.
Instead, we verify finitely many inequalities for the profiles themselves
and bound their residuals in the exact Volterra fixed-point equation.
This avoids numerical assumptions about the derivatives of the exact
trajectory or its stopping time.

Throughout this subsection, use the endpoints $a=0.5802$ and $b=0.5805$
and set
\[
h:=\frac1{96000},\qquad \eta:=\frac{33}{50},\qquad
K:=\frac{\eta}{h}=63360.
\]
Let $t_i:=ih$ for $0\le i\le K$. Recall the functionals
\begin{align*}
\mathcal H_u(t,s)
&:=1-e^{-t}+(1-t+s)u(t)-\int_0^s e^{-z}u(z)\,\sd z,\\
\mathcal A(u)
&:=1-\frac1{e^2}-\int_0^1 e^{-z}u(z)\,\sd z.
\end{align*}
For a profile defined on $[0,\eta]$, the first formula is used for
$0\le s\le t\le\eta$. The second is used after extension to $[0,1]$.
For each $\alp\in(1/e,1)$, Appendix~\ref{app:gam-f} constructs the
unique continuous fixed point $y_\alp$ of
\begin{equation}\label{eq:cert-volterra}
(\mathcal T_\alp u)(t)
:=\max\left\{\frac1e,\ \max_{0\le s\le t}
\frac{\alp-1+e^{-t}+\int_0^s e^{-z}u(z)\,\sd z}{1-t+s}
\right\},\qquad 0\le t<1.
\end{equation}
This construction applies to every $\alp\in(1/e,1)$, whether or not
$\alp$ is feasible for the functional program. On $[0,\eta]$, the
operator satisfies
\begin{equation}\label{eq:cert-volterra-lipschitz}
|\mathcal T_\alp u(t)-\mathcal T_\alp v(t)|
\le\frac1{1-\eta}\int_0^t|u(z)-v(z)|\,\sd z.
\end{equation}
We also use the regularity conclusion from Appendix~\ref{app:gam-f}:
$y_\alp$ is locally Lipschitz and is non-increasing on any interval
where $y_\alp(t)<e^{-t}$. Before reaching $1/e$, its decreasing
branch is $C^1$ and satisfies the flatness identity and the shooting
system.

\subsubsection{Finite hypotheses for an explicit profile}

Let $u$ be the linear interpolant of nodal values $u_i$ at $t_i$ and
extend $u$ constantly by $1/e$ on $[\eta,1]$. Define
\[
m_u(t):=\min_{0\le s\le t}\mathcal H_u(t,s),
\qquad 0\le t\le\eta.
\]
For a specified $\alp\in\{a,b\}$, we verify the following conditions:
\begin{align}
&u_0=\alp,\quad u_K=1/e,\quad
1/e\le u_{i+1}\le u_i\le1,\quad
\frac{u_i-u_{i+1}}h\le\frac12
&& (0\le i<K),\label{eq:cert-shape}\\
&e^{-t_{i+1}}-u_i>\frac7{50}
&& (0\le i<K),\label{eq:cert-gap}\\
&m_u(t_i)\ge\alp-10^{-6}
&& (0\le i\le K),\label{eq:cert-node-lower}\\
&m_u(t_i)\le\alp+10^{-6}
&& (0\le i\le K\text{ with }u_i>1/e),\label{eq:cert-node-upper}\\
&m_u(\eta)>\alp+\frac1{1000}.
&&\label{eq:cert-terminal}
\end{align}
Below we explain how each nodal minimum is rigorously enclosed by a
finite computation. In particular, the minima in these hypotheses do not
require an unvalidated numerical root finder.

\begin{lemma}[Residual bound from finite tests]
\label{lem:cert-residual}
If $u$ satisfies~\eqref{eq:cert-shape}--\eqref{eq:cert-terminal}, then
\[
\|\mathcal T_\alp u-u\|_{\infty,[0,\eta]}
\le\varepsilon,
\qquad
\varepsilon:=\frac{10^{-6}+2h/3}{1-\eta}
=\frac{143}{6120000}.
\]
\end{lemma}

\begin{proof}
By~\eqref{eq:cert-shape}, $u$ is Lipschitz, and its derivative lies in
$[-1/2,0]$ almost everywhere. Condition~\eqref{eq:cert-gap} implies
$e^{-t}-u(t)>7/50$ throughout $[0,\eta]$: for
$t\in[t_i,t_{i+1}]$,
\[
e^{-t}-u(t)\ge e^{-t_{i+1}}-u_i>7/50.
\]
Put $p(s):=e^{-s}u(s)$. This function is strictly decreasing and
absolutely continuous, with
\[
-\frac{\sd p(s)}{\sd s}
=e^{-s}\left(u(s)-\frac{\sd u(s)}{\sd s}\right)
\ge e^{-(\eta+1)}>0
\quad\text{almost everywhere}.
\]
Its inverse is Lipschitz on its image. For each $t$, the minimizing
load $\ell_u(t)$ is the unique solution in $[0,t]$ of
$p(\ell_u(t))=u(t)$: indeed,
$\partial_s\mathcal H_u(t,s)=u(t)-p(s)$ and
$p(t)\le u(t)\le p(0)$. Thus $\ell_u=p^{-1}\circ u$ is Lipschitz.
Differentiating the envelope almost everywhere and using the root
identity to cancel the terms involving $\sd\ell_u/\sd t$ gives
\[
\frac{\sd m_u(t)}{\sd t}
=e^{-t}-u(t)+(1-t+\ell_u(t))\frac{\sd u(t)}{\sd t}.
\]
Since $0\le\ell_u(t)\le t$, $u\ge1/e$, and $e^{-t}-u(t)>0$,
\[
-\frac12\le\frac{\sd m_u(t)}{\sd t}
\le1-\frac1e<\frac23
\quad\text{almost everywhere}.
\]
The envelope is therefore $2/3$-Lipschitz.

For $t\in[t_i,t_{i+1}]$, the lower nodal test yields
\[
m_u(t)\ge\alp-\left(10^{-6}+\frac{2h}{3}\right).
\]
If $u(t)>1/e$, monotonicity gives $u_i>1/e$, so the upper nodal
test also yields
\[
m_u(t)\le\alp+\left(10^{-6}+\frac{2h}{3}\right).
\]
At $t=\eta$, the lower bound follows directly from its nodal test.
From the operator definition,
\[
\mathcal T_\alp u(t)-u(t)
=\max\left\{\frac1e-u(t),\ \max_{0\le s\le t}
\frac{\alp-\mathcal H_u(t,s)}{1-t+s}\right\}.
\]
The lower bound on $m_u$ and $1-t+s\ge1-\eta$ give
$\mathcal T_\alp u-u\le\varepsilon$. When $u(t)>1/e$, choosing a
minimizing load gives the reverse bound
$\mathcal T_\alp u-u\ge-\varepsilon$. When $u(t)=1/e$, this reverse
bound follows from $\mathcal T_\alp u(t)\ge1/e=u(t)$. Together,
these estimates prove the claim.
\end{proof}

\begin{lemma}[Exact-profile and return-value enclosures]
\label{lem:cert-return}
Suppose $u$ satisfies~\eqref{eq:cert-shape}--\eqref{eq:cert-terminal}.
Then the exact fixed point $y_\alp$ reaches $1/e$ at a time
$\tau_\alp\in(0,\eta)$ while remaining strictly below $e^{-t}$.
Its constant extension $g_\alp$ belongs to $\mathcal G$ and satisfies
all partial-load constraints at level $\alp$. Moreover,
\begin{align}
\|y_\alp-u\|_{\infty,[0,\eta]}&<\frac1{5000},
\label{eq:cert-profile-error}\\
|q(\alp)-\widehat Q_\alp|&<\frac1{10000},
\label{eq:cert-return-error}
\end{align}
where
\begin{align}
\widehat Q_\alp
&:=\mathcal A(u)
=1-e^{-(\eta+1)}-\int_0^\eta e^{-z}u(z)\,\sd z,
\label{eq:cert-candidate-score}\\
q(\alp)
&:=\mathcal A(g_\alp)
=\int_0^{\tau_\alp}e^{-z}(1-y_\alp(z))\,\sd z
 +(1-1/e)e^{-\tau_\alp}.
\label{eq:cert-exact-return}
\end{align}
No feasibility assumption on $\alp$ is required.
\end{lemma}

\begin{proof}
Let $d(t):=|y_\alp(t)-u(t)|$. The fixed-point equation,
\eqref{eq:cert-volterra-lipschitz}, and
Lemma~\ref{lem:cert-residual} give
\[
d(t)\le\varepsilon+\frac1{1-\eta}\int_0^t d(z)\,\sd z.
\]
Gronwall's inequality implies
\[
\|y_\alp-u\|_{\infty,[0,\eta]}
\le\varepsilon e^{\eta/(1-\eta)}
<7\varepsilon
=\frac{1001}{6120000}
<\frac1{5000}.
\]
Here $\eta/(1-\eta)=33/17$ and $e^{33/17}<7$; the latter inequality
is also checked by the exact-arithmetic verifier below.

The gap condition and this error estimate imply
$y_\alp(t)<e^{-t}$ on $[0,\eta]$. The fixed-point regularity result
from Appendix~\ref{app:gam-f} therefore makes $y_\alp$ non-increasing
there. We next prove that it reaches $1/e$ strictly before $\eta$.
For every $s\in[0,\eta]$, the terminal margin and $u(\eta)=1/e$ imply
\begin{align*}
&\alp-1+e^{-\eta}+\int_0^s e^{-z}y_\alp(z)\,\sd z<\frac{1-\eta+s}{e}-\frac1{1000}
       +\frac{1-e^{-\eta}}{5000}<\frac{1-\eta+s}{e}.
\end{align*}
In the last inequality we used $1-e^{-\eta}<1/2$. Thus the inner
maximum in~\eqref{eq:cert-volterra} is strictly below $1/e$ at
$t=\eta$. It is continuous as a function of $t$ on $[0,\eta]$, so
it is also below $1/e$ throughout some left neighborhood of $\eta$.
The fixed-point equation then gives $y_\alp=1/e$ in that neighborhood.
Since $y_\alp(0)=\alp>1/e$, the first hitting time belongs to
$(0,\eta)$.

Before $\tau_\alp$, the inequalities $1/e<y_\alp(t)<e^{-t}$ identify
the decreasing flat branch from Appendix~\ref{app:gam-f}. Extend
this branch constantly by $1/e$. The resulting $g_\alp$ is continuous,
non-increasing, and lies in $[1/e,1]$. Before the threshold, the
fixed-point equation implies all the partial constraints and makes their
minimum equal to $\alp$; the same statements hold at the threshold
by continuity. After the threshold, the constraints follow from
\[
\partial_s\mathcal H_{g_\alp}(t,s)
=\frac{1-e^{-s}}e\ge0
\qquad(\tau_\alp\le s\le t\le1)
\]
and
\[
\mathcal H_{g_\alp}(t,s)-\mathcal H_{g_\alp}(\tau_\alp,s)
=\int_{\tau_\alp}^t(e^{-z}-1/e)\,\sd z\ge0
\qquad(s\le\tau_\alp\le t\le1).
\]
On $[0,\eta]$, this extension equals $y_\alp$ by monotonicity and the
lower bound $y_\alp\ge1/e$. It is also a fixed point on every
$[0,\eta_0]$ with $0<\eta_0<1$, and hence coincides with the global
$y_\alp$ there by uniqueness.

Both $u$ and $g_\alp$ equal $1/e$ on $[\eta,1]$. Consequently,
\[
|\mathcal A(g_\alp)-\mathcal A(u)|
\le\int_0^\eta e^{-z}|y_\alp(z)-u(z)|\,\sd z
<\frac{1-e^{-\eta}}{5000}<\frac1{10000}.
\]
The equality between $\mathcal A(g_\alp)$ and the original
stopping-time return value follows by integrating the constant tail.
This proves~\eqref{eq:cert-exact-return} and the error estimate.
\end{proof}

\subsubsection{Exact verification of the finite hypotheses}

All numerical data and supporting files for this certification, including
the candidate profiles, candidate-generation code, exact-arithmetic verifier,
and verification report, are available in the
\href{https://drive.google.com/drive/folders/1ZuRxkxKS-7HyNTcOYBQq8b80ORtnO3Dc?usp=sharing}
{\emph{supplementary repository}}.

The supplementary files \texttt{profile\_a.txt} and
\texttt{profile\_b.txt} specify the two profiles. Each contains $K+1$
nodal entries. An integer entry $n_i$ means the exact value
$u_i=n_i/10^{15}$; the literal token \texttt{r} means the exact constant
$1/e$. This token is part of the data-file format, not an additional
mathematical parameter. In particular, the constant tail is not replaced
by a rounded decimal. Linear interpolation specifies the profile
completely. The candidate generator uses ordinary floating-point Euler
steps, but its computations are not trusted by the verifier.

\paragraph{Enclosing the minimum at a grid point.}
The function $p(s)=e^{-s}u(s)$ is strictly decreasing. The verifier
brackets the root $p(s)=u_i$ between two grid points $t_j,t_k$, with
$k-j\le2$, by proving the interval inequalities
\[
p(t_j)\ge u_i\ge p(t_k).
\]
For $i=0$, the minimizing load is exactly zero. Otherwise this bracket
contains the unique minimizer $\ell_u(t_i)$. The slope bounds yield
\[
\left|\frac{\sd p(s)}{\sd s}\right|
=e^{-s}\left(u(s)-\frac{\sd u(s)}{\sd s}\right)
\le\frac32
\quad\text{almost everywhere}.
\]
Therefore
\begin{align*}
0&\le\mathcal H_u(t_i,t_j)-m_u(t_i)=\int_{t_j}^{\ell_u(t_i)}
   \bigl(p(s)-p(\ell_u(t_i))\bigr)\,\sd s\le\frac34\bigl(\ell_u(t_i)-t_j\bigr)^2
\le3h^2.
\end{align*}
An outward interval $[L_i,U_i]$ for $\mathcal H_u(t_i,t_j)$ thus
provides the rigorous enclosure
\[
m_u(t_i)\in[L_i-3h^2,U_i].
\]
This suffices for every nodal test and the terminal margin.

\paragraph{Integrals and arithmetic.}
For the cell slope $v_i:=(u_{i+1}-u_i)/h$, the integral of the
explicit profile is evaluated using the exact formula
\[
\int_{t_i}^{t_{i+1}}e^{-z}u(z)\,\sd z
=u_i\bigl(e^{-t_i}-e^{-t_{i+1}}\bigr)
 +v_i\bigl(e^{-t_i}-(1+h)e^{-t_{i+1}}\bigr).
\]
There is no quadrature remainder in this formula. All rounding errors
in its evaluation are included in outward interval arithmetic.

The verifier \texttt{certify.py} uses only exact Python integers and
standard-library rational arithmetic. An interval is stored as a pair
of integers specifying endpoints on the lattice $10^{-50}\mathbb Z$;
every arithmetic operation is rounded outward by integer floor and
ceiling operations. Exponential enclosures are obtained from exact
rational Taylor polynomials:
\[
\sum_{j=0}^{81}\frac{(-x)^j}{j!}
\le e^{-x}\le
\sum_{j=0}^{80}\frac{(-x)^j}{j!},\qquad x\ge0,
\]
using the sign of the Taylor remainder. The mesh exponential $e^{-h}$
is enclosed once, and the intervals for $e^{-t_i}$ are obtained by
outward interval multiplication. The same procedure encloses $1/e$
and checks $1/e>1/3$, $e^{-\eta}>1/2$, and $e^{33/17}<7$.

The verifier checks all the shape, cell-gap, and nodal inequalities,
not merely the summary statistics. Interval comparisons that cannot
establish a root bracket are rejected rather than treated as numerical
signs. Both supplied profiles pass. Table~\ref{tab:cert-finite-results}
records outward-rounded bounds obtained by the verifier. The reported
nodal error is the maximum of the lower-test violation
$(\alp-m_u(t_i))_+$ at all nodes and the upper-test violation
$(m_u(t_i)-\alp)_+$ at nodes with $u_i>1/e$.

\begin{table}[ht]
\centering
\small
\begin{tabular}{@{}lcc@{}}
\toprule
Verified quantity & $\alp=a$ & $\alp=b$\\
\midrule
Maximum negative cell slope & $<0.420$ & $<0.420$\\
Minimum cell gap $e^{-t_{i+1}}-u_i$ & $>0.1489$ & $>0.1489$\\
Nodal envelope error & $<7.25\times10^{-7}$ & $<7.28\times10^{-7}$\\
Terminal margin $m_u(\eta)-\alp$ & $>0.0020143$ & $>0.0016343$\\
$\widehat Q_\alp=\mathcal A(u)$
 & $[0.5804315,0.5804316]$ & $[0.5802332,0.5802333]$\\
\bottomrule
\end{tabular}
\caption{Finite, outward-rounded certificate bounds. These concern the
explicit piecewise-linear profiles, not sampled derivatives of an
unknown exact solution. The verifier uses a nodal tolerance of $10^{-6}$
and a terminal margin of $10^{-3}$.}
\label{tab:cert-finite-results}
\end{table}

\subsubsection{Endpoint signs and localization of the optimum}

\begin{corollary}[Validated endpoint enclosure]
\label{cor:cert-endpoints}
The exact return values satisfy
\[
q(a)\in[0.5803315,0.5805316],\qquad
q(b)\in[0.5801332,0.5803333].
\]
In particular, $q(a)>a$ and $q(b)<b$, and the optimal value
satisfies
\[
\gam^*\in[0.5802,0.5805].
\]
\end{corollary}

\begin{proof}
The finite checks and Lemma~\ref{lem:cert-return} justify enlarging each
score interval in Table~\ref{tab:cert-finite-results} by $10^{-4}$.
This gives the stated enclosures, with
\[
q(a)-a\ge0.0001315>0,\qquad q(b)-b\le-0.0001667<0.
\]
For completeness, the localization can be deduced without assuming a
preliminary numerical enclosure of $\gam^*$ or existence of shooting
trajectories outside the certified interval. The profile $g_a$ satisfies
all partial constraints at level $a$, and
$\mathcal A(g_a)=q(a)>a$. It is therefore feasible, proving
$\gam^*\ge a$.

Suppose $\gam^*\ge b$, and let $g^*$ be an optimal normalized profile,
whose existence was established in Appendix~\ref{app:gam-f}. It is
feasible at level $b$, so $g^*\ge\mathcal T_b g^*$. Volterra comparison
gives $g^*\ge y_b=g_b$ on $[0,\eta]$, and on $[\eta,1]$ the same
inequality follows from $g^*\ge1/e=g_b$. Hence
\[
\gam^*\le\mathcal A(g^*)\le\mathcal A(g_b)=q(b)<b,
\]
a contradiction. Thus $a\le\gam^*<b$, and the claimed closed enclosure
follows. The exact characterization identifies this value
with $\Gam(\cF_0)$.
\end{proof}

\xhdr{Relation to the shooting formulation.}
For every $\alp\in[a,b]$, the profile $y_b$ is a supersolution of
$\mathcal T_\alp$, so
\[
1/e\le y_\alp(t)\le y_b(t)<e^{-t},
\qquad 0\le t\le\eta.
\]
It follows that $y_\alp$ is non-increasing and reaches $1/e$ no later
than $y_b$. Thus all the shooting trajectories on $[a,b]$ are
well-defined. More generally, for $\alp_1<\alp_2$ in this interval,
comparison gives $y_{\alp_1}\le y_{\alp_2}$, and their distinct
initial values make the weighted integral difference strictly positive.
Consequently, $q(\alp_1)>q(\alp_2)$. Parameter continuity from
Appendix~\ref{app:gam-f} and the common constant tail after $\eta$
make $q$ continuous on $[a,b]$. The endpoint inequalities therefore
identify its unique fixed point on this interval. In particular, this
interpretation does not require an unverified well-posedness assertion
over a larger numerically chosen search interval.


\xhdr{AI methodology.}
After exploring a variational approach based on KKT conditions,
the author used ChatGPT (OpenAI) to help develop the proof of
Theorem~\ref{thm:main-structural} in Appendix~\ref{app:gam-f}
using a comparison principle for a Volterra integral equation.
The author also used ChatGPT to help develop the Euler shooting procedure
for candidate generation in Subsection~\ref{sec:numerical-gamma-star}
and the residual-based argument for validating the endpoints in
Subsection~\ref{sec:validated-gamma-star}. For the auxiliary-LP
analysis, ChatGPT assisted with the continuum-consistency argument
in Proposition~\ref{pro:2-18-a} and the exact-arithmetic verifier
for the finite LP bounds in Appendix~\ref{sec:ub}.
The author takes full responsibility for the proofs, code, and reported results.

\newpage
\bibliographystyle{alpha} 
\bibliography{EC_21}

@article{xu2026certifying,
  title={Certifying Capacity and Gaps of the Rank-Dependent Gain-Sharing Framework for Ranking},
  author={Xu, Pan},
  journal={arXiv preprint arXiv:2610.17545},
  year={2026}
}

@book{teschl2012ordinary,
  author    = {Gerald Teschl},
  title     = {Ordinary Differential Equations and Dynamical Systems},
  series    = {Graduate Studies in Mathematics},
  volume    = {140},
  publisher = {American Mathematical Society},
  year      = {2012},
  doi       = {10.1090/gsm/140}
}

@book{gripenberg1990volterra,
  author    = {G. Gripenberg and S.-O. Londen and O. Staffans},
  title     = {Volterra Integral and Functional Equations},
  series    = {Encyclopedia of Mathematics and its Applications},
  volume    = {34},
  publisher = {Cambridge University Press},
  year      = {1990},
  doi       = {10.1017/CBO9780511662805}
}

@article{bessenyei2021applications,
  author  = {Mih{\'a}ly Bessenyei and Zsolt P{\'a}les},
  title   = {Applications of the {Bielecki} Renorming Technique},
  journal = {Journal of Fixed Point Theory and Applications},
  volume  = {23},
  year    = {2021},
  note    = {Article 15},
  doi     = {10.1007/s11784-020-00845-x}
}

@article{fernandes2015systematic,
  author  = {Cristina G. Fernandes and Luis A. A. Meira and
             Fl{\'a}vio K. Miyazawa and Lehilton L. C. Pedrosa},
  title   = {A Systematic Approach to Bound Factor-Revealing {LPs} and
             Its Application to the Metric and Squared Metric
             Facility Location Problems},
  journal = {Mathematical Programming},
  volume  = {153},
  year    = {2015},
  pages   = {655--685},
  doi     = {10.1007/s10107-014-0821-x}
}

@inproceedings{derakhshan2025,
  title={Improved Approximation for Ranking on General Graphs},
  author={Derakhshan, Mahsa and Roghani, Mohammad and Saneian, Mohammad and Yu, Tao},
  booktitle={Proceedings of the 2026 Annual ACM-SIAM Symposium on Discrete Algorithms (SODA)},
  pages={52--84},
  year={2026},
  organization={SIAM}
}

@InProceedings{bhangale2025optimal,
  author =	{Bhangale, Amey and Chakraborty, Arghya and Harsha, Prahladh},
  title =	{{Optimal Online Bipartite Matching in Degree-2 Graphs}},
  booktitle =	{36th International Symposium on Algorithms and Computation (ISAAC 2025)},
  pages =	{13:1--13:19},
  series =	{Leibniz International Proceedings in Informatics (LIPIcs)},
  ISBN =	{978-3-95977-408-6},
  ISSN =	{1868-8969},
  year =	{2025},
  volume =	{359},
  editor =	{Chen, Ho-Lin and Hon, Wing-Kai and Tsai, Meng-Tsung},
  publisher =	{Schloss Dagstuhl -- Leibniz-Zentrum f{\"u}r Informatik},
  address =	{Dagstuhl, Germany},
  URL =		{https://drops.dagstuhl.de/entities/document/10.4230/LIPIcs.ISAAC.2025.13},
  URN =		{urn:nbn:de:0030-drops-249216},
  doi =		{10.4230/LIPIcs.ISAAC.2025.13}
}

@article{huang2024online,
  title={Online matching: A brief survey},
  author={Huang, Zhiyi and Tang, Zhihao Gavin and Wajc, David},
  journal={ACM SIGecom Exchanges},
  volume={22},
  number={1},
  pages={135--158},
  year={2024},
  publisher={ACM New York, NY, USA}
}

@article{huang2020online,
  title={Online Primal Dual Meets Online Matching with Stochastic Rewards: Configuration LP to the Rescue},
  author={Huang, Zhiyi and Zhang, Qiankun},
  journal={SIAM Journal on Computing},
  volume={53},
  number={5},
  pages={1217--1256},
  year={2024},
  publisher={SIAM}
}

@inproceedings{albers2025online,
  title={Online b-Matching with Stochastic Rewards},
  author={Albers, Susanne and Schubert, Sebastian},
  booktitle={International Conference on Current Trends in Theory and Practice of Computer Science},
  pages={37--50},
  year={2025},
  organization={Springer}
}

@article{goyal2023online,
  title={Online matching with stochastic rewards: Optimal competitive ratio via path-based formulation},
  author={Goyal, Vineet and Udwani, Rajan},
  journal={Operations Research},
  volume={71},
  number={2},
  pages={563--580},
  year={2023},
  publisher={INFORMS}
}

@inproceedings{peng2025revisiting,
  title={Revisiting ranking for online bipartite matching with random arrivals: the primal-dual analysis},
  author={Peng, Bo and Tang, Zhihao Gavin},
  booktitle={Proceedings of the 26th ACM Conference on Economics and Computation},
  pages={820--835},
  year={2025}
}

@inproceedings{jin2021improved,
  title={Improved analysis of RANKING for online vertex-weighted bipartite matching in the random order model},
  author={Jin, Billy and Williamson, David P},
  booktitle={International Conference on Web and Internet Economics},
  pages={207--225},
  year={2021},
  organization={Springer}
}

@inproceedings{huang2020fully,
  title={Fully online matching ii: Beating ranking and water-filling},
  author={Huang, Zhiyi and Tang, Zhihao Gavin and Wu, Xiaowei and Zhang, Yuhao},
  booktitle={2020 IEEE 61st Annual Symposium on Foundations of Computer Science (FOCS)},
  pages={1380--1391},
  year={2020},
  organization={IEEE}
}

@article{buchbinder2009design,
  title={The design of competitive online algorithms via a primal--dual approach},
  author={Buchbinder, Niv and Naor, Joseph Seffi},
  journal={Foundations and Trends{\textregistered} in Theoretical Computer Science},
  volume={3},
  number={2--3},
  pages={93--263},
  year={2009},
  publisher={Now Publishers, Inc.}
}

@inproceedings{buchbinder2007online,
  title={Online primal-dual algorithms for maximizing ad-auctions revenue},
  author={Buchbinder, Niv and Jain, Kamal and Naor, Joseph},
  booktitle={European Symposium on Algorithms},
  pages={253--264},
  year={2007},
  organization={Springer}
}

@inproceedings{devanur2013randomized,
  title={Randomized primal-dual analysis of ranking for online bipartite matching},
  author={Devanur, Nikhil R and Jain, Kamal and Kleinberg, Robert D},
  booktitle={Proceedings of the twenty-fourth annual ACM-SIAM symposium on Discrete algorithms},
  pages={101--107},
  year={2013},
  organization={SIAM}
}

@inproceedings{mehta2013online,
  title={Online matching with stochastic rewards},
  author={Mehta, Aranyak and Panigrahi, Debmalya},
  booktitle={2012 IEEE 53rd annual symposium on foundations of computer science},
  pages={728--737},
  year={2012},
  organization={IEEE}
}

@inproceedings{huang2023online,
  title={Online matching with stochastic rewards: Advanced analyses using configuration linear programs},
  author={Huang, Zhiyi and Jiang, Hanrui and Shen, Aocheng and Song, Junkai and Wu, Zhiang and Zhang, Qiankun},
  booktitle={International Conference on Web and Internet Economics},
  pages={384--401},
  year={2023},
  organization={Springer}
}

@article{mehta2007adwords,
  title={Adwords and generalized online matching},
  author={Mehta, Aranyak and Saberi, Amin and Vazirani, Umesh and Vazirani, Vijay},
  journal={Journal of the ACM (JACM)},
  volume={54},
  number={5},
  pages={22--es},
  year={2007},
  publisher={ACM New York, NY, USA}
}

@article{huang2019online,
  title={Online vertex-weighted bipartite matching: Beating $1-1/e$ with random arrivals},
  author={Huang, Zhiyi and Tang, Zhihao Gavin and Wu, Xiaowei and Zhang, Yuhao},
  journal={ACM Transactions on Algorithms (TALG)},
  volume={15},
  number={3},
  pages={1--15},
  year={2019},
  publisher={ACM New York, NY, USA}
}

@inproceedings{mahdian2011online,
  title={Online bipartite matching with random arrivals: an approach based on strongly factor-revealing lps},
  author={Mahdian, Mohammad and Yan, Qiqi},
  booktitle={Proceedings of the forty-third annual ACM symposium on Theory of computing},
  pages={597--606},
  year={2011},
  organization={ACM}
}

@article{mehta2012online,
  author    = {Aranyak Mehta},
  title     = {Online Matching and Ad Allocation},
  journal   = {Foundations and Trends in Theoretical Computer Science},
  volume    = {8},
  number    = {4},
  pages     = {265--368},
  year      = {2013},
}
\end{document}